\documentclass[letterpaper, 10 pt, conference]{ieeeconf}%
\usepackage{graphicx}
\usepackage{epsfig}
\usepackage{mathptmx}
\usepackage{times}
\usepackage{amsmath}
\usepackage{thmtools,thm-restate}
\usepackage{amssymb}
\usepackage[section]{placeins}
\usepackage{placeins}
\usepackage{amsmath}
\usepackage{multirow}
\usepackage{graphicx}
\usepackage[normalem]{ulem}
\usepackage{hyperref}
\usepackage{subcaption}
\usepackage{algorithm,tabularx}
\usepackage{algpseudocode}
\usepackage{amsfonts}
\usepackage{cases}
\usepackage{romannum}
\usepackage{textcomp}%
\providecommand{\U}[1]{\protect\rule{.1in}{.1in}}

\IEEEoverridecommandlockouts
\newcommand{\R}{\mathbb{R}}

\newcommand{\mb}[1]{\mathbf{#1}}
\newtheorem{assumption}{Assumption}
\newtheorem{theorem}{Theorem}
\newtheorem{corollary}{Corollary}
\newtheorem{lemma}{Lemma}
\newtheorem{remark}{Remark}
\useunder{\uline}{\ul}{}
\makeatletter
\newcommand{\multiline}[1]{  \begin{tabularx}{\dimexpr\linewidth-\ALG@thistlm}[t]{@{}X@{}}
#1
\end{tabularx}
}
\makeatother
 \usepackage{xcolor}

\usepackage[noadjust]{cite}%
\begin{document}

\title{{\LARGE \textbf{Event-Driven Receding Horizon Control For Distributed Persistent Monitoring in Network Systems}}}
\author{Shirantha Welikala and Christos G. Cassandras \thanks{$^{\star}$Supported in part by NSF under grants ECCS-1931600, DMS-1664644, CNS-1645681, by AFOSR under grant FA9550-19-1-0158, by ARPA-E's NEXTCAR program under grant DE-AR0000796 and by the MathWorks.} \thanks{The authors are with the Division of Systems Engineering and Center for Information and Systems Engineering, Boston University, Brookline, MA 02446, \texttt{{\small \{shiran27,cgc\}@bu.edu}}.}}
\maketitle

\begin{abstract}
We address the multi-agent persistent monitoring problem defined on a set of nodes (targets) interconnected over a network topology. A measure of mean overall node state uncertainty evaluated over a finite period is to be minimized by controlling the motion of a cooperating team of agents. To address this problem, we propose an event-driven receding horizon control approach that is computationally efficient, distributed and on-line. The proposed controller differs from the existing on-line gradient-based parametric controllers and off-line greedy cycle search methods that often lead to either low-performing local optima or computationally intensive centralized solutions. A critical novel element in this controller is that it automatically optimizes its planning horizon length, thus making it parameter-free. We show that explicit globally optimal solutions can be obtained for every distributed optimization problem encountered at each event where the receding horizon controller is invoked. Numerical results are provided showing improvements compared to state of the art distributed on-line parametric control solutions.
\end{abstract}

\thispagestyle{empty} \pagestyle{empty}

\section{Introduction}

A \emph{persistent monitoring} problem arises when a dynamically changing environment is to be monitored by a set of mobile agents. In contrast to coverage problems, where agents equally value every point in the environment \cite{Lin2013,Huynh2010}, we focus on monitoring only a finite set of \textquotedblleft points of interest\textquotedblright\ (henceforth called \textquotedblleft targets\textquotedblright) which hold a value \cite{Hari2019,Zhou2019,Yu2016,Welikala2019P3,Rezazadeh2019,Zhou2018,Lan2013,Khazaeni2018,Song2014} and which the agent team senses (or collects information from) in order to reduce an \textquotedblleft uncertainty metric\textquotedblright\ associated with the target state. Typically, a target's uncertainty metric decreases when an agent can monitor the target by dwelling in its vicinity and increases when no agent is monitoring it. The global objective is to control each agent's motion so as to collectively minimize an overall measure of target uncertainties evaluated over a fixed period of interest. 

Such persistent monitoring problems encompass applications that include surveillance \cite{Leahy2016}, environmental sensing \cite{Trevathan2018}, event detection \cite{Yu2015}, data collecting \cite{khazaeni2018b,Smith2011} and energy management \cite{Maini2018,Mathew2015}. Moreover, different persistent monitoring problem settings have been considered in the literature varying in the specific global objective to be optimized, including event counts \cite{Yu2015}, idle times \cite{Hari2019}, error covariances \cite{Pinto2020,Lan2013,Welikala2020Ax6} or visibility states \cite{Maini2018}, as well as the nature of the target state dynamics which may be deterministic \cite{Yu2016,Zhou2019,Song2014} or probabilistic \cite{Rezazadeh2019,Pinto2020,Lan2013}. 

Classical optimal control techniques are exploited in \cite{Zhou2018} to solve persistent monitoring problems in 1D environments where the optimal solutions minimizing an average target uncertainty metric have been proven to be threshold-based parametric controllers. However, as shown in \cite{Lin2013}, this synergy between optimal control and parametric controllers does not extend to 2D environments. Instead, for such problems, agent trajectories may be restricted to specific parametric families (e.g., elliptical or Fourier) \cite{Lin2013,Khazaeni2018} and optimal solutions can still be determined within these families. Aside from the generally sub-optimal trajectories obtained, this approach is also limited by its inability to react to dynamic changes in target uncertainty states and the dependence of performance on the initial target/agent conditions which leads to local optima. An alternative approach is to exploit the network structure of the system consisting of targets and agents which is modeled as a graph, where targets are associated with nodes and inter-target agent trajectory segments are associated with edges, to formulate \emph{Persistent Monitoring on Networks} (PMN) problems.

In PMN problems \cite{Yu2020,Zhou2019,Rezazadeh2019}, each agent trajectory is defined by the sequence of \emph{visited targets} and the sequence of \emph{dwell times} spent at each visited target. Therefore, searching for the optimal set of agent decision sequences is a computationally intensive process. To overcome this issue, \cite{Rezazadeh2019} exploits the submodularity property of the objective function and proposes a sub-optimal greedy solution with a performance bound guarantee (see also \cite{Sun2020}). However, this approach and many others \cite{Yu2015,Maini2018,Hari2019,Song2014} are limited to centralized settings. The work in \cite{Zhou2019} overcomes this challenge by adopting a \emph{distributed} Threshold-based Control Policy (TCP) where each agent enforces a set of thresholds on its neighboring target uncertainty values to decide immediate trajectory decisions in a distributed manner: the dwell-time to be spent at the current target and the next-target to visit. A gradient-based technique using Infinitesimal Perturbation Analysis (IPA) \cite{Cassandras2010b} is then used to optimize on line the threshold values. However, due to the use of gradient techniques, this TCP approach to PMN problems often converges to poor locally optimal solutions. This issue is addressed in \cite{Welikala2019P3} by appending an \emph{off-line} \emph{centralized} threshold initialization scheme which has been shown to considerably increase performance at the expense of significant computational effort for the initialization process. However, since the on-line component of this solution is still governed by the TCP method \cite{Zhou2019}, any state (or system parameter) perturbation would trigger a new threshold-tuning process with a considerable amount of recovery time.

Aiming to address the challenges mentioned above, this paper (which is an extended version of \cite{Welikala2020P5}) departs from gradient-based approaches and follows an entirely different direction for PMN problems. Specifically, the event-driven nature of PMN systems is exploited to derive an Event-Driven \emph{Receding Horizon Controller} (RHC) to optimally control each of the agents in an \emph{on-line} \emph{distributed} manner using only a minimal amount of computational power. The conventional use of a RHC involves selecting a \emph{planning horizon} over which an optimization problem is solved (e.g., \cite{Li2006,Wang2017,khazaeni2018b,Chen2019}). A novelty in the proposed RHC approach is the ability to simultaneously determine the optimal value of the planning horizon to be used locally at each agent, making it not only gradient-free but also a parameter-free controller.

As a first step, we show that each agent's trajectory is fully characterized by the sequence of decisions it makes at specific discrete event times. Then, considering an agent at any one of these event times, we formulate a \emph{Receding Horizon Control Problem} (RHCP) to determine the immediate optimal decisions over an optimally determined planning horizon; these decisions are subsequently executed only over a shorter \emph{action horizon} with the process sequentially repeated as new events take place. Next, we exploit several structural properties of this RHCP which takes the form of a non-convex constrained optimization problem, to derive a unique global optimal solution for it in closed form. By introducing some modifications to the RHC architecture, we also show how to obtain higher-performing solutions. Finally, we investigate the performance improvement achieved (compared to the IPA-TCP method in \cite{Zhou2019}) and the controller's ability to take into account the presence of random effects affecting the system behavior.

The paper is organized as follows. Section \ref{Sec:ProblemFormulation} presents the problem formulation, some preliminary results and an overview of the RHC solution. In Section \ref{Sec:SolvingRHCPs}, we show how the RHCP is explicitly solved and modifications to it are presented in Section \ref{Sec:Improvements}. The performance and robustness of the proposed RHC method are illustrated through simulation results in Section \ref{Sec:SimulationResults}. Finally, Section \ref{Sec:Conclusion} concludes the paper.

\section{Problem Formulation} \label{Sec:ProblemFormulation}

\subsection{Persistent Monitoring On Networks (PMN) Problem}

We consider an $n$-dimensional mission space containing $M$ targets (nodes) in the set $\mathcal{T}=\{1,2,\ldots,M\}$ where the location of target $i\in\mathcal{T}$ is fixed at $Y_{i}\in\mathbb{R}^{n}$. A team of $N$ agents in the set $\mathcal{A}=\{1,2,\ldots,N\}$ is deployed to monitor the targets. Each agent $a\in\mathcal{A}$ moves within this mission space and its location at time $t$ is denoted by $s_{a}(t)\in\mathbb{R}^{n}$.

\paragraph{\textbf{Target Model}}

Each target $i\in\mathcal{T}$ has an associated \emph{uncertainty state} $R_{i}(t)\in\mathbb{R}$ which follows the dynamics \cite{Zhou2019}:
\begin{equation}
\dot{R}_{i}(t) =
\begin{cases}
A_{i}-B_{i}N_{i}(t) & \mbox{ if }R_{i}(t)>0\mbox{ or }A_{i}-B_{i}N_{i}(t)>0\\
0 & \mbox{ otherwise,}
\end{cases}
\label{Eq:TargetDynamics}%
\end{equation}
where $A_{i},B_{i}$ and $R_{i}(0)$ values are prespecified and $N_{i}(t)=\sum_{a\in\mathcal{A}}\mathbf{1}\{s_{a}(t)=Y_{i}\}$ ($\mathbf{1}\{\cdot\}$ is the indicator function). Therefore, $N_{i}(t)$ represents the number of agents present at target $i$ at time $t$. Following from \eqref{Eq:TargetDynamics}: (\romannum{1}) $R_{i}(t)$ increases at a rate $A_{i}$ when no agent is visiting target $i$,  (\romannum{2}) $R_{i}(t)$ decreases at a rate $B_{i}N_{i}(t)-A_{i}$ where $B_{i}$ is the uncertainty removal rate by a visiting agent (i.e., agent sensing or data collection rate) to the target $i$, and, (\romannum{3}) $R_{i}(t)\geq0,\ \forall t$.

This problem setup (same as in \cite{Zhou2019,Welikala2019P3}) has an attractive queueing system interpretation \cite{Zhou2019} where $A_{i}$ and $B_{i}N_{i}(t)$ are respectively thought of as the arrival rate and the controllable service rate at target (server) $i\in\mathcal{T}$ in a queueing network. Another application example of this problem setup is as follows. Each target $i$ can be thought of as an ecological station that monitors its surrounding environment while collecting data at a rate of $A_i$. The agent team is tasked to collect the data from the targets. Each agent is assumed to be capable of downloading data rate $B_i$ upon visiting a target $i$. In this paradigm, the target state $R_i$ represents the amount of data locally stored. Alternatively, $R_i$ can also be seen as the time since the last visit by an agent.



\paragraph{\textbf{Agent Model}}

Some persistent monitoring models (e.g., \cite{Zhou2018,Pinto2020}) assume each agent $a\in\mathcal{A}$ to have a finite sensing range $r_{a}>0$ so that it can decrease $R_{i}(t)$ whenever it is in the vicinity of target $i\in\mathcal{T}$ (i.e., whenever $\Vert s_{a}(t)-Y_{i}\Vert\leq r_{a}$). In our graph-based topology, the condition $\Vert s_{a}(t)-Y_{i}\Vert\leq r_{a}$ is represented by agent $a$ residing at the $i$\textsuperscript{th} vertex of a graph (i.e., $\mathbf{1}\{s_{a}(t)=Y_{i}\}$) and $N_{i}(t)$ is used to replace the role of the joint detection probability \cite{Zhou2018,Pinto2020} of a target $i$ by the team of agents.

\paragraph{\textbf{Graph Topology}}

We embed a directed graph topology $\mathcal{G}=(\mathcal{T},\mathcal{E})$ into the mission space such that the \emph{targets} are represented by the graph \emph{vertices} $\mathcal{T}=\{1,2,\ldots,M\}$ and the inter-target \emph{trajectory segments} are represented by the graph \emph{edges} $\mathcal{E}\subseteq\{(i,j):i,j\in\mathcal{T}\}$. We point out that these trajectory segments may take arbitrary shapes in $\mathbb{R}^{n}$ so as to account for potential constraints (e.g., physical obstacles) in the agent motion. In the graph $\mathcal{G}$, each trajectory segment represented by an edge $(i,j)\in\mathcal{E}$ is assigned a (predefined) value $\rho_{ij}\in\mathbb{R}_{\geq 0}$ representing the \emph{transit time} an agent spends to travel from target $i$ to $j$. Based on $\mathcal{E}$, the \emph{neighbor-set} and the \emph{neighborhood} of a target $i\in\mathcal{T}$ are defined respectively as
\begin{equation}
\mathcal{N}_{i}\triangleq\{j:(i,j)\in\mathcal{E}\}\mbox{ and }\bar
{\mathcal{N}}_{i}=\mathcal{N}_{i}\cup\{i\}. \label{Eq:Neighborset}%
\end{equation}

Note that our analysis in this paper is independent of the agent motion dynamic model which ultimately determines the values of $\rho_{ij}$ for edges $(i,j)\in\mathcal{E}$. An extension of this model is to consider $\rho_{ij}$ as functions of controllable motion variables (e.g., speed, acceleration). 

\paragraph{\textbf{Objective}}

Our objective is to minimize the \emph{mean system uncertainty} $J_{T}$ over a finite time interval $[0,T]$:
\begin{equation}
J_{T} \triangleq\frac{1}{T}\int_{0}^{T}\sum_{i\in\mathcal{T}}R_{i}(t)dt,
\label{Eq:MainObjective}%
\end{equation}
by controlling the motion of the team of agents through a suitable set of feasible distributed controllers described next.

\paragraph{\textbf{Control}}

Based on the graph topology $\mathcal{G}$, whenever an agent $a\in\mathcal{A}$ is ready to leave a target $i\in\mathcal{T}$, its \emph{next-visit} target $j$ is selected from $\mathcal{N}_{i}$. Next, the agent travels on the trajectory segment $(i,j)\in\mathcal{E}$ to arrive at target $j$ spending a transit time $\rho_{ij}$. Subsequently, it selects a \emph{dwell-time} $\tau_{j}\in\mathbb{R}_{\geq0}$ to spend at target $j$ (which contributes to decreasing $R_{j}(t)$), and then makes another next-visit decision.

Therefore, in a PMN problem, the control exerted on an agent consists of a sequence of \emph{next-visit} targets $j\in\mathcal{N}_{i}$ and \emph{dwell-times} $\tau_{i}\in\mathbb{R}_{\geq0}$. Our goal is to determine $(\tau_{i},j)$ for any agent residing at any target $i$ at any time $t\in\lbrack0,T]$ which are collectively optimal in the sense of minimizing \eqref{Eq:MainObjective}. As pointed out in \cite{Zhou2019,Welikala2019P3}, this is a challenging task due to the nature of the feasible control space, even for the simplest PMN problem configurations.

\paragraph{\textbf{Receding Horizon Control}}

The on-line distributed gradient-based TCP method proposed in \cite{Zhou2019} requires each agent to use a set of \emph{thresholds} applied to its neighborhood target uncertainties $\{R_{j}(t):j\in\bar{\mathcal{N}}_{i}\}$ in order to determine its dwell-time $\tau_{i}$ and next-visit $j\in \mathcal{N}_{i}$ decisions. Thus, the objective in \eqref{Eq:MainObjective} is viewed as a function of these threshold parameters. Starting from an arbitrary set of thresholds, each agent iteratively updates them using a gradient technique that exploits the information from observed events in agent trajectories. Although this TCP approach is efficient due to the use of IPA, it is limited by the presence of local optima. Indeed, all gradient-based methods for designing optimal agent trajectories are subject to this limitation in view of the non-convex nature of the objective functions involved. 

To address this limitation, this paper proposes an \emph{Event-Driven Receding Horizon Controller} (RHC) for each agent $a\in\mathcal{A}$. The basic idea of using a receding horizon in seeking solutions to hard dynamic optimization problems has its root in Model Predictive Control (MPC). Our approach (i) exploits the event-driven nature of the control actions in PMN problems, (ii) includes the planning horizon for each iteration of the RHC as a decision variable to be optimized, and (iii) provides solutions to each event-driven optimization problem which do not require a gradient-based method. As a result, the advantages of this approach are (i) a reduction in computational complexity by orders of magnitude due to the flexibility in the frequency of control updates, (ii) performance improvements by avoiding many local optima resulting from gradient-based optimization methods, and (iii) a parameter-free controller by optimizing the planning horizon length rather than treating it as a tunable parameter.

As introduced in \cite{Li2006} and extended later on in \cite{khazaeni2018b},\cite{Chen2019}, an event-driven receding horizon controller solves an optimization problem of the form \eqref{Eq:MainObjective} but limited to a prespecified \emph{planning horizon} whenever an \emph{event} is observed; the resulting (optimal) control is then executed over a generally shorter \emph{action horizon} defined by the occurrence of the next event of interest to the controller. This process is iteratively repeated in event-driven fashion.

In the PMN problem, the aim of the RHC, when invoked at time $t$ for an agent residing at target $i\in\mathcal{T}$, is to determine the immediate next-visit target $j\in\mathcal{N}_{i}$ and dwell times at targets $i$ and $j$ (i.e., $\tau_{i}$ and $\tau_{j}$ respectively). These three decisions jointly form a control $U_{i}(t)$ and its optimal value is determined by solving an optimization problem of the form:
\begin{equation}
U_{i}^{\ast}(t)=\underset{U_{i}(t)\in\mathbb{U}(t)}{\arg\min}\ \left[
J_{H}(X_{i}(t),U_{i}(t);H)+\hat{J}_{H}(X_{i}(t+H))\right]
,\label{Eq:RHCProblem}%
\end{equation}
where $X_{i}(t)$ is the current local state and $\mathbb{U}(t)$ is the \emph{feasible control set} at $t$ (whose exact definition will be provided later). The term $J_{H}(X_{i}(t),U_{i}(t);H)$ is the immediate cost over the planning horizon $[t,t+H]$ and $\hat{J}_{H}(X_{i}(t+H)$ is an estimate of the future cost evaluated at the end of the planning horizon $t+H$. In prior work \cite{Li2006,khazaeni2018b,Chen2019}, the value of the planning horizon length $H$ is selected \emph{exogenously}. However, in this paper we will include this value into the optimization problem and ignore the $\hat{J}_{H}(X_{i}(t+H))$ term. Thus, by optimizing the planning horizon, we compensate for the complexity and intrinsic inaccuracy of the $\hat{J}_{H}(X_{i}(t+H))$ term whose evaluation requires information from the full network. However, this is not possible in the proposed \emph{distributed} RHC setting which allows each agent to separately solve \eqref{Eq:RHCProblem} using only local state information.

Even though the target uncertainty model \eqref{Eq:TargetDynamics} and the main objective \eqref{Eq:MainObjective} play a crucial role in deriving the exact solution to the RHCP, the proposed overall RHC architecture can be readily adopted for persistent monitoring problems with different target state models and objectives such as the ones used in \cite{Lan2013,Song2014,Rezazadeh2019,Pinto2020}.

\subsection{Preliminary Results}

According to \eqref{Eq:TargetDynamics}, the target state $R_{i}(t)$, $i\in\mathcal{T}$, is piece-wise linear and its gradient $\dot{R}_{i}(t)$ changes only when one of the following (\emph{strictly local} to target $i$) \emph{events} occurs: 
(\romannum{1}) An agent arrival at $i$, 
(\romannum{2}) An event $[R_{i}\rightarrow0^{+}]$, or 
(\romannum{3}) An agent departure from $i$. 
Let the occurrence of such events associated with target $i$ be indexed by $k=1,2,\ldots$ and the respective event occurrence times be denoted by $t_{i}^{k}$ with $t_{i}^{0}=0$. Then, 
\begin{equation}
\dot{R}_{i}(t)=\dot{R}_{i}(t_{i}^{k}),\ \ \forall t\in\lbrack t_{i}^{k},t_{i}^{k+1}). \label{Eq:TargetDynamics2}
\end{equation}

As pointed out in \cite{Welikala2019P3,Yu2016}, allowing overlapping dwell intervals at some target (also referred to as \textquotedblleft simultaneous target sharing\textquotedblright) is known to lead to solutions with poor performance levels (clearly, this issue does not apply to single-agent problems). This observation motivates us to enforce a constraint on the controller to ensure:
\begin{equation}
\label{Eq:NoTargetSharing}N_{i}(t)\in\{0,1\},\ \forall t\in[0,T],\ \forall
i\in\mathcal{T}.
\end{equation}

If the control constraint \eqref{Eq:NoTargetSharing} is enforced, it follows from \eqref{Eq:TargetDynamics} and \eqref{Eq:TargetDynamics2} that the sequence $\{\dot{R}_{i}(t_{i}^{k})\}$, $k=0,1,2,\ldots$, is a \emph{cyclic order} of three elements: $\{-(B_{i}-A_{i}),0,A_{i}\}$.
Next, in order to ensure that each agent is capable of enforcing the event $[R_{i}\rightarrow0^{+}]$ at any $i\in\mathcal{T}$, the following simple stability condition is assumed (similar to \cite{Welikala2019P3}). 
\begin{assumption}
\label{As:TargetUncertaintyRateInequality} Target uncertainty rate parameters $A_{i}$ and $B_{i}$ of each target $i\in\mathcal{T}$ satisfy $0\leq A_{i} < B_{i}$.
\end{assumption}

\paragraph{\textbf{Decomposition of the Objective Function}}

Let the contribution of target $i$ to the objective $J_{T}$ in \eqref{Eq:MainObjective} during a time period $[t_{0},t_{1})$ be $\frac{1}{T}\,J_{i}(t_{0},t_{1})$ where
\[
J_{i}(t_{0},t_{1})\triangleq\int_{t_{0}}^{t_{1}}R_{i}(t)dt.
\]
Theorem \ref{Th:Contribution} provides a target-wise and temporal decomposition of the objective function $J_{T}$ based on $J_{i}(t_{0},t_{1})$.

\begin{theorem}
\label{Th:Contribution} The contribution to the objective $J_{T}$ by target $i\in\mathcal{T}$ during a time period $[t_{0},t_{1})\subseteq\lbrack t_{i}^{k},t_{i}^{k+1})$ for some $k\in\mathbb{Z}_{\geq0}$ is $\frac{1}{T}\,J_{i}(t_{0},t_{1})$, where,
\begin{equation}
J_{i}(t_{0},t_{1})=\frac{(t_{1}-t_{0})}{2}\left[  2R_{i}(t_{0})+\dot{R}%
_{i}(t_{0})(t_{1}-t_{0})\right]  . \label{Eq:Contribution}%
\end{equation}
\end{theorem}

\begin{proof}
In \eqref{Eq:MainObjective}, by taking the summation operator out of the integration and then splitting the time interval $t\in [0,T]$ of the integration of $R_i(t)$ profile into three parts gives
\begin{align}
    J_T
    &= \frac{1}{T} \left[ \sum_{j\in \mathcal{T}\backslash \{i\}} \int_0^T R_j(t)dt \right] \nonumber \\ 
    &+ \frac{1}{T} \left[ \int_0^{t_0} R_i(t)dt +\int_{t_0}^{t_1} R_i(t)dt + \int_{t_1}^T R_i(t)dt \right],
\end{align}
where $\cdot\backslash\cdot$ represents the set subtraction operator. Therefore, clearly the contribution of target $i$ to the main objective $J_T$ during the time period $t\in[t_0,t_1)$ is $\frac{1}{T}J_i(t_0,t_1)$ where,
$$J_i(t_0,t_1) = 
\int_{t_0}^{t_1}R_i(t)dt.$$
Moreover, since $[t_0,t_1) \subseteq [t_i^k,t_i^{k+1})$, the relationship \eqref{Eq:TargetDynamics2} implies that $\int_{t_0}^{t_1}R_i(t)dt$ represents  the area of a trapezoid (whose parallel sides are $R_i(t_0)$ and $R_i(t_1)$). Therefore, 
$$ J_i(t_0,t_1) = 
\left[ \frac{R_i(t_0)+R_i(t_1)}{2}\times (t_1-t_0)\right].$$ 
Also, \eqref{Eq:TargetDynamics2} gives that $R_i(t_1) = R_i(t_0) + \dot{R}_i(t_0)(t_1-t_0)$. Therefore,
\begin{align*}
    J_i(t_0,t_1) = \frac{(t_1-t_0)}{2}
    \left[ 2R_i(t_0) + \dot{R}_i(t_0)(t_1-t_0)\right].
\end{align*}
\end{proof}

A simple corollary of Theorem \ref{Th:Contribution} is to extend it to any interval $[t_{0},t_{1})$ which may include one or more event times $t_{i}^{k}$.

\begin{corollary}
\label{Col:Contribution} Let $t_{0}=t_{i}^{k}$ be the time when an agent arrived at target $i\in\mathcal{T}$, followed by an $[R_{i}\rightarrow0^{+}]$ event at $t=t_{i}^{k+1}$ and a departure event at $t=t_{i}^{k+2}$. Then, for any $t_{1}$ such that $t_{i}^{k+2} \leq t_{1}\leq t_{i}^{k+3}$, 
\begin{equation}
J_{i}(t_{0},t_{1})=\frac{u_{i}^{0}}{2}\left[  2R_{i}(t_{0})-(B_{i}-A_{i}%
)u_{i}^{0}\right]  +\frac{u_{i}^{1}}{2}\left[  A_{i}u_{i}^{1}\right]  ,
\label{Eq:Contribution2}%
\end{equation}
where $u_{i}^{0}=t_{i}^{k+1}-t_{0}$ and $u_{i}^{1}=t_{1}-t_{i}^{k+2}$.
\end{corollary}

\begin{proof}
Applying Theorem \ref{Th:Contribution} to the three interested inter-event intervals gives  
\begin{align*}
    J_i(t_0,t_1) = 
&\frac{(t_i^{k+1}-t_0)}{2}
    \left[2R_i(t_0) + \dot{R}_i(t_0)(t_i^{k+1}-t_0)\right]\\
&+\frac{(t_i^{k+2}-t_i^{k+1})}{2}
    \left[2R_i(t_i^{k+1}) + \dot{R}_i(t_i^{k+1})(t_i^{k+2}-t_i^{k+1})\right]\\
&+\frac{(t_1-t_i^{k+2})}{2}
    \left[2R_i(t_i^{k+2}) + \dot{R}_i(t_i^{k+2})(t_1-t_i^{k+2})\right].
\end{align*}
Now, using: (\romannum{1}) the definitions of $u_i^0,u_i^1$ and $v_i \triangleq (t_i^{k+2}-t_i^{k+1})$, (\romannum{2}) the $\dot{R}_i$ values stated earlier and (\romannum{3}) the fact that $R_i(t_i^{k+1})=R_i(t_i^{k+2})=0$, the above expression can be simplified as
\begin{align*}
    J_i(t_0,t_1) = 
&\frac{u_i^0}{2}
    \left[2R_i(t_0) - (B_i-A_i)u_i^0\right] +\frac{v_i}{2}
    \left[2\times 0 + 0\times v_i\right]\\
&+\frac{u_i^1}{2}
    \left[2\times 0 + A_iu_i^1\right],\\
=& \frac{u_i^0}{2}
    \left[2R_i(t_0) - (B_i-A_i)u_i^0\right] +\frac{u_i^1}{2}
    \left[A_iu_i^1\right].
\end{align*}
\end{proof}

\paragraph{\textbf{Local Objective Function}}

In developing a distributed event-driven controller for an agent residing at some target $i\in\mathcal{T}$, this agent has to have the access to any necessary \emph{local} information from the neighborhood $\bar{\mathcal{N}}_{i}$. In particular we assume: 
\begin{enumerate}
    \item Target $i\in\mathcal{T}$ receives the information $\{A_{j},B_{j}\}$ at $t=0$ from its neighbors $j\in\mathcal{N}_{i}$.
    \item Target $i\in\mathcal{T}$ receives information $\{R_{j}(t),\dot{R}_{j}(t)\}$ at any time $t$ from its neighbors $j\in\mathcal{N}_{i}$. 
    \item Any agent $a\in\mathcal{A}$ residing at target $i\in\mathcal{T}$ at time $t$ (i.e., $s_{a}(t)=Y_{i}$) can obtain the above two types of information from any neighboring target $j\in\bar{\mathcal{N}}_{i}$.
\end{enumerate}

Next we define the \emph{local objective function} of target $i$ over a time period $[t_{0},t_{1})\subseteq\lbrack0,T]$ as 
\begin{equation}
\bar{J}_{i}(t_{0},t_{1})\triangleq\sum_{j\in\bar{\mathcal{N}}_{i}}J_{j}(t_{0},t_{1}). \label{Eq:LocalObjectiveFunction}
\end{equation}
The value of each $J_{j}(t_{0},t_{1})$ above is obtained through Theorem \ref{Th:Contribution} and its extension in Corollary \ref{Col:Contribution} if $[t_{0},t_{1})$ includes additional events (where $[t_{0},t_{1})$ is decomposed into a sequence of corresponding inter-event time intervals). In essence, $\bar{J}_{i}(t_{0},t_{1})$ can be evaluated by an agent residing at target $i$ at any required (event-driven) time instant (more details are provided in the sequel).

\subsection{RHC Problem (RHCP) Formulation}
\label{SubSec:RHCIntro}

Consider a situation where agent $a\in\mathcal{A}$ resides at target $i\in\mathcal{T}$ at some $t\in\lbrack0,T]$. In our distributed setting, we assume that agent $a$ is made aware of only \emph{local events} occurring in the neighborhood $\bar{\mathcal{N}}_{i}$. As defined earlier, the control $U_{i}(t)$ consists of (i) the \emph{dwell-time} $\tau_{i}$ at the current target $i$, (ii) the \emph{next target} $j\in\mathcal{N}_{i}$ to visit and (iii) the \emph{dwell-time} $\tau_{j}$ at the selected next target $j$. Moreover, a dwell time decision $\tau_{i}$ (or $\tau_{j}$) can be divided into two interdependent decisions: (\romannum{1}) the \emph{active time} $u_{i}$ (or $u_{j}$) when $R_{i}(t)>0$ ($R_{j}(t)>0$) and (\romannum{2}) the \emph{inactive }(or\emph{ idle}) \emph{time} $v_{i}$ (or $v_{j}$) when $R_{i}(t)=0$ ($R_{j}(t)=0$), as shown in Fig. \ref{Fig:OneVisitTimeline}. Thus, agent $a$ has to optimally choose five decision variables which form the control vector
\[
U_{i}(t)\triangleq\lbrack u_{i}(t),\,v_{i}(t),\,j(t),\,u_{j}(t),\,v_{j}(t)].
\]
Note that $j(t)$ is discrete while the remaining four components of $U_{i}(t)$ are real-valued. The time argument of each component of $U_{i}(t)$ is omitted henceforth for notational convenience.

\begin{figure}[b]
\centering
\includegraphics[width=3.4in]{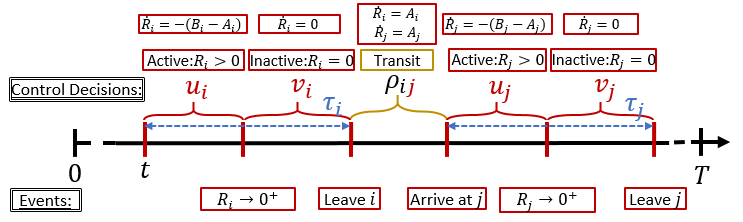} \caption{Event
timeline and control decisions under RHC.}%
\label{Fig:OneVisitTimeline}%
\end{figure}

\paragraph{\textbf{Fixed Planning Horizon}}

Recalling \eqref{Eq:RHCProblem}, the RHC depends on the planning horizon $H\in\mathbb{R}_{\geq0}$ which is normally viewed as a \emph{fixed} control parameter. Intuitively, selecting $H\geq\max_{(i,j)\in\mathcal{E}}\rho_{ij}$ ensures that all agents will consider traveling to all of their current neighboring targets. However, note that $t+H$ is constrained by $t+H\leq T$, hence, if this is violated, without loss of generality, we redefine (truncate) the planning horizon to be $H=T-t$ (i.e., $H\leq T-t$).


In \eqref{Eq:RHCProblem}, let the current local state be $X_{i}(t)=\{R_{m}(t):m\in\bar{\mathcal{N}}_{i}\}$ and let us decompose the control $U_{i}(t)$ into its real-valued components and its discrete component (omitting time arguments) as $U_{ij} \triangleq[u_{i},v_{i},u_{j},v_{j}]\in\mathbb{U}$ and $j\in\mathcal{N}_{i}$ respectively. Now, if the objective function $J_{H}(\cdot)$ in \eqref{Eq:RHCProblem} is chosen to reflect the contribution to the main objective $J_{T}$ in \eqref{Eq:MainObjective} by the targets in the neighborhood $\bar{\mathcal{N}}_{i}$ over the fixed time period $[t,t+H]$ (which is provided by \eqref{Eq:LocalObjectiveFunction} and Theorem \ref{Th:Contribution}), then,
\begin{equation}
\label{Eq:RHCConventionalChoices}\begin{aligned} &J_H(X_i(t),U_{ij};H) = \frac{1}{H}\ \bar{J}_i(t,t+H) \mbox{ with }\\ &\mathbb{U} = \{U:U\in \mathbb{R}^4,\ U \geq 0,\ \vert U\vert + \rho_{ij} = H\}. \end{aligned}
\end{equation}
The feasible control set $\mathbb{U}$ is such that $u_{i},$ $v_{i},$ $u_{j},$ and $v_{j}$ are non-negative real variables. Note that the notation $|\cdot|$ is used to represent the 1-norm or the cardinality operator when the argument is respectively a vector or a set.

In this setting, the optimal controls are obtained by solving the following set of optimization problems, henceforth called the RHC Problem (RHCP): 
\begin{align}
U_{ij}^{\ast}  &  =\underset{U_{ij}\in\mathbb{U}}{\arg\min}\ J_{H}%
(X_{i}(t),U_{ij};H);\ \forall j\in\mathcal{N}_{i}%
\ \mbox{ and }\label{Eq:RHCGenSolStep1}\\
j^{\ast}  &  =\underset{j\in\mathcal{N}_{i}}{\arg\min}\ J_{H}(X_{i}%
(t),U_{ij}^{\ast};H). \label{Eq:RHCGenSolStep2}%
\end{align}
Observe that \eqref{Eq:RHCGenSolStep1} involves solving $|\mathcal{N}_{i}|$ optimization problems, one for each $j\in\mathcal{N}_{i}$. Then, \eqref{Eq:RHCGenSolStep2} determines $j^{\ast}$ through a simple numerical comparison. Therefore, the final optimal decision variables are $U_{ij^{\ast}}^{\ast}$ and $j^{\ast}$.

According to \eqref{Eq:RHCConventionalChoices}, the choices for the four control variables in $U_{ij}$ are restricted by $U_{ij}\in\mathbb{U}$ such that $|U_{ij}|+\rho_{ij}=H$ (see also Fig.\ref{Fig:OneVisitTimeline}). Therefore, the selection of $H$ directly affects the RHCP's optimal solution. For example, if $H$ is very small, clearly the resulting optimal decisions $U_{ij^{\ast}}^{\ast}$ and $j^{\ast}$ (i.e., $U_{i}^{\ast}(t)$) are myopic.  Attempting to find the optimal choice of $H$ without compromising the on-line distributed nature of the proposed RHCP solution is a challenging task.

\paragraph{\textbf{Variable Planning Horizon}}

We address this problem by introducing a \emph{variable horizon} $w$ defined as
\begin{equation}
w\ \triangleq\ |U_{ij}|+\rho_{ij}\ =\ u_{i}+v_{i}+\rho_{ij}+u_{j}+v_{j},
\label{Eq:VariableHorizon}%
\end{equation}
and replacing $H$ in \eqref{Eq:RHCConventionalChoices} by $w$ while, at the same time, imposing the constraint $w\leq H$. It is important to observe that $w$ defined in \eqref{Eq:VariableHorizon} is a function of $u_{i}(t),v_{i}(t),u_{j}(t),v_{j}(t)$ but we omit explicitly showing this dependence for notational convenience. It is also important to note that now the value of $H$ is not critical as long as it is sufficiently large; for instance, it can be chosen to be $T-t$. Thus, we see that the solution of the RHCP \eqref{Eq:RHCGenSolStep1}-\eqref{Eq:RHCGenSolStep2} can now be obtained without any tunable parameters, making the resulting controller parameter-free. The objective function $J_{H}$ and the feasible control set $\mathbb{U}$ in the RHCP are now chosen as
\begin{equation}
\begin{aligned} &J_H(X_i(t),U_{ij};H) = \frac{1}{w}\bar{J}_i(t,t+w) \mbox{ and } \\ &\mathbb{U} = \{U:U \in \mathbb{R}^4,\ U \geq 0,\ \vert U \vert +\rho_{ij} \leq H\}. \end{aligned} \label{Eq:RHCNewChoices}%
\end{equation}
Therefore, this novel RHCP formulation allows us to simultaneously determine the \emph{optimal planning horizon} size $w^{\ast}$ in terms of the optimal control $U_{i}^{\ast}(t)$ as
\begin{equation}
w^{\ast}=|U_{ij^{\ast}}^{\ast}|+\rho_{ij^{\ast}}. \label{Eq:OptimumHorizon}%
\end{equation}
On the other hand, this incorporation of $w$ in \eqref{Eq:RHCNewChoices}, as opposed to \eqref{Eq:RHCConventionalChoices}, makes the denominator term of the objective function control-dependent and introduces new technical challenges that we address in the rest of the paper. To accomplish this, we will exploit structural properties of \eqref{Eq:RHCNewChoices} and show that the RHCP in \eqref{Eq:RHCGenSolStep1} can be solved analytically and efficiently to obtain its globally optimal solution.

\paragraph{\textbf{Event-Driven Action Horizon}}

As in all receding horizon controllers, the solution of each optimization problem over a certain planning horizon is executed only over a shorter \emph{action horizon} $h$. In the distributed RHC setting, the value of $h$ is determined by the first event that the agent observes after $t$, the time instant when the RHCP was last solved by the agent. Thus, in contrast to time-driven receding horizon control, the RHC solution is updated whenever asynchronous events occur; this prevents unnecessary steps to re-solve the RHCP \eqref{Eq:RHCGenSolStep1}-\eqref{Eq:RHCGenSolStep2} with \eqref{Eq:RHCNewChoices}.

Figure \ref{Fig:ActionHorizon} shows an example of three consecutive action horizons (labeled $h_{1},h_{2}$ and $h_{3}$) observed by an agent $a$ after an event at $t$ triggers the solution of the RHCP. Note that $w_{1}^{\ast},$ $w_{2}^{\ast},$ $w_{3}^{\ast}$ represent the three optimal planning horizon sizes (i.e., $w^{\ast}$ in \eqref{Eq:OptimumHorizon}) determined at each respective local event time $t$, $t+h_{1}$ and $t+h_{1}+h_{2}$.

\begin{figure}[t]
\centering
\includegraphics[width=3.4in]{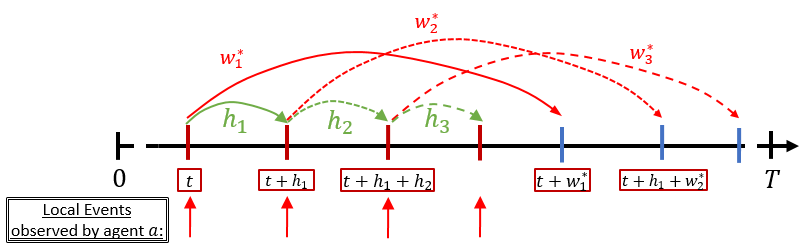}
\caption{Event driven receding horizon control approach.}%
\label{Fig:ActionHorizon}%
\end{figure}

Although our PMN problem setting is deterministic, in a general RHC setting there may be uncontrollable random events that trigger the controller (we will consider such cases when investigating the robustness of the RHC in Section \ref{Sec:SimulationResults}). Thus, in general, the determination of the action horizon $h$ may be \emph{controllable} or \emph{uncontrolled}. The former corresponds to the occurrence of any one event resulting from an agent solving a RHCP at some earlier time. We define next the three controllable events associated with an agent when it resides at target $i$; each of these events defines the action horizon $h$ following the solution $U_{i}^{\ast}(t)$ of a RHCP obtained by this agent at some time $t\in\lbrack0,T]$:

\textbf{1. Event }$[h\rightarrow u_{i}^{\ast}]$\textbf{: } This event occurs at time $t+u_{i}^{\ast}(t)$. If $R_{i}(t+u_{i}^{\ast}(t))=0$, this event coincides with an $[R_{i}\rightarrow0^{+}]$ event. Otherwise, $R_{i}(t+u_{i}^{\ast}(t))>0$ implies that the solution of the associated RHCP dictates ending the active time at target $i$ before the $[R_{i}\rightarrow0^{+}]$ event. Therefore, in that case, by definition, no inactive time may follow, i.e., $v_{i}^{\ast}(t)=0$, and $[h\rightarrow u_{i}^{\ast}]$ coincides with a departure event from target $i$.

\textbf{2. Event }$[h\rightarrow v_{i}^{\ast}]$\textbf{: } This event occurs at time $t+v_{i}^{\ast}(t)$. It is only feasible after an event $[h \rightarrow u_{i}^{\ast}]$ has occurred, including the possibility that $u_{i}^{\ast}(t)=0$ in the RHCP solution determined at $t$. Clearly, this always coincides with an agent departure event from target $i$.

\textbf{3. Event }$[h\rightarrow\rho_{ij^{\ast}}]$\textbf{: } This event occurs at time $t+\rho_{ij^{\ast}(t)}$. It is only feasible after an event $[h\rightarrow u_{i}^{\ast}]$ or $[h\rightarrow v_{i}^{\ast}]$ has occurred, including the possibility that $u_{i}^{\ast}(t)=0$ and $v_{i}^{\ast}(t)=0$ in the RHCP solution determined at $t$. Clearly, this coincides with an arrival event at target $j^{\ast}(t)$ as determined by the RHCP solution obtained at time $t$.

Observe that these events are mutually exclusive, i.e., only one is feasible at any one time. In addition, there are uncontrollable events associated with a neighboring target $j\in\mathcal{N}_{i}$ other than target $i$. In particular, let us define two additional events that may occur at any neighbor $j\in\mathcal{N}_{i}$ and trigger an event at the agent residing at target $i$. These events have been designed to enforce the no-simultaneous-target-sharing policy (i.e., the control constraint \eqref{Eq:NoTargetSharing}) and apply only to \emph{multi}-agent persistent monitoring problems.

A target $j\in\mathcal{T}$ is said to be \emph{covered} at time $t$ if it already has a residing agent or if an agent is en route to visit it from a neighboring target in $\mathcal{N}_{j}$. Thus, target $j$ is covered only if $\exists k\in\mathcal{N}_{j}$ and $\tau\in\lbrack t,t+\rho_{kj})$ such that $\sum_{a\in\mathcal{A}}\mathbf{1}\{s_{a}(\tau)=Y_{j}\}>0$. Since neighboring targets communicate with each other, this information can be determined at any target in $\bar{\mathcal{N}}_{j}$ at any time $t$. Therefore, an agent $a\in\mathcal{A}$ residing at target $i$ can prevent target sharing at target $j\in\mathcal{N}_{i}$ by simply modifying the neighbor set $\mathcal{N}_{i}$ used in the RHCP solved at time $t$ to exclude all covered targets. Let us use $\mathcal{N}_{i}(t)$ to indicate a \emph{time-varying} neighborhood of target $i$. Then, if target $j$ becomes covered at time $t$, we set
\begin{equation}
\mathcal{N}_{i}(t)=\mathcal{N}_{i}(t^{-})\backslash\{j\}.
\label{Eq:ReducedNeighborhood}%
\end{equation}
The effect of this modification is clear if a RHCP solved by an agent at target $i$ at some time $t$ leads to a next-visit solution $j^{\ast} \in\mathcal{N}_{i}(t)$: if this is followed by an event at $t^{\prime}>t$ causing target $j^{\ast}$ to become covered, then $\mathcal{N}_{i}(t^{\prime})=\mathcal{N}_{i}(t)\backslash\{j^{\ast}\}$ and the agent at target $i$ (whether active or inactive) must re-solve the RHCP at $t^{\prime}$ with the new $\mathcal{N}_{i}(t^{\prime})$. Note that as soon as an agent $a$ is en route to $j^{\ast}$, then $j^{\ast}$ becomes covered, hence preventing any other agent from visiting $j^{\ast}$ prior to agent $a$'s subsequent departure from $j^{\ast}$.

Based on this discussion, we define the following two additional \emph{neighbor-induced local events} triggered at $j\in\mathcal{N}_{i}$ and affecting an agent $a$ residing at target $i$:

\textbf{4. Covering Event }$C_{j}\mathbf{,}$ $j\in\mathcal{N}_{i}$: This event causes $\mathcal{N}_{i}(t)$ to be modified to $\mathcal{N}_{i}(t^{-})\backslash\{j\}$.

\textbf{5. Uncovering Event }$\bar{C}_{j}\mathbf{,}$ $j\in\mathcal{N}_{i}$: This event causes $\mathcal{N}_{i}(t)$ to be modified to $\mathcal{N}_{i}(t^{-})\cup\{j\}$.

If one of these two events takes place while an agent residing at target $i$ is either active or inactive, then the RHCP \eqref{Eq:RHCGenSolStep1}-\eqref{Eq:RHCGenSolStep2} is re-solved to account for the updated $\mathcal{N}_{i}(t)$. This may affect the values of the optimal solution $U_{i}^{\ast}$ from the previous solution. Note, however, that the new solution will still give rise to an event $[h\rightarrow u_{i}^{\ast}]$ (if the RHCP is solved while the agent is active) or $[h\rightarrow v_{i}^{\ast}]$ (if the RHCP is solved while the agent is inactive).

The nature of the events we have defined and the event timeline (Fig. \ref{Fig:OneVisitTimeline}), together with the fact that all transit times are non-zero, ensures that: (i) each agent can only have a finite number of events between each of its arrival and subsequent departure events, and (ii) each agent has to spend a finite time between each of its departure and subsequent arrival events. Therefore, Zeno behaviors (where an infinite number of events occur in finite time) do not occur in this setting with the proposed RHC architecture.

\paragraph{\textbf{Computing }$\bar{J}_i$} 

The existence of multiple agents hinders the ability to analytically express the function $\bar{J}_i(t,t+w)$ involved in the RHCP as it requires the agent $a$ (whom is residing in $i$ at $t$ planning a trajectory that visits neighbor $j$) to have the knowledge of the events that will occur at each neighbor $m\in \mathcal{N}_i\backslash\{j\}$ during the future time period $[t,t+w)$ (see \eqref{Eq:RHCNewChoices}, \eqref{Eq:LocalObjectiveFunction} and Theorem \ref{Th:Contribution}). 

However, this task becomes tractable when the aforementioned neighbor-set modification in \eqref{Eq:ReducedNeighborhood} is employed. For example, upon using \eqref{Eq:ReducedNeighborhood}, if some neighbor $m\in\mathcal{N}_i(t)$, then, there is no other agent residing in or en route to target $m$ at $t$. Therefore, clearly, $\dot{R}_m(\tau) = A_m$ for the  period $\tau\in[t,t+r)$ where $r \geq \min_{q\in \mathcal{N}_m} \rho_{qm}$. Now, if $[t,t+r)\subseteq[t,t+w)$, projections are used to estimate the remaining portion of the  $R_m(\tau)$ profile (i.e. for $\tau \in [t+r,t+w]$). This enables expressing $\bar{J}_i(t,t+w)$ analytically.

\paragraph{\textbf{Three Forms of The RHCP}}

It is clear from this discussion that the exact form of the RHCP to be solved at time $t$ depends on the event that triggered the end of the previous action horizon (i.e., the event occurring at time $t$) and the target state $R_{i}(t)$. In particular, there are three possible forms of the RHCP \eqref{Eq:RHCGenSolStep1}-\eqref{Eq:RHCGenSolStep2}: 

\textbf{RHCP1:} This problem is solved by an agent arriving at target $i$, i.e., when an event $[h\rightarrow\rho_{ki}]$ occurs at time $t$ for any $k\in\mathcal{N}_{i}(t)$. The solution $U_{i}^{\ast}(t)$ includes $u_{i}^{\ast}(t)\geq0$, representing the amount of time that the agent should be active at target $i$. This problem is also solved while the agent is active at target $i$ (i.e., while $R_{i}(t)>0$) if a $C_{j}$ or $\bar{C}_{j}$ event occurs for any $j\in\mathcal{N}_{i}(t)$.

\textbf{RHCP2: } This problem is solved by an agent residing at target $i$ when an event $[h\rightarrow u_{i}^{\ast}]$ occurs at time $t$ with $R_{i}(t)=0$. This problem is also solved while the agent is inactive (i.e., $R_{i}(t)=0$) at $i$ if a $C_{j}$ or $\bar{C}_{j}$ event occurs for any $j\in\mathcal{N}_{i}(t)$. In both cases, the solution $U_{i}^{\ast}(t)$ is now constrained to include $u_{i}^{\ast}(t)=0$ by default, since the agent can no longer be active at target $i$.

\textbf{RHCP3: } This problem is solved by an agent departing from target $i$ and may be triggered by one of two events: (i) Event $[h \rightarrow u_{i}^{\ast}]$ at time $t$ with $R_{i}(t)>0$. The solution $U_{i}^{\ast}(t)$ is constrained to include $u_{i}^{\ast}(t)=0$ by default; in addition, it is constrained to have $v_{i}^{\ast}(t)=0$ since the agent ceases being active while $R_{i}(t)>0$, implying that it must immediately depart from target $i$ without becoming inactive. (ii) Event $[h\rightarrow v_{i}^{\ast}]$ at time $t$, implying that the agent is no longer inactive and must depart from target $i$. As in case (i), the solution $U_{i}^{\ast}(t)$ is constrained to have both $u_{i}^{\ast}(t)=0$ and $v_{i}^{\ast}(t)=0$ by default.

\paragraph{\textbf{Complexity of RHCPs}}

As we will show next, all three problem forms of the RHCP discussed above can be solved to obtain the corresponding globally optimal solutions in closed form. Therefore, their complexity is constant and the overall RHC complexity scales linearly with the number of events occurring in $[0,T]$.

\section{Solving the Event-Driven Receding Horizon Control Problems}
\label{Sec:SolvingRHCPs}

In this section, we present the solutions to the identified three forms  of RHCPs discussed above. We begin with \textbf{RHCP3} due to its relative simplicity.

\subsection{Solution of \textbf{RHCP3}}

Recall that an agent solves \textbf{RHCP3} when it is ready to leave the target where it resides. Therefore, $u_{i}^{\ast}(t)=0$ and $v_{i}^{\ast}(t)=0$ by default and $U_{ij}$ in \eqref{Eq:RHCGenSolStep1} is reduced to $U_{ij}=[u_{j},v_{j}]$ with $U_{i}(t)=[j,u_{j},v_{j}]$. The obtained $j^{\ast}(t)$ directly defines the next destination to visit. Clearly, \textbf{RHCP3} plays a crucial role in defining agent trajectories in terms of targets visited.

The variable horizon $w$ \eqref{Eq:VariableHorizon} for this case is $w=\rho_{ij}+u_{j}+v_{j}$ and, from \eqref{Eq:RHCNewChoices}, $w$ is constrained so that $\rho_{ij}\leq w\leq H$. Therefore, $\rho_{ij}\leq H$ is assumed henceforth in this section.

\paragraph{\textbf{Constraints}}

We begin with identifying an upper bound for the active time control variable $u_{j}$. This is the maximum active time possible at target $j$, which is defined by the condition $R_{j}(t+\rho_{ij}+u_{j})=0$. Denoting this upper-bound by $u_{j}^{B}$, it follows from \eqref{Eq:TargetDynamics} that
\begin{equation}
u_{j}^{B}(t)\triangleq\frac{R_{j}(t+\rho_{ij})}{B_{j}-A_{j}}=\frac
{R_{j}(t)+A_{j}\rho_{ij}}{B_{j}-A_{j}}. \label{Eq:Lambdaj00}%
\end{equation}
Note that the dependence of $u_{j}^{B}(t)$ on $t$ captures its dependence on the initial condition $R_{j}(t)$; for notational simplicity, we shall henceforth omit this time dependence. A tighter upper-bound than $u_{j}^{B}$ on $u_{j}$, as well as an upper-bound on $v_{j}$, denoted respectively by $\bar{u}_{j}$ and $\bar{v}_{j}$ are imposed by the variable horizon constraint $w=u_{j}+v_{j}+\rho_{ij}\leq H$ as follows:
\begin{equation} 
\bar{u}_{j}\triangleq\min\{u_{j}^{B},\ H-\rho_{ij}\}\ \ \mbox{ and }\ \ \bar
{v}_{j}\triangleq H-(\rho_{ij}+u_{j}^{B}). \label{Eq:Lambdaj000}%
\end{equation}
Moreover, in order to have a positive inactive time $v_{j}>0$ a necessary condition is that it first spends the maximum active time possible $u_{j}=u_{j}^{B}$. Therefore, we now see that any feasible pair $U_{ij}=[u_{j},v_{j}]\in\mathbb{U}$ in \eqref{Eq:RHCGenSolStep1} belongs to one of the two constraint sets:
\begin{equation}
\mathbb{U}_{1}=\{0\leq u_{j}\leq\bar{u}_{j},\ v_{j}=0\}\mbox{ or }\mathbb{U}%
_{2}=\{u_{j}=u_{j}^{B},\ 0\leq v_{j}\leq\bar{v}_{j}\}, \label{Eq:Constraints1}%
\end{equation}
where $u_{j}^{B},\bar{u}_{j},\bar{v}_{j}$ are given in \eqref{Eq:Lambdaj00},\eqref{Eq:Lambdaj000} and $U_{ij}=[u_{j}^{B},0]$ is allowed to be a feasible control in both sets.

\paragraph{\textbf{Objective}}

Following from \eqref{Eq:RHCNewChoices}, the objective function corresponding to \textbf{RHCP3} is taken as 
$$J_{H}(U_{ij})=J_{H}(X_{i}(t),[0,0,U_{ij}];H)=\frac{1}{w}\bar{J}_{i}(t,t+w).$$
To obtain an exact expression for $J_{H}(U_{ij})$, first the local objective function $\bar{J}_{i}$ is decomposed using \eqref{Eq:LocalObjectiveFunction}:
\begin{equation}
\bar{J}_{i}=J_{j}+\sum_{m\in\bar{\mathcal{N}}_{i}\backslash\{j\}}J_{m}.
\label{Eq:ObjDerivationOP3Step1}%
\end{equation}
Considering the state trajectories shown in Fig. \ref{Fig:RHCP3Graphs} for the case where agent $a$ goes from target $i$ to target $j$ with decisions $u_{j}$ and $v_{j}$, both $J_{j}$ and $J_{m}$ terms in \eqref{Eq:ObjDerivationOP3Step1} are evaluated for the period $[t,t+w)$ using Theorem \ref{Th:Contribution} as 
\begin{align*}
    J_j &=
    \frac{\rho_{ij}}{T}\left[2R_j(t) + A_j\rho_{ij}\right]  + \frac{u_j}{T}\left[ 2(R_j(t) + A_j\rho_{ij})\right.  \\
    &-\left. (B_j-A_j)u_j \right] \mbox{ and }\\
    J_m &= 
    \frac{(\rho_{ij}+u_j+v_j)}{T}\left[2R_m(t) + A_m(\rho_{ij}+u_j+v_j)\right].
\end{align*}

\begin{figure}[!h]
    \centering
    \includegraphics[width=3in]{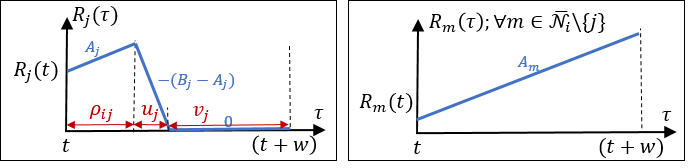}
    \caption{State trajectories during $[t,t+w)$ for the \textbf{RHCP3}.}
    \label{Fig:RHCP3Graphs}
\end{figure}

Now, combining the above two results and substituting it in \eqref{Eq:ObjDerivationOP3Step1} gives the complete objective function $J_H(U_{ij})$ as
\begin{equation}\label{Eq:OP3ObjectiveSimplified}
J_H(u_j,v_j) = \frac{C_{1} u_j^2 + C_{2} v_j^2 + C_{3} u_jv_j + C_{4} u_j + C_{5} v_j + C_{6}}{\rho_{ij}+u_j+v_j},
\end{equation}
where,
\begin{equation*}\label{Eq:OneStepCaseA}
\begin{aligned}
C_{1} &= \frac{1}{2}\left[\bar{A}-B_j\right],\ 
C_{2} = \frac{\bar{A}_j}{2},\ 
C_{3} = \bar{A}_j,\ 
C_{4} = \left[\bar{R}(t)+\bar{A}\rho_{ij}\right],\\ 
C_{5} &= \left[\bar{R}_j(t)+\bar{A}_j\rho_{ij}\right],\  
C_{6} = \frac{\rho_{ij}}{2}\left[2\bar{R}(t) + \bar{A}\rho_{ij}\right]
\end{aligned}
\end{equation*}
and (the neighbor-set parameters)
\begin{equation} \label{Eq:NeighborhoodParameters}
    \begin{aligned}
    \bar{A}_{ij} = \sum_{m\in \mathcal{N}_i \backslash \{j\}} A_m,\ \ \ \ 
    \bar{R}_{ij}(t) = \sum_{m\in \mathcal{N}_i \backslash \{j\}} R_m(t),\\
    \bar{A}_i = \bar{A}_{ij} + A_j,\ \ \bar{A}_j = \bar{A}_{ij} + A_i,\ \ 
    \bar{A} = \bar{A}_{ij} + A_i + A_j,\\
    \bar{R}_i = \bar{R}_{ij} + R_j,\ \ \bar{R}_j = \bar{R}_{ij} + R_i,\ \ 
    \bar{R} = \bar{R}_{ij} + R_i + R_j.\\
    \end{aligned}
\end{equation}
Note that $C_{i}\geq0$ for all $i$ except $C_{1}$ which is non-negative only when $B_{j}\leq\bar{A}$.

\paragraph{\textbf{Solving \textbf{RHCP3} for Optimal Control} $(u_{j}^{\ast},v_{j}^{\ast})$}

Based on the first step of RHCP \eqref{Eq:RHCGenSolStep1}, $(u_{j}^{*},v_{j}^{*})$ is given by 
\begin{equation}
\label{Eq:OP3_Formal}(u_{j}^{*},v_{j}^{*}) =\ \underset{(u_{j},v_{j})}%
{\arg\min}\ {J_{H}(u_{j},v_{j}),}%
\end{equation}
where $(u_{j},v_{j}) \in\mathbb{U}_{1}$ or $(u_{j},v_{j}) \in\mathbb{U}_{2}$ as in \eqref{Eq:Constraints1}.

\paragraph*{\textbf{- Case 1}}
$(u_{j},v_{j})\in\mathbb{U}_{1}=\{0\leq u_{j}\leq\bar{u}_{j},\ v_{j}=0\}$. Clearly, $v_{j}^{\ast}=0$ and \eqref{Eq:OP3_Formal} takes the form:
\begin{equation}
\begin{aligned} 
u_j^* =\  &\underset{u_j}{\arg\min}\ {J_H(u_j,0).} 
\\
&0 \leq u_j \leq \bar{u}_j
\end{aligned} \label{Eq:OP3_FormalPart1}%
\end{equation}

\begin{lemma}
\label{Lm:OP3_FormalPart1} The unique optimal solution of \eqref{Eq:OP3_FormalPart1} is 
\begin{equation}
u_{j}^{\ast}=%
\begin{cases}
\bar{u}_{j} & \mbox{ if }\bar{u}_{j}\geq u_{j}^{\#}\mbox{ and }\bar{A}%
<B_{j},\\
0 & \mbox{ otherwise,}
\end{cases}
\label{Eq:OP3_FormalPart1Sol}%
\end{equation}
where
\begin{equation}
u_{j}^{\#}=\frac{\bar{A}\rho_{ij}}{B_{i}-\bar{A}}.
\label{Eq:OP3_FormalPart1UCrit}%
\end{equation}
\end{lemma}

\begin{proof}
Using \eqref{Eq:OP3ObjectiveSimplified}, first and second order derivatives of $J_H(u_j,0)$  can be obtained respectively as $J^{\prime}(u_j)$ and $J^{\prime\prime}(u_j)$, where 
\begin{eqnarray*}
    J^{\prime}(u_j) &=& 
    \frac{\bar{A}-B_j}{2} + \frac{B_j \rho_{ij}^2}{2(\rho_{ij}+u_j)^2} \mbox{ and }\\  
    J^{\prime\prime}(u_j) &=& 
    -\frac{B_j\rho_{ij}^2}{(\rho_{ij}+u_j)^3}.
\end{eqnarray*}
Notice that $J^{\prime}(0)>0$ and $J^{\prime\prime}(u_j)<0, \forall u_j \geq 0$. This implies that $J^{\prime}(u_j)$ is monotonically decreasing with $u_j \geq 0$. Also note that $\lim_{u_j\rightarrow \infty} J^{\prime}(u_j) = \frac{\bar{A}-B_j}{2}$. 

Therefore, for the case where $\bar{A} \geq B_j$, the objective $J_H(u_j,0)$ is monotonically increasing with $u_j$. Hence $u_j^* = 0$ in \eqref{Eq:OP3_FormalPart1}.

For the case where $\bar{A}<B_j$, the limiting value of $J^{\prime}(u_j)$ is negative. This implies an existence of a maximum (of $J_H(u_j,0)$) at some $u_j \geq 0$. However, such a maximizing $u_j$ value is irrelevant to  \eqref{Eq:OP3_FormalPart1}. Nevertheless, a crucial $u_j$ value is located at the point where $J_H(0,0) = J_H(u_j,0)$ occurs. Using \eqref{Eq:OP3ObjectiveSimplified}, this can be determined as $u_j = \frac{C_6-C_4\rho_{ij}}{C_1}$ which simplifies to $u_j = u_j^{\#}$ where $u_j^{\#}$ is given in \eqref{Eq:OP3_FormalPart1UCrit}. According to the nature of $J^{\prime}(u_j)$ and $J^{\prime\prime}(u_j)$, it is clear that $J_H(u_j,0)$ should be decreasing with $u_j \geq u_j^{\#}$ (below $J_H(0,0)$ value). Therefore, when $\bar{u}_j \geq u_j^{\#}$ (and $\bar{A}<B_j$), $u_j^* = \bar{u}_j$ in \eqref{Eq:OP3_FormalPart1}. 
\end{proof}

Note that $u_{j}^{\#}$ in \eqref{Eq:OP3_FormalPart1UCrit} is completely known to the agent and it can be thought of as a break-even point for $u_{j}$, where if $\bar{u}_{j}$ allows $u_{j}$ to increase beyond the $u_{j}^{\#}$ value, it is always optimal to do so by choosing the extreme point $u_{j}^{\ast}=\bar{u}_{j}$. It is also worth pointing out that Lemma \ref{Lm:OP3_FormalPart1} still holds for \emph{time-varying} target parameters or the transit time $\rho_{ij}$. However, in such a case, the solution $u_{j}^{\ast}$ in \eqref{Eq:OP3_FormalPart1} will become time-dependent as it has to switch between $\bar{u}_{j}$ and $0$ depending on the conditions $\bar{u}_{j} \geq u_{j}^{\#}$ and $\bar{A}\leq B_{j}$.

\begin{remark}
When $H$ is sufficiently large, according to \eqref{Eq:Constraints1} and \eqref{Eq:Lambdaj00}, $\bar{u}_{j}=u_{j}^{B}=\frac{R_{j}(t)+A_{j}\rho_{ij}}{B_{j}-A_{j}}$. Therefore, the condition $\bar{u}_{j}\geq u_{j}^{\#}$ used in \eqref{Eq:OP3_FormalPart1} becomes explicitly dependent on the target state $R_{j}(t)$:
\begin{equation}
u_{j}^{\ast}=
\begin{cases}
\bar{u}_{j} & \mbox{ if }R_{j}(t)\geq\rho_{ij}\left[  \frac{B_{j}-A_{j}}%
{B_{j}-\bar{A}}\cdot\bar{A}-A_{j}\right]  \mbox{ and }\bar{A}<B_{j}\\
0 & \mbox{ otherwise.}
\end{cases}
\label{Eq:OP3_FormalPart1SolStateDep}%
\end{equation}
\end{remark}

\paragraph*{\textbf{- Case 2}} $(u_{j},v_{j})\in\mathbb{U}_{2}=\{u_{j}=u_{j}^{B},\ 0 \leq v_{j} \leq\bar{v}_{j}\}$. Obviously, $u_{j}^{\ast}=u_{j}^{B}$ and \eqref{Eq:OP3_Formal} takes the form: 
\begin{equation}
\begin{aligned} v_j^* =\  &\underset{v_j}{\arg\min}\ {J_H(u_j^B,v_j).} \\
&0 \leq v_j \leq \bar{v}_j
\end{aligned} \label{Eq:OP3_FormalPart2}%
\end{equation}

\begin{lemma}
\label{Lm:OP3_FormalPart2} The unique optimal solution of \eqref{Eq:OP3_FormalPart2} is 
\begin{equation}
v_{j}^{\ast} =
\begin{cases}
0 & \mbox{ if }\bar{A}\geq B_{j}\left[  1-\frac{\rho_{ij}^{2}}{(\rho
_{ij}+u_{j}^{B})^{2}}\right] \\
\min\{v_{j}^{\#},\,\bar{v}_{j}\} & \mbox{ otherwise, }
\end{cases}
\label{Eq:OP3_FormalPart2Sol}%
\end{equation}
where
\begin{equation}
v_{j}^{\#}=\sqrt{\frac{(B_{j}-A_{j})(\rho_{ij}+u_{j}^{B})^{2}-B_{j}\rho
_{ij}^{2}}{\bar{A}_{j}}}-(\rho_{ij}+u_{j}^{B}).
\label{Eq:OP3_FormalPart2VCrit}%
\end{equation}
\end{lemma}

\begin{proof}
Similar to the proof of Lemma \ref{Lm:OP3_FormalPart1}, using \eqref{Eq:OP3ObjectiveSimplified}, first and second order derivatives of $J_H(u_j^B,v_j)$ with respect to $v_j$ can be obtained respectively as $J^{\prime}(v_j)$ and $J^{\prime\prime}(v_j)$, where   
\begin{eqnarray*}
    J^{\prime}(0) &=& 
    \frac{\bar{A}}{2} - \frac{B_j}{2}\left[1-\frac{\rho_{ij}^2}{(\rho_{ij}+u_j^B)^2}\right] \mbox{ and }\\ J^{\prime\prime}(v_j) &=&
    \frac{R_j^2(t)+2B_j\rho_{ij}R_j(t)+A_jB_j\rho_{ij}^2}{(B_j-A_j)(\rho_{ij}+u_j^B+v_j)}.
\end{eqnarray*} 
Note that $J^{\prime\prime}(v_j)>0$ for all $v_j \geq 0$. This implies that $J_H(u_j^B,v_j)$ is convex in the positive orthant of $v_j$, and  $J^{\prime}(v_j)$ is increasing with $v_j \geq 0$ starting from $J^{\prime}(0)$ given above. 

Now, if $J^{\prime}(0) \geq 0$, it implies that $J_H(u_j^B,v_j)$ is monotonically increasing with $v_j \geq 0$. Therefore, for such a case, $v_j^* = 0$ and it proves the first case in \eqref{Eq:OP3_FormalPart2Sol}.

When $J^{\prime}(0) < 0$, there should exist a minimum to $J_H(u_j^B,v_j)$ at some $v_j \geq 0$. Using calculus, the minimizing $v_j$ value can be found easily as $v_j=v_j^{\#}$ given in \eqref{Eq:OP3_FormalPart2VCrit}.

Now, based on the constraint $0\leq v_j\leq\bar{v}_j$ in \eqref{Eq:OP3_FormalPart2Sol} and the convexity of $J_H(u_j^B,v_j)$, it is clear that whenever $v_j^{\#} \leq \bar{v}_j \implies v_j^* = v_j^{\#}$ in \eqref{Eq:OP3_FormalPart2Sol} and whenever $v_j^{\#} > \bar{v}_j \implies v_j^*=\bar{v}_j$. This proves the second case in \eqref{Eq:OP3_FormalPart2Sol}.
\end{proof}

Similar to $u_{j}^{\#}$ given in \eqref{Eq:OP3_FormalPart1UCrit}, $v_{j}^{\#}$ in \eqref{Eq:OP3_FormalPart2VCrit} is completely known to the agent. However, unlike $u_{j}^{\#}$, $v_{j}^{\#}$ represents an optimal choice available for $v_{j}$. Therefore, whenever the constraints on $v_{j}$ in \eqref{Eq:OP3_FormalPart2} (i.e., $0\leq v_{j}\leq\bar{v}_{j}$) allow it, $v_{j}^{\ast}=v_{j}^{\#}$ should be chosen. Moreover, note that similar to Lemma \ref{Lm:OP3_FormalPart1}, Lemma \ref{Lm:OP3_FormalPart2} also holds for time-varying system parameters.

\begin{remark}
The terms $u_{j}^{B}$ and $\bar{v}_{j}$ involved in \eqref{Eq:OP3_FormalPart2Sol} can be simplified (using \eqref{Eq:Lambdaj00} and \eqref{Eq:Constraints1} respectively) to illustrate the state dependent nature of $v_{j}^{\ast}$ as follows:
\begin{equation} 
v_{j}^{\ast} =
\begin{cases}
0 & \mbox{ if }\bar{A}\geq B_{j}\ \mbox{ or }R_{j}(t)\leq\left[  \frac
{\rho_{ij}(B_{j}-A_{j})\sqrt{B_{j}}} {\sqrt{B_{j}-\bar{A}}}-\rho_{ij}
B_{j}\right] \\
v_{j}^{\#} &
\begin{split}
\mbox{ else if }R_{j}(t)  &  \leq\Big[\sqrt{(B_{j}-A_{j})(H^{2}\bar{A}%
_{j}+\rho_{ij}^{2}B_{j})}\\
&  -\rho_{ij}B_{j}\Big]
\end{split}
\\
\bar{v}_{j} & \mbox{ otherwise. }
\end{cases}
\label{Eq:OP3_FormalPart2StateDep}%
\end{equation}
\end{remark}

Cases 1 and 2 discussed above can be combined to yield the following result.

\begin{theorem}
\label{Th:OP3FormalCombined} The optimal solution of \eqref{Eq:OP3_Formal} is
\begin{equation}
\label{Eq:OP3FormalCombined}(u_{j}^{*},v_{j}^{*})=
\begin{cases}
(0,0) & \mbox{ if } u_{j}^{\#} > \bar{u}_{j} \mbox{ or } \bar{A} \geq B_{j}\\
(\bar{u}_{j},0) & \mbox{ else if } \bar{u}_{j} < u_{j}^{B}\\
(u_{j}^{B},0) & \mbox{ else if } B_{j} > \bar{A} \geq B_{j}\left[
1-\frac{\rho_{ij}^{2}}{(\rho_{ij}+u_{j}^{B})^{2}}\right] \\
(u_{j}^{B},v_{j}^{\#}) & \mbox{ else if } v_{j}^{\#} \leq\bar{v}_{j}\\
(u_{j}^{B},\bar{v}_{j}) & \mbox{ otherwise,}
\end{cases}
\end{equation}
where $u_{j}^{\#}$ is given in \eqref{Eq:OP3_FormalPart1UCrit} and $v_{j}^{\#}$ is given in \eqref{Eq:OP3_FormalPart2VCrit}.
\end{theorem}

\emph{Proof: } This result is a composition of the respective solutions given in Lemmas \ref{Lm:OP3_FormalPart1} and \ref{Lm:OP3_FormalPart2}.
\hfill$\blacksquare$

The above theorem implies that whenever: (\romannum{1}) $H$ is sufficiently large (ensuring $\bar{u}_{j}=u_{j}^{B}$ in \eqref{Eq:Lambdaj000}), (\romannum{2}) the sensing capabilities are high enough to ensure $\bar{A}<B_{j}$ and (\romannum{3}) target uncertainty $R_{j}(t)$ exceeds a known threshold (ensuring $u_{j}^{\#}<\bar{u}_{j}=u_{j}^{B}$), it is optimal to select $u_{j}^{\ast}=u_{j}^{B}$, hence planning ahead to drive $R_{j}(t)$ to zero. This conclusion is in line with Theorem 1 in \cite{Zhou2019}.

\paragraph{\textbf{Solving for Optimal Next Destination $j^{\ast}$}}
Using Theorem \ref{Th:OP3FormalCombined}, when agent $a$ is ready to leave target $i$ at some local event time $t$, it can compute the optimal trajectory costs $J_{H}(u_{j}^{\ast},v_{j}^{\ast})$ for all $j\in\mathcal{N}_{i}$. Based on the second step of the RHCP \eqref{Eq:RHCGenSolStep2}, the optimal neighbor to visit next is $j^{\ast}$ where 
\begin{equation}
j^{\ast}=\underset{j\in\mathcal{N}_{i}}{\arg\min}\ J_{H}(u_{j}^{\ast},v_{j}^{\ast}). \label{Eq:OP3_FormalStep2}
\end{equation}
Thus, upon solving \textbf{RHCP3} agent $a$ departs from target $i$ at time $t$ and follows the path $(i,j^{\ast})\in\mathcal{E}$ to visit target $j^{\ast}$. This optimal control will be updated upon the occurrence of the next event, which, in this case, will be the arrival of the agent at target $j^{\ast}$, triggering the solution of an instance of \textbf{RHCP1} at $j^{\ast}$ with the new neighborhood $\mathcal{N}_{j^{\ast}}$.

On a separate note, note that both $J_H(u_j^*,v_j^*)$ and $j^*$ are heavily dependent on $\{R_j(t),\bar{u}_j,\bar{v}_j,A_j,B_j,\rho_{ij}: \forall j\in\mathcal{N}_i\}$ values. Having such a dependence on the neighborhood is crucial as it forces the agent to consider the underlying network topology and monitoring aspects. Appendix B provides a counter example where the same \textbf{RHCP3} have been considered but with a different objective function form (other than \eqref{Eq:RHCNewChoices}). For that case, it is proven (in in Theorem \ref{Eq:OP3CombinedResultOld}) that $J_H(u_j^*,v_j^*)$ is dependent only on $\rho_{ij}$ (see \eqref{Eq:Op3OptimumCostF}) - which leads to unfavorable results (see Theorem \ref{Th:OP3jStarOld}).

\subsection{Solution of \textbf{RHCP2}}

An agent $a$ residing in target $i$ has to solve \textbf{RHCP2} only when an observed event is: (\romannum{1}) $[R_{i}\rightarrow0^{+}]$ or (\romannum{2}) a neighbor induced event $C_{j}$ or $\bar{C}_{j}$, $j\in\mathcal{N}_{i}$, while $R_{i}(t)=0$. Therefore, $u_{i}^{\ast}(t)=0$ by default and $U_{i}(t)=[v_{i},j,u_{j},v_{j}]$ with $U_{ij}=[v_{i},u_{j},v_{j}]$ in \eqref{Eq:RHCGenSolStep1}. Upon solving this \textbf{RHCP2}, the obtained $v_{i}^{\ast}$ defines the remaining inactive time to be spent on target $i$ until the next local event occurs. The variable horizon $w$ defined in \eqref{Eq:VariableHorizon} for this case is $w=v_{i}+\rho_{ij}+u_{j}+v_{j}$.

\paragraph{\textbf{Constraints}}

Following the previously used notation, the maximal possible active time at target $j$ forms an upper-bound to the control variable $u_{j}$ given by 
\begin{equation}
u_{j}^{B}(v_{i})=\frac{R_{j}(t+v_{i}+\rho_{ij})}{B_{j}-A_{j}}=\frac
{R_{j}(t)+A_{j}\rho_{ij}}{B_{j}-A_{j}}+\frac{A_{j}}{B_{j}-A_{j}}\cdot v_{i}.
\label{Eq:Lambdaj01}%
\end{equation}
Note that due to the inclusion of $v_{i}$ in \textbf{RHCP2} (compared to \textbf{RHCP3}), $u_{j}^{B}$ is now dependent on $v_{i}$ (see \eqref{Eq:Lambdaj00}). Based on the same arguments as in the analysis of \textbf{RHCP3}: (\romannum{1}) to spend a positive inactive time at target $j$, the agent has to first spend the maximum active time possible $u_{j}^{B}(v_{i})$, and (\romannum{2}) the variable horizon is subject to $w\leq H$, we now see that any feasible $U_{ij}=[v_{i},u_{j},v_{j}]\in\mathbb{U}$ in \eqref{Eq:RHCGenSolStep1} for \textbf{RHCP2} should belong to one of the two constraint sets:
\begin{equation}
\begin{aligned} 
\mathbb{U}_1 =&\  \{0 \leq v_i \leq \bar{v}_i(u_j,v_j),\  0 \leq u_j \leq \bar{u}_j(v_i),\  v_j = 0 \},\\ \mathbb{U}_2 =&\  \{0 \leq v_i \leq \bar{v}_i(u_j,v_j),\  u_j = u_j^B(v_i),\  0 \leq v_j \leq \bar{v}_j(v_i) \}, \end{aligned} \label{Eq:Constraints2}%
\end{equation}
where, $u_{j}^{B}(v_{i})$ is given in \eqref{Eq:Lambdaj01} and
\begin{align*}
\bar{v}_{i}(u_{j},v_{j})  &  =H-(\rho_{ij}+u_{j}+v_{j}),\\
\bar{u}_{j}(v_{i})  &  =\min\{u_{j}^{B}(v_{i}),\ H-(v_{i}+\rho_{ij})\},\\
\bar{v}_{j}(v_{i})  &  =H-(v_{i}+\rho_{ij}+u_{j}^{B}(v_{i})).
\end{align*}
Similar to \eqref{Eq:Constraints1}, $\bar{u}_{j}$ and $\bar{v}_{j}$ respectively represent the limiting values of active and inactive times at $j$. Along the same lines, $\bar{v}_{i}$ is the upper bound to the inactive time at $i$. However, in contrast to \eqref{Eq:Constraints1}, the aforementioned three quantities are control-dependent in \eqref{Eq:Constraints2}. Moreover, note that in \eqref{Eq:Constraints2}, under $\mathbb{U}_{1}$, $\bar{v}_{i}(u_{j},v_{j})=\bar{v}_{i}(u_{j},0)$ and under $\mathbb{U}_{2}$, $\bar{v}_{i}(u_{j},v_{j})=\bar{v}_{i}(u_{j}^{B}(v_{i}),v_{j})$.

\paragraph{\textbf{Objective}}
Following from \eqref{Eq:RHCNewChoices}, the objective function corresponding to \textbf{RHCP2} is 
$$J_{H}(U_{ij})=J_{H}(X_{i}(t),[0,U_{ij}];H)=\frac{1}%
{w}\bar{J}_{i}(t,t+w).$$
To obtain an explicit expression for $J_{H}(U_{ij})$, $\bar{J}_{i}$ in \eqref{Eq:LocalObjectiveFunction} is decomposed as,
\begin{equation}
\label{Eq:ObjDerivationOP2Step1}\bar{J}_{i}=J_{i}+J_{j}+\sum_{m\in
\mathcal{N}_{i}\backslash\{j\}}J_{m}.
\end{equation}
Considering the feasible state trajectories shown in Fig. \ref{Fig:RHCP2Graphs} for this case, the three terms $J_{i},J_{j}$ and $J_{m}$ in \eqref{Eq:ObjDerivationOP2Step1} are evaluated for the period $[t,t+w)$ using Theorem \ref{Th:Contribution} as
\begin{align*}
    J_i =&
    \frac{A_i(\rho_{ij}+u_j+v_j)^2}{2},\\
    J_j =&
    \frac{(v_i+\rho_{ij})}{2}\left[2R_j(t) + A_j(v_i+\rho_{ij})\right] \\
    &+ \frac{u_j}{2}\left[ 2(R_j(t) + A_j(v_i+\rho_{ij})) - (B_j-A_j)u_j \right],\\
    J_m=& 
    \frac{(v_i+\rho_{ij}+u_j+v_j)}{2} \left[ 2R_m(t) + A_m(v_i+\rho_{ij}+u_j+v_j)\right].
\end{align*}

\begin{figure}[!h]
    \centering
    \includegraphics[width=3.3in]{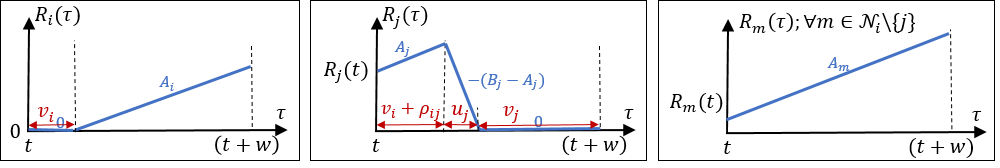}
    \caption{State trajectories of targets in $\bar{\mathcal{N}}_i$ during $[t,t+w)$.}
    \label{Fig:RHCP2Graphs}
\end{figure}

Combining these results and substituting them in \eqref{Eq:ObjDerivationOP2Step1} gives the complete objective function $J_{H}(U_{ij})$ as
\begin{equation}%
\begin{split}
J_{H}(v_{i},u_{j},v_{j})=  &  \frac{1}{v_{i}+\rho_{ij}+u_{j}+v_{j}}[C_{1}%
v_{i}^{2}+C_{2}u_{j}^{2}+C_{3}v_{j}^{2}\\
&  +C_{4}v_{i}u_{j}+C_{5}v_{i}v_{j}+C_{6}u_{j}v_{j}+C_{7}v_{i}\\
&  +C_{8}u_{j}+C_{9}v_{j}+C_{10}],
\end{split}
\label{Eq:OP2ObjectiveSimplified}%
\end{equation}
where 
\begin{equation*}
\begin{aligned}
C_{1} &= \frac{\bar{A}_i}{2},\ \ 
C_{2} = \frac{\bar{A}-B_j}{2},\ \ 
C_{3} = \frac{\bar{A}_j}{2},\ \ 
C_{4} = \bar{A}_i,\ \ 
C_{5} = \bar{A}_{ij},\\
C_{6} &= \bar{A}_j,\ \
C_{7} = \left[\bar{R}_i(t)+\bar{A}_i\rho_{ij}\right],\ \ 
C_{8} = \left[\bar{R}_i(t)+\bar{A}\rho_{ij}\right],\\ 
C_{9} &= \left[\bar{R}_{ij}(t) + \bar{A}_j\rho_{ij}\right] \mbox{ and } 
C_{10} = \frac{\rho_{ij}}{2}\left[2\bar{R}_{i}(t) + \bar{A}\rho_{ij}\right].
\end{aligned}
\end{equation*}
The remaining parameters are same as \eqref{Eq:NeighborhoodParameters}. Note that $C_{i}\geq0$ for all $i$ except for $C_{2}$, where $C_{2}\geq 0\iff\bar{A}\geq B_{j}$.

\paragraph{\textbf{Solving the \textbf{RHCP2} for Optimal Control }$(v_{i}^{*}, u_{j}^{*},v_{j}^{*})$}

Based on the first step of \eqref{Eq:RHCGenSolStep1}, the optimal controls for \textbf{RHCP2} are determined by
\begin{equation}
(v_{i}^{\ast},u_{j}^{\ast},v_{j}^{\ast})=\ \underset{(v_{i},u_{j},v_{j})}%
{\arg\min}\ {J_{H}(v_{i},u_{j},v_{j}),} \label{Eq:OP2_Formal}%
\end{equation}
where $(v_{i},u_{j},v_{j})\in\mathbb{U}_{1}$ or $(v_{i},u_{j},v_{j}%
)\in\mathbb{U}_{2}$ as in \eqref{Eq:Constraints2}.

\paragraph*{\textbf{- Case 1}}

$(v_{i},u_{j},v_{j})\in\mathbb{U}_{1}$ in \eqref{Eq:Constraints2}: Then, $v_{j}^{\ast}=0$ and \eqref{Eq:OP2_Formal} takes the form: 
\begin{equation}
\begin{aligned} (v_i^*,u_j^*) =\  &\underset{(v_i,u_j)}{\arg\min}\ {J_H(v_i,u_j,0)}\\ 
&v_i \geq 0,\ \ \ 0 \leq u_j \leq u_j^B(v_i),\\ 
&v_i + u_j \leq H-\rho_{ij}.\\ \end{aligned} \label{Eq:OP2_FormalPart1}%
\end{equation}
The above constraints follow from \eqref{Eq:Constraints2} and the relationships:
\begin{align*}
&  v_{i}\leq\bar{v}_{i}(u_{j},0)\iff v_{i}\leq H-(\rho_{ij}+u_{j}%
)\mbox{ and }\\
&  u_{j}\leq\bar{u}_{j}(v_{i})\iff u_{j}\leq u_{j}^{B}(v_{i})\And u_{j}\leq
H-(v_{i}+\rho_{ij}).
\end{align*}
Note that $u_{j}^{B}(v_{i})$ is linear in $v_{i}$ (see \eqref{Eq:Lambdaj01}). Prior to presenting the approach for solving \eqref{Eq:OP2_FormalPart1}, let us also formulate the second sub-problem of \eqref{Eq:OP2_Formal} based on $\mathbb{U}_{2}$ in \eqref{Eq:Constraints2}.

\paragraph*{\textbf{- Case 2}}

$(v_{i},u_{j},v_{j})\in\mathbb{U}_{2}$ in \eqref{Eq:Constraints2}: Then, $u_{j}^{\ast}=u_{j}^{B}(v_{i}^{\ast})$ and \eqref{Eq:OP2_Formal} takes the form:
\begin{equation}
\begin{aligned} (v_i^*,v_j^*) =\  &\underset{(v_i,v_j)}{\arg\min}\ {J_H(v_i,u_j^B(v_i),v_j)}\\ 
&v_i \geq 0,\ \ \ v_j \geq 0, \\ 
&v_i + u_j^B(v_i) + v_j \leq H-\rho_{ij}. \end{aligned} \label{Eq:OP2_FormalPart2}%
\end{equation}
The constraints in \eqref{Eq:OP2_FormalPart2} follow from
\eqref{Eq:Constraints2} and the relationships:
\begin{align*}
&  v_{i}\leq\bar{v}_{i}(u_{j}^{B}(v_{i}),v_{j})\iff v_{i}\leq H-(\rho
_{ij}+u_{j}^{B}(v_{i})+v_{j})\mbox{ and }\\
&  v_{j}\leq\bar{v}_{j}(v_{i})\iff v_{j}\leq H-(v_{i}+\rho_{ij}+u_{j}%
^{B}(v_{i})).
\end{align*}

\paragraph*{\textbf{- Combined Result}}

The two optimization problems \eqref{Eq:OP2_FormalPart1} and \eqref{Eq:OP2_FormalPart2} belong to the class of \emph{constrained bi-variate rational function optimization problems} (RFOPs) in \eqref{Eq:ConstrainedOptimizationBivariate} discussed in Appendix A. Specifically, Appendix A presents a computationally efficient, analytical procedure for obtaining the globally optimal solution of such RFOPs. As we will see in the rest of this paper, all remaining problems we need to solve belong to the class of RFOPs.

To provide details, note that \eqref{Eq:OP2_FormalPart1} is a special case of \eqref{Eq:ConstrainedOptimizationBivariate} where
\begin{equation}
\begin{aligned}
    x =& v_i,\ \ y = u_j,\ \ H(x,y) = J_H(v_i,u_j,0),\\
    \mb{P} =& \frac{A_j}{B_j-A_j},\ \  
    \mb{L} = \frac{R_j(t)+A_j\rho_{ij}}{B_j-A_j},\ \ 
    \mb{Q} = 1,\\  
    \mb{M} =& H-\rho_{ij},\ \ 
    \mb{N} = \infty. 
\end{aligned}
\end{equation}
Similarly, \eqref{Eq:OP2_FormalPart2} is also a special case of \eqref{Eq:ConstrainedOptimizationBivariate} where \begin{equation}
\begin{aligned}
    x =& v_i,\ \ y = v_j,\ \ H(x,y) = J_H(v_i,u_j^B(v_i),v_j),\\
    \mb{P} =& 0,\ \  
    \mb{L} = \infty,\ \ 
    \mb{Q} = \frac{B_j}{B_j-A_j},\\  
    \mb{M} =& H-\rho_{ij}-\frac{R_j(t)+A_j\rho_{ij}}{B_j-A_j},\ \ 
    \mb{N} = \infty. 
\end{aligned}
\end{equation}

Upon individually obtaining solutions to \eqref{Eq:OP2_FormalPart1} and \eqref{Eq:OP2_FormalPart2}, the main optimization problem \eqref{Eq:OP2_Formal} is solved by simply comparing the objective function values of those individual solutions.

\paragraph{\textbf{Solving for Optimal (Planned) Next Destination $j^{*}$}}

The second step of \textbf{RHCP2} (i.e., \eqref{Eq:RHCGenSolStep2}) is to choose the optimal neighbor $j\in\mathcal{N}_{i}$ according to 
\begin{equation}
j^{\ast}=\underset{j\in\mathcal{N}_{i}}{\arg\min}\ J_{H}(v_{i}^{\ast}%
,u_{j}^{\ast},v_{j}^{\ast}). \label{Eq:OP2_FormalStep2Rev}%
\end{equation}
This step requires the objective value of the optimal solution $U_{ij}^{\ast}=[v_{i}^{\ast},u_{j}^{\ast},v_{j}^{\ast}]$ obtained for each $j\in\mathcal{N}_{i}$ (in \eqref{Eq:OP2_Formal}). Now, $v_{i}^{\ast}$ taken from $U_{ij^{\ast}}^{\ast}$ defines the inactive time that the agent should spend at current target $i$ starting from the current time $t$ until the next local event occurs. This next event is either: (\romannum{1}) $[h \rightarrow v_{i}^{\ast}]$ or (\romannum{2}) a neighbor-induced $C_{j}$ or $\bar{C}_{j}$ event for some $j\in\mathcal{N}_{i}$. Depending on this event, the agent will have to subsequently solve an instance of \textbf{RHCP3} or \textbf{RHCP2} respectively.

\subsection{Solution of \textbf{RHCP1}}

An agent $a$ residing at target $i$ has to solve \textbf{RHCP1} only when an observed event is: (\romannum{1}) agent $a$'s arrival at target $i$, or (\romannum{2}) a neighbor induced event (i.e., $C_{j}$ or $\bar{C}_{j}$ for some $j\in\mathcal{N}_{i}$) while $R_{i}(t)>0$. Therefore, \textbf{RHCP1} involves all the decision variables $U_{ij}\triangleq\lbrack u_{i},v_{i},u_{j},v_{j}]$ and $j$ included in \eqref{Eq:RHCNewChoices}. Upon solving \textbf{RHCP1}, the obtained $u_{i}^{\ast}$ gives the active time remaining to be spent at target $i$ - until the next local event occurs. The variable horizon $w$ for this case is $w=u_{i}+v_{i}+\rho_{ij}+u_{j}+v_{j}$, as in \eqref{Eq:VariableHorizon}.

\paragraph{\textbf{Constraints}}

Following the same notation as before, the maximum possible active times at targets $i$ and $j$ respectively are $u_{i}^{B}$ and $u_{j}^{B}$ (omitting the dependence on time $t$) where
\begin{equation}
\begin{aligned} u_i^B =& \frac{R_i(t)}{B_i-A_i} \mbox{ and }\\ u_j^B(u_i,v_i) =& \frac{R_j(t)+A_j\rho_{ij}}{B_j-A_j} + \frac{A_j}{B_j-A_j}\cdot(u_i+v_i). \end{aligned} \label{Eq:Lambdaj02}%
\end{equation}
Note that $u_{j}^{B}$ is now dependent on the control variables $u_{i}$ and $v_{i}$. Based on the same arguments as in the analysis of \textbf{RHCP2}: (\romannum{1}) to spend a positive inactive time at any target, the agent should spend the maximum possible active time at that target, and (\romannum{2}) the variable horizon is subject to $w\leq H$. Thus, any feasible $(u_{i},v_{i},u_{j},v_{j})$ in \eqref{Eq:RHCGenSolStep1} belongs to one of the four constraint set pairs: $(\mathbb{U}_{ik},\mathbb{U}_{jl}),$ $k,l\in\{1,2\}$ where
\begin{equation}
\begin{aligned} \mathbb{U}_{i1} =& \{0 \leq u_i \leq \bar{u}_i(u_j,v_j),\  v_i = 0\},\\ \mathbb{U}_{i2} =& \{u_i = u_i^B,\  0 \leq v_i \leq \bar{v}_i(u_j,v_j)\},\\ \mathbb{U}_{j1} =& \{ 0 \leq u_j \leq \bar{u}_j(u_i,v_i),\ v_j = 0\},\\ \mathbb{U}_{j2} =& \{u_j = u_j^B(u_i,v_i),\  0 \leq v_j \leq \bar{v}_j(u_i,v_i)\}, \end{aligned} \label{Eq:Constraints3}%
\end{equation}
with $u_{i}^{B}$ and $u_{j}^{B}(u_{i},v_{i})$ given in \eqref{Eq:Lambdaj02} and
\begin{align*}
\bar{u}_{i}(u_{j},v_{j})  &  =\min\{u_{i}^{B},\ H-(\rho_{ij}+u_{j}+v_{j})\},\\
\bar{v}_{i}(u_{j},v_{j})  &  =H-(u_{i}^{B}+\rho_{ij}+u_{j}+v_{j}),\\
\bar{u}_{j}(u_{i},v_{i})  &  =\min\{u_{j}^{B}(u_{i},v_{i}),\ H-(u_{i}%
+v_{i}+\rho_{ij})\},\\
\bar{v}_{j}(u_{i},v_{i})  &  =H-(u_{i}+v_{i}+\rho_{ij}+u_{j}^{B}(u_{i}%
,v_{i})).
\end{align*}
These are similar to \eqref{Eq:Constraints1} and \eqref{Eq:Constraints2}, but, unlike \eqref{Eq:Constraints1} or \eqref{Eq:Constraints2}, each of these four limiting values are now dependent on two control decisions.

\paragraph{\textbf{Objective}}

According to \eqref{Eq:RHCNewChoices}, the objective function corresponding to \textbf{RHCP1} is 
$$J_{H}(U_{ij}) = J_{H}(X_{i}(t),U_{ij};H) = \frac{1}{w}\bar{J}_{i}(t,t+w).$$
To obtain an explicit expression for $J_{H}(U_{ij})$, $\bar{J}_{i}$ in \eqref{Eq:LocalObjectiveFunction} is decomposed as in \eqref{Eq:ObjDerivationOP2Step1} and the three terms $J_{i},\,J_{j}$ and $J_{m}$ are evaluated for trajectories such that the agent moves from target $i$ to $j$ following decisions $[u_{i},v_{i},u_{j},v_{j}]$ during the period $[t,t+w)$. With the aid of Fig. \ref{Fig:RHCP1Graphs} and Theorem \ref{Th:Contribution} we obtain:
\begin{align*}
    J_i=& \frac{u_i}{2}\left[2R_i(t)-(B_i-A_i)u_i\right]
    +\frac{(\rho_{ij}+u_j+v_j)}{2}\\
    &\times \left[2(R_i(t)-(B_i-A_i)u_i)+A_i(\rho_{ij}+u_j+v_j)\right],\\
    J_j=& \frac{(u_i+v_i+\rho_{ij})}{2}\left[2R_j(t) + A_j(u_i+v_i+\rho_{ij})\right] \\
    &+ \frac{u_j}{2}\left[ 2(R_j(t) + A_j(u_i+v_i+\rho_{ij}))-(B_j-A_j)u_j\right],\\
    J_m=& \frac{(u_i+v_i+\rho_{ij}+u_j+v_j)}{2} \left[ 2R_m(t)\right. \\
    &+A_m(u_i+v_i+\rho_{ij}+u_j+v_j)\left.\right].
\end{align*}

\begin{figure}[!h]
    \centering
    \includegraphics[width=3.3in]{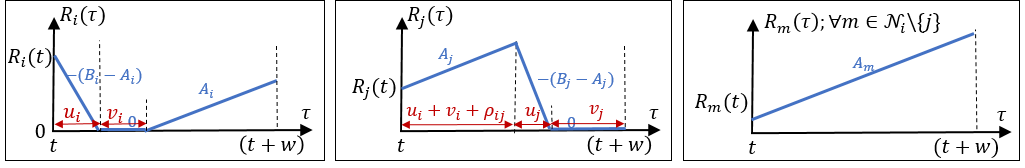}
    \caption{State trajectories of targets in $\bar{\mathcal{N}}_i$ during $[t,t+w)$.}
    \label{Fig:RHCP1Graphs}
\end{figure}

Combining the above three results and substituting it in \eqref{Eq:ObjDerivationOP2Step1} gives the complete objective function $J_H(U_{ij})$ as \begin{equation}\label{Eq:OP1ObjectiveSimplified}
\begin{split}
&J_H(u_i,v_i,u_j,v_j) = \hfill \\
&\frac{\left[
\begin{split}
C_1u_i^2 + C_2v_i^2 + C_3u_j^2 + C_4v_j^2 
+ C_5u_iv_i\\ + C_6u_iu_j + C_7u_iv_j
+ C_8v_iu_j + C_9v_iv_j + C_{10}u_jv_j\\ + C_{11}u_i 
+ C_{12}v_i + C_{13}u_j + C_{14}v_j + C_{15} 
\end{split}
\right]}{u_i + v_i+\rho_{ij}+u_j+v_j},
\end{split}
\end{equation}
where
\begin{equation*}
\begin{aligned}
C_{1} &= \frac{\bar{A}-B_i}{2},\ \ 
C_{2} = \frac{\bar{A}_i}{2},\ \ 
C_{3} = \frac{\bar{A}-B_j}{2},\ \ 
C_{4} = \frac{\bar{A}_j}{2},\\
C_{5} &= \bar{A}_{i},\ \ 
C_{6} = \bar{A}-B_i,\ \
C_{7} = \bar{A}_j-B_i,\ \ 
C_{8} = \bar{A}_i,\\
C_{9} &= \bar{A}_{ij},\ \ 
C_{10} = \bar{A}_j,\ \  
C_{11} = \left[\bar{R}(t)+(\bar{A}-B_i)\rho_{ij}\right],\\ 
C_{12} &= \left[\bar{R}_i(t)+\bar{A}_i\rho_{ij}\right],\ \
C_{13} = \left[\bar{R}(t)+\bar{A}\rho_{ij}\right],\\
C_{14} &= \left[\bar{R}_j(t)+\bar{A}_j\rho_{ij}\right] \mbox{ and }\ \ 
C_{15} = \frac{\rho_{ij}}{2}\left[2\bar{R}(t) + \bar{A}\rho_{ij}\right].
\end{aligned}
\end{equation*}
The remaining parameters are same as \eqref{Eq:NeighborhoodParameters}. Note that all the coefficients stated above are non-negative except for $C_1,C_3,C_6,C_7$ and $C_{11}$.

\paragraph{\textbf{Solving the \textbf{RHCP1} for Optimal Control } $(u_{i}^{*}, v_{i}^{*}, u_{j}^{*},v_{j}^{*})$}

The first step of \textbf{RHCP1} \eqref{Eq:RHCGenSolStep1} can be stated using \eqref{Eq:OP1ObjectiveSimplified} as
\begin{equation}
(u_{i}^{\ast},v_{i}^{\ast},u_{j}^{\ast},v_{j}^{\ast})=\ \underset{(u_{i}%
,v_{i},u_{j},v_{j})}{\arg\min}\ {J_{H}(u_{i},v_{i},u_{j},v_{j})},
\label{Eq:OP1_Formal}%
\end{equation}
where $(u_{i},v_{i})\in\mathbb{U}_{ik},\ (u_{j},v_{j})\in\mathbb{U}_{jl}$ for $k,l\in\{1,2\}$ as in \eqref{Eq:Constraints3}. There are four different cases for this problem depending on which of the four constraint set pairs in \eqref{Eq:Constraints3} is used.

\paragraph*{\textbf{- Case 1}}
$(u_{i},v_{i})\in\mathbb{U}_{i1},\ (u_{j},v_{j})\in\mathbb{U}_{j1}$ in \eqref{Eq:Constraints3}: Then, $v_{i}^{\ast}=0,\ v_{j}^{\ast}=0$ and \eqref{Eq:OP1_Formal} takes the form:
\begin{equation}
\begin{aligned} (u_i^*,u_j^*) =\  &\underset{(u_i,u_j)}{\arg\min}\ {J_H(u_i,0,u_j,0)}\\ 
&0 \leq u_i \leq u_i^B,\ \ \ 
0 \leq u_j \leq u_j^B(u_i,0),\\
&u_i + u_j \leq H - \rho_{ij}. 
\end{aligned} \label{Eq:OP1_FormalPart1}
\end{equation}
The above constraints follow from \eqref{Eq:Constraints3} and the relationships:
\begin{align*}
&  u_{i}\leq\bar{u}_{i}(u_{j},0)\iff u_{i}\leq u_{i}^{B}\And u_{i}\leq
H-(\rho_{ij}+u_{j})\mbox{ and }\\
&  u_{j}\leq\bar{u}_{j}(u_{i},0)\iff u_{j}\leq u_{j}^{B}(u_{i},0)\And
u_{j}\leq H-(u_{i}+\rho_{ij}).
\end{align*}
Note that $u_{j}^{B}(u_{i},0)$ is linear and increasing with $u_{i}$ (see \eqref{Eq:Lambdaj02}). Similar to before, prior to presenting the approach for solving \eqref{Eq:OP1_FormalPart1}, let us also formulate the remaining sub-problems of \eqref{Eq:OP1_Formal}.

\paragraph*{\textbf{- Case 2}}

$(u_{i},v_{i}) \in\mathbb{U}_{i1},\ (u_{j},v_{j}) \in\mathbb{U}_{j2}$ in \eqref{Eq:Constraints3}: Then, $v_{i}^{*} = 0,\ u_{j}^{*} = u_{j}^{B}(u_{i}^{*},0)$ and \eqref{Eq:OP1_Formal} takes the form:
\begin{equation}
\label{Eq:OP1_FormalPart2}\begin{aligned} (u_i^*,v_j^*) =\  &\underset{(u_i,v_j)}{\arg\min}\ {J_H(u_i,0,u_j^B(u_i,0),v_j)}\\ 
&0 \leq u_i \leq u_i^B,\ \ \ 
v_j \geq 0,\\ 
&u_i+u_j^B(u_i,0)+v_j \leq H-\rho_{ij}. \end{aligned}
\end{equation}
The constraints in \eqref{Eq:OP1_FormalPart2} are from \eqref{Eq:Constraints3} and the relationships:
\begin{align*}
&u_i \leq \bar{u}_i(u_j^B(u_i,0),v_j) \iff \\
&u_i \leq u_i^B \And u_i \leq   H-(\rho_{ij}+u_j^B(u_i,0)+v_j) \mbox{ and }\\
&v_j \leq \bar{v}_j(u_i,0) \iff v_j \leq H-(u_i+\rho_{ij}+u_j^B(u_i,0)).
\end{align*}

\paragraph*{\textbf{- Case 3}}

$(u_{i},v_{i}) \in\mathbb{U}_{i2},\ (u_{j},v_{j}) \in\mathbb{U}_{j1}$ in \eqref{Eq:Constraints3}: Then, $u_{i}^{*} = u_{i}^{B},\ v_{j}^{*} = 0$ and \eqref{Eq:OP1_Formal} takes the form:
\begin{equation}
\label{Eq:OP1_FormalPart3}\begin{aligned} (v_i^*,u_j^*) =\  &\underset{(v_i,u_j)}{\arg\min}\ {J_H(u_i^B,v_i,u_j,0)}\\
&v_i \geq 0,\ \ \ 
0 \leq u_j \leq u_j^B(u_i^B,v_i),\\ 
&v_i+u_j \leq H-(u_i^B+\rho_{ij}). 
\end{aligned}
\end{equation}
The constraints in \eqref{Eq:OP1_FormalPart3}, are from \eqref{Eq:Constraints3} and the relationships:
\begin{align*}
&v_i \leq \bar{v}_i(u_j,0) \iff v_i \leq  H-(u_i^B+\rho_{ij}+u_j) \mbox{ and }\\
&u_j \leq \bar{u}_j(u_i^B,v_i) \iff \\
&u_j \leq u_j^B(u_i^B,v_i) \And 
  u_j \leq H-(u_i^B+v_i+\rho_{ij}).
\end{align*} 
Note that $u_{j}^{B}(u_{i}^{B},v_{i})$ is linear and increasing with $v_{i}$ \eqref{Eq:Lambdaj02}.

\paragraph*{\textbf{- Case 4}}

$(u_{i},v_{i}) \in\mathbb{U}_{i2},\ (u_{j},v_{j})\in\mathbb{U}_{j2}$ in \eqref{Eq:Constraints3}: Then, $u_{i}^{*} = u_{i}^{B},\ u_{j}^{*} = u_{j}^{B}(u_{i}^{B},v_{i}^{*})$ and \eqref{Eq:OP1_Formal} takes the form:
\begin{equation}
\label{Eq:OP1_FormalPart4}\begin{aligned} (v_i^*,v_j^*) =\  &\underset{(v_i,v_j)}{\arg\min}\ {J_H(u_i^B,v_i,u_j^B(u_i^B,v_i),v_j)}\\ 
&v_i \geq 0, \ \ \ 
v_j \geq 0, \\ 
&v_i + v_j + u_j^B(u_i^B,v_i) \leq H-(u_i^B+\rho_{ij}). 
\end{aligned}
\end{equation}
To write the last constraint in \eqref{Eq:OP1_FormalPart4}, \eqref{Eq:Constraints3} and the relationships:
\begin{align*}
v_i \leq&\ \bar{v}_i(u_j^B(u_i^B,v_i),v_j) \iff \\ 
&v_i \leq H-(u_i^B+\rho_{ij}+u_j^B(u_i^B,v_i)+v_j) \mbox{ and }\\
v_j \leq&\ \bar{v}_j(u_i^B,v_i) \iff v_j \leq H-(u_i^B+v_i+\rho_{ij}+u_j^B(u_i^B,v_i)),
\end{align*} 
have been used. 

\paragraph*{\textbf{- Combined Result}}

The optimization problems \eqref{Eq:OP1_FormalPart1}, \eqref{Eq:OP1_FormalPart2}, \eqref{Eq:OP1_FormalPart3} and \eqref{Eq:OP1_FormalPart4} belong to the same class of RFOPs in \eqref{Eq:ConstrainedOptimizationBivariate} discussed in Appendix A (similar to \eqref{Eq:OP2_FormalPart1} and \eqref{Eq:OP2_FormalPart2}). Therefore, each of these four problems are solved exploiting the computationally efficient, analytical, and globally optimal solution presented in Appendix A. 

To provide details, note that \eqref{Eq:OP1_FormalPart1}-\eqref{Eq:OP1_FormalPart4} conforms to \eqref{Eq:ConstrainedOptimizationBivariate} using following four sets of mappings respectively: \newline (\romannum{1}) Case 1:
\begin{equation}
\begin{aligned}
    x =& u_i,\ \ y = u_j,\ \ H(x,y) = J_H(u_i,0,u_j,0),\\
    \mb{P} =& \frac{A_j}{B_j-A_j},\ \ 
    \mb{L} = \frac{R_j(t)+A_j\rho_{ij}}{B_j-A_j},\ \ 
    \mb{Q} = 1,\\  
    \mb{M} =& H-\rho_{ij},\ \ 
    \mb{N} = \frac{R_i(t)}{B_i-A_i}, 
\end{aligned}
\end{equation}
(\romannum{2}) Case 2:
\begin{equation}
\begin{aligned}
    x =& u_i,\ \ y = v_j,\ \ H(x,y) = J_H(u_i,0,u_j^B(u_i,0),v_j),\\
    \mb{P} =& 0,\ \ 
    \mb{L} = \infty,\ \
    \mb{Q} = \frac{B_j}{B_j-A_j},\\  
    \mb{M} =& H-\rho_{ij}-\frac{R_j(t)+A_j\rho_{ij}}{B_j-A_j},\ \
    \mb{N} = \frac{R_i(t)}{B_i-A_i},
\end{aligned}
\end{equation}
(\romannum{3}) Case 3:
\begin{equation}
\begin{aligned}
    x =& v_i,\ \ y = u_j,\ \ H(x,y) = J_H(u_i^B,v_i,u_j,0),\\
    \mb{P} =& \frac{A_j}{B_j-A_j},\ \ 
    \mb{L} = \frac{R_j(t)+A_j\rho_{ij}}{B_j-A_j} + \frac{A_j}{B_j-A_j}\times\frac{R_i(t)}{B_i-A_i},\\
    \mb{Q} =& 1,\ \  
    \mb{M} = H-\rho_{ij}-\frac{R_i(t)}{B_i-A_i},\ \ 
    \mb{N} = \infty, 
\end{aligned}
\end{equation}
(\romannum{4}) Case 4:
\begin{equation}
\begin{aligned}
    x =& v_i,\ \ y = v_j,\ \ H(x,y) = J_H(u_i^B,v_i,u_j^B(u_i^B,v_i),v_j),\\
    \mb{P} =& 0,\ \ 
    \mb{L} = \infty,\ \
    \mb{Q} = \frac{B_j}{B_j-A_j},\ \ 
    \mb{N} = \infty,\\  
    \mb{M} =& H-\rho_{ij}-\frac{R_j(t)+A_j\rho_{ij}}{B_j-A_j}-\frac{B_j}{B_j-A_j} \times\frac{R_i(t)}{B_i-A_i}.
\end{aligned}
\end{equation}

Upon obtaining the individual solutions to \eqref{Eq:OP1_FormalPart1}-\eqref{Eq:OP1_FormalPart4}, the main optimization problem \eqref{Eq:OP1_Formal} is solved by simply comparing the obtained optimal objective function values.

\paragraph{\textbf{Solving for Optimal (Planned) Next Destination $j^{*}$}}

The second step of \textbf{RHCP1}, same as \eqref{Eq:RHCGenSolStep2}, is to choose the optimal next neighbor $j$ according to 
\begin{equation}
j^{\ast}=\underset{j\in\mathcal{N}_{i}}{\arg\min}\ J_{H}(u_{i}^{\ast}%
,v_{i}^{\ast},u_{j}^{\ast},v_{j}^{\ast}). \label{Eq:OP1_FormalStep2Rev}%
\end{equation}
This step requires the objective values corresponding to the optimal solutions $U_{ij}^{\ast}=[u_{i}^{\ast},v_{i}^{\ast},u_{j}^{\ast},v_{j}^{\ast}]$ for each $j\in\mathcal{N}_{i}$ in \eqref{Eq:OP1_Formal}. Finally, recall that $u_{i}^{\ast}$ within the optimal solution $U_{ij^{\ast}}^{\ast}$ defines the active time that the agent should spend at current target $i$ until the next local event occurs. This next event is either (\romannum{1}) $[h \rightarrow u_{i}^{\ast}]$ with $R_{i}(t+u_{i}^{\ast})>0$, (\romannum{2}) $[h\rightarrow u_{i}^{\ast}]$ with $R_{i}(t+u_{i}^{\ast})=0$, or (\romannum{3}) a neighbor induced $C_{j}$ or $\bar{C}_{j}$ event for some $j\in\mathcal{N}_{i}$ (while $R_{i}>0)$. Therefore, the agent will have to subsequently solve an instance of  \textbf{RHCP3}, \textbf{RHCP2} or \textbf{RHCP1} respectively.

\section{Controller Enhancements}
\label{Sec:Improvements}

There are two reasons why our distributed RHC approach cannot guarantee a global minimum of \eqref{Eq:MainObjective}. First, in order to operate in distributed fashion, we have omitted the future cost estimate term $\hat{J}_{H}(X_{i}(t+H))$ in \eqref{Eq:RHCProblem} and have defined a local objective function \eqref{Eq:LocalObjectiveFunction} for an agent at target $i$ which reflects the structure of \eqref{Eq:MainObjective} limited to the neighborhood of target $i$. In \eqref{Eq:LocalObjectiveFunction}, all neighboring target states are equally weighted which does not take into account the specific neighboring topology. In particular, an optimal next-visit target $j^{\ast}$ determined by \textbf{RHCP3} favors neighbors with smaller $\rho_{ij}$ and $R_{j}(t)$ values. This can be alleviated by adopting different weights in the targets $j\in\mathcal{N}_{i}$ included in \eqref{Eq:LocalObjectiveFunction}.

The second reason also stems from the distributed nature of the RHC, as we have limited the information available to an agent located at target $i$ to its neighborhood $\bar{\mathcal{N}}_{i}$. One can expect that performance can be improved by considering an extended neighborhood whereby additional information may become available to the agent.

In this section, we address these two issues. Specifically, we focus on improving the formulation of \textbf{RHCP3} as it involves the crucial next-visit decision $j^{\ast}$, which highly affects the agent trajectories.

\subsection{Using a Weighted Local Objective in \textbf{RHCP3}}
\label{SubSec:EDRHCAlphaMethod}

We generalize the local objective function decomposition in \eqref{Eq:ObjDerivationOP3Step1} by introducing a weighted version of it as follows:
\begin{equation}
\bar{J}_{i}=\alpha J_{j}+(1-\alpha)\sum_{m\in\bar{\mathcal{N}}_{i}%
\backslash\{j\}}J_{m}, \label{Eq:ModifiedObjectiveForRHCP3}%
\end{equation}
where $\alpha\in\lbrack0,1]$. This approach is more effective as it can emphasize the contribution to the global cost by the \textquotedblleft neglected neighbor targets\textquotedblright\ $m\in\bar{\mathcal{N}}_{i}\backslash\{j\}$ due to the choice of target $j\in\mathcal{N}_{i}$ as the next visit. We will refer to the modified RHC approach using \eqref{Eq:ModifiedObjectiveForRHCP3} as the \textquotedblleft RHC$^{\alpha}$ method\textquotedblright. It should be noted that this modification has no significant effect on the theoretical results of the previous sections.

Regarding desirable values of the weight $\alpha$, we have found that $\alpha<0.5$ is preferred. In order to extend the parameter-free nature of the original RHC method to this RHC$^{\alpha}$ method, we will use $\alpha=\frac{1}{|\bar{\mathcal{N}}_{i}|^{2}}$ as a nominal choice, so as to reduce the importance of target $j$ based on the size of the neighborhood of target $i$.

\begin{lemma}
\label{Lm:ED-RHCPalpha} 
If $\alpha=0$ is used in RHC$^{\alpha}$, the optimal next-visit target $j^{\ast}$ given by \textbf{RHCP3} (i.e., the solution to \eqref{Eq:RHCGenSolStep2}) is 
\begin{equation}
j^{\ast}=\underset{j\in\bar{\mathcal{N}}_{i}}{\arg\min}\left[  (2\bar
{R}(t)+\bar{A}\rho_{ij})-(2R_{j}(t)+A_{j}\rho_{ij})\right]  .
\label{Eq:ED-RHCPalpha}%
\end{equation}
\end{lemma}

\emph{Proof: } Recall that $U_{ij}=[u_{j},v_{j}]$ and $w=\rho_{ij}+u_{j}+v_{j}$ for \textbf{RHCP3}. Applying $\alpha=0$ in \eqref{Eq:ModifiedObjectiveForRHCP3} and using it in the \textbf{RHCP3} objective $J_{H}(U_{ij})=\frac{1}{w}\,\bar{J}_{i}(t,t+w)$ gives $J_{H}(U_{ij})=\bar{R}_{j}(t)+\frac{1}{2}\bar{A}_{j}(\rho_{ij}+u_{j}+v_{j})$. Therefore, clearly the minimizing $U_{ij}$ choice is $u_{j}^{\ast}=v_{j}^{\ast}=0$ (i.e., the solution to \eqref{Eq:RHCGenSolStep1}). Hence, the optimal next-visit target $j^{\ast}$ following from \eqref{Eq:RHCGenSolStep2} is $j\in\mathcal{N}_{i}$ with the minimum $\bar{R}_{j}(t)+\frac{1}{2}\bar{A}_{j}\rho_{ij}$ value. Using the relationships $\bar{R}_{j}=\bar{R}-R_{j}$ and $\bar{A}_{j}=\bar{A}-A_{j}$ (see \eqref{Eq:NeighborhoodParameters}), this $j^{\ast}$ choice yields \eqref{Eq:ED-RHCPalpha}. \hfill$\blacksquare$

Note that the first and second terms in \eqref{Eq:ED-RHCPalpha} approximate the contribution to the main objective \eqref{Eq:MainObjective} during the transit time $\rho_{ij}$ of the neighborhood and of target $j$ respectively. This is an important result (even if it is valid only under $\alpha=0$) as it provides a direct, simple and intuitive approach to obtain the next-visit target decision $j^{\ast}$ (skipping \eqref{Eq:RHCGenSolStep1}) for \textbf{RHCP3}.

Finally, to provide some intuition regarding how the choice of $\alpha$ affects the PMN problem performance, Fig \ref{Fig:AlphaVariation} shows examples of how performance varies with $\alpha$ in two specific PMN problems. We can see that $\alpha=0$ is sometimes directly the optimal choice while in other cases there may be a particular $\alpha<0.5$ which provides the optimal performance.

\begin{figure}[t]
\centering
\begin{subfigure}[b]{0.48\columnwidth}
\centering
\includegraphics[width = 1.5in]{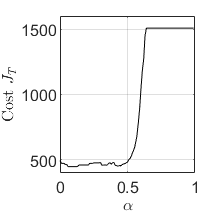}
\caption{Single-Agent Case in Fig. \ref{Fig:SASE4}}
\end{subfigure}
\hfill\begin{subfigure}[b]{0.48\columnwidth}
\centering
\includegraphics[width = 1.5in]{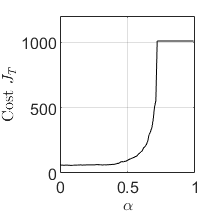}
\caption{Multi-Agent Case in Fig. \ref{Fig:MASE4}}
\end{subfigure}
\caption{Variation of $J_{T}$ in \eqref{Eq:MainObjective} with $\alpha$ in
\eqref{Eq:ModifiedObjectiveForRHCP3}.}%
\label{Fig:AlphaVariation}%
\end{figure}

\subsection{Extending \textbf{RHCP3} to a Two-Target Look Ahead}
\label{Sec:TwoTargetLookahead}

In accordance with the goal of the RHC being decentralized, \textbf{RHCP3} limits feasible agent trajectories to a one-target lookahead $j \in \mathcal{N}_{i}$ ahead of target $i$. Therefore, an obvious extension expected to provide improvements is to consider agent trajectories two targets ahead of target $i$, assuming such information can be provided to target $i$. This is achieved by extending the associated planning horizon of \textbf{RHCP3} as shown in Fig. \ref{Fig:ExtendedTimeline} so that it includes an extra target visit to $k\in\mathcal{N}_{j}$ beyond vising target $j\in\mathcal{N}_{i}$.

\begin{figure}[h]
\centering
\includegraphics[width=3.2in]{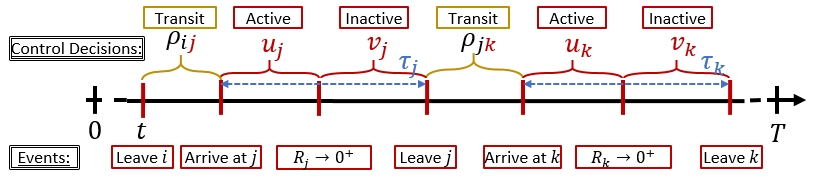}\caption{Extended
planning timeline for \textbf{RHCP3}.}%
\label{Fig:ExtendedTimeline}%
\end{figure}

In this case, the real-valued and discrete decision variables become $U_{ijk}=[u_{j},v_{j},u_{k},v_{k}]$ and $\{j,k\}$ respectively. The variable horizon $w$ defined in \eqref{Eq:VariableHorizon} becomes $w=\rho_{ij}+u_{j}+v+j+\rho_{jk}+u_{k}+v_{k}$. To obtain the optimal values of these decision variables, we first extend the concepts of neighborhood $\bar{\mathcal{N}}_{i}$ \eqref{Eq:Neighborset}, local objective $\bar{J}_{i}$ \eqref{Eq:LocalObjectiveFunction} and local state $X_{i}(t)$ respectively as $\tilde{\mathcal{N}}_{i}$, $\tilde{J}_{i}$, and $\tilde{X}_{i}(t)$ where
\[
\tilde{\mathcal{N}}_{i}=\cup_{j\in\bar{\mathcal{N}}_{i}}\mathcal{N}%
_{j},\ \ \ \tilde{J}_{i}=\sum_{m\in\tilde{\mathcal{N}}_{i}}J_{m}%
,\ \ \ \tilde{X}_{i}(t)=\{R_{m}(t);m\in\tilde{\mathcal{N}}_{i}\}.
\]
The extended \textbf{RHCP3} now takes the following form by modifying \eqref{Eq:RHCGenSolStep1}, \eqref{Eq:RHCGenSolStep2} and \eqref{Eq:RHCNewChoices}:
\begin{align}
&  U_{ijk}^{\ast}=\underset{U_{ijk}\in\mathbb{U}}{\arg\min}\ J_{H}(\tilde
{X}_{i}(t),U_{ijk};H);\forall j\in\mathcal{N}_{i},\forall k\in\mathcal{N}%
_{j},\label{Eq:RHCP3Extended1}\\
&  \{j^{\ast},k^{\ast}\}=\underset{j\in\mathcal{N}_{i},\ k\in\mathcal{N}_{j}}{\arg\min}\ J_{H}(\tilde{X}_{i}(t),U_{ijk}^{\ast};H) \mbox{ with }\label{Eq:RHCP3Extended2}
\end{align}%
\begin{equation}
\begin{aligned}\label{Eq:RHCNewChoices2} J_H(\tilde{X}_i(t),U_{ijk};H)=\frac{1}{w}\tilde{J}_i(t,t+w) \mbox{ and } \\ \mathbb{U} = \{U:U\in \mathbb{R}^4, U\geq 0, \vert U \vert +\rho_{ij}+\rho_{jk}\leq H\} \end{aligned}
\end{equation}

Note that target sequences $\{j,k\}$ where $\rho_{ij}+\rho_{jk}>H$ are omitted from evaluating \eqref{Eq:RHCP3Extended1}. Moreover, target sequences where one of the targets is covered are also omitted from evaluating \eqref{Eq:RHCP3Extended1}. As we will see next, there are many similarities between solving \eqref{Eq:RHCP3Extended1} and solving \eqref{Eq:RHCGenSolStep1} under \textbf{RHCP1}.

\paragraph{\textbf{Constraints}}

Following the same arguments used in obtaining \eqref{Eq:Lambdaj02}, an upper-bounds on control decisions $u_{j}$ and $u_{k}$ can be obtained as $u_{j}^{B}$ and $u_{k}^{B}$ where 
\begin{equation}
\label{Eq:Lambdaj03}\begin{aligned} u_j^B =& \frac{R_j(t)+A_j\rho_{ij}}{B_j-A_j} \mbox{ and }\\ u_k^B = u_k^B(u_j,v_j) =& \frac{R_k(t)+A_k(\rho_{ij}+\rho_{jk})}{B_k-A_k} + \frac{A_k}{B_k-A_k}\cdot(u_j+v_j). \end{aligned}
\end{equation}
Moreover, following \eqref{Eq:Constraints3}, four constraint set pairs $(\mathbb{U}_{jl},\mathbb{U}_{kn}),\, l,n\in\{1,2\}$ that define feasible $(u_{j},v_{j},u_{k},v_{k})$ in \eqref{Eq:RHCP3Extended1} (i.e., $U_{ijk}\in\mathbb{U}$) can be obtained as:
\begin{equation}
\label{Eq:Constraints4}
\begin{aligned} 
\mathbb{U}_{j1} =& \{0 \leq u_j \leq \bar{u}_j(u_k,v_k),\  v_j = 0\},\\ 
\mathbb{U}_{j2} =& \{u_j = u_j^B,\  0 \leq v_j \leq \bar{v}_j(u_k,v_k)\},\\ 
\mathbb{U}_{k1} =& \{0 \leq u_k \leq \bar{u}_k(u_j,v_j),\  v_k = 0\},\\
\mathbb{U}_{k2} =& \{u_k = u_k^B(u_j,v_j),\  0 \leq v_k \leq \bar{v}_k(u_j,v_j)\}, 
\end{aligned}
\end{equation}
where, $u_{j}^{B}$ and $u_{k}^{B}(u_{j},v_{j})$ are given in \eqref{Eq:Lambdaj03} and
\begin{align*}
\bar{u}_{j}(u_{k},v_{k})  &= \min\{u_{j}^{B},\ H-(\rho_{ij}+\rho_{jk}+u_{k}+v_{k})\},\\
\bar{v}_{j}(u_{k},v_{k}) &= H-(\rho_{ij}+u_{j}^{B}+\rho_{jk}+u_{k}+v_{k}),\\
\bar{u}_{k}(u_{j},v_{j})  &= \min\{u_{k}^{B}(u_{j},v_{j}),\ H-(\rho_{ij}+u_{j}+v_{j}+\rho_{jk})\},\\
\bar{v}_{k}(u_{j},v_{j})  &= H-(\rho_{ij}+u_{j}+v_{j}+\rho_{jk}+u_{k}^{B}(u_{j},v_{j})).
\end{align*}
The notation $\bar{u}_{j}$ and $\bar{v}_{j}$ respectively represent the limiting values of active and inactive times feasible at $j$. And $\bar{u}_{k}$ and $\bar{v}_{k}$ represent the same for target $k$.

\paragraph{\textbf{Objective}}

Following from the definition in \eqref{Eq:RHCNewChoices2}, the objective function of the extended \textbf{RHCP3} is taken as $J_{H}(U_{ijk}) = J_{H}(\tilde{X}_{i}(t),U_{ijk};H) = \frac{1}{w}\,\tilde{J}_{i}(t,t+w)$. To obtain an expression for $J_{H}(U_{ijk})$, extended neighborhood objective $\tilde{J}_{i}$ is decomposed as
\begin{equation}
\label{Eq:ObjDerivationOP3EStep1}
\tilde{J}_{i} = J_{j} + J_{k} + \sum_{m\in\tilde{\mathcal{N}}_{i}\backslash\{j,k\}}J_{m}.
\end{equation}
Next, the three terms $J_{j},\ J_{k}$ and $J_{m}$ are evaluated for a case where the agent goes from target $i$ to $j$ and then to $k$ following decisions $U_{ijk}$ during the period $[t,t+w)$. State trajectories for a such scenario is given in Fig. \ref{Fig:RHCP3EGraphs}. Theorem \ref{Th:Contribution} is utilized for this purpose to obtain:
\begin{align*}
    J_j=& 
    \frac{\rho_{ij}}{2}\left[ 2R_j(t) + A_j\rho_{ij} \right]  
    +\frac{u_j}{2}\left[2(R_j(t)+A_j\rho_{ij})-(B_j-A_j)u_j\right]\\
    &+\frac{(\rho_{jk}+u_k+v_k)}{2} \left[2(R_j(t)+A_j\rho_{ij}-(B_j-A_j)u_j)\right.\\
    &+\left.A_j(\rho_{jk}+u_k+v_k)\right],\\
    J_k=& \frac{(\rho_{ij}+u_j+v_j+\rho_{jk})}{2}\left[2R_k(t) + A_k(\rho_{ij}+u_j+v_j+\rho_{jk})\right] \\
    &+ \frac{u_k}{2}\left[ 2(R_k(t) + A_k(\rho_{ij}+u_j+v_j+\rho_{jk}))-(B_k-A_k)u_k\right],\\
    J_m=& \frac{(\rho_{ij}+u_j+v_j+\rho_{jk}+u_k+v_k)}{2} \left[ 2R_m(t)\right. \\
    &+A_m(\rho_{ij}+u_j+v_j+\rho_{jk}+u_k+v_k)\left.\right].
\end{align*}

\begin{figure}[!h]
    \centering
    \includegraphics[width=3.3in]{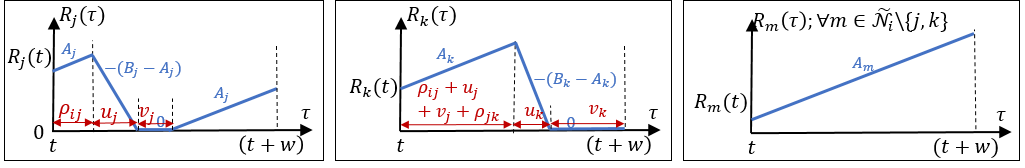}
    \caption{State trajectories of targets in $\bar{\mathcal{N}}_j$ during $[t,t+w)$.}
    \label{Fig:RHCP3EGraphs}
\end{figure}

Combining the above three results and substituting it in
\eqref{Eq:ObjDerivationOP3EStep1} gives the complete objective function $J_{H}(U_{ijk})$ as \begin{equation}\label{Eq:OP3EObjectiveSimplified}
\begin{split}
&J_H(u_j,v_j,u_k,v_k) = \hfill \\
&\frac{\left[
\begin{split}
C_1u_j^2 + C_2v_j^2 + C_3u_k^2 + C_4v_k^2 
+ C_5u_jv_j\\ + C_6u_ju_k + C_7u_jv_k
+ C_8v_ju_k + C_9v_jv_k + C_{10}u_kv_k\\ + C_{11}u_j 
+ C_{12}v_j + C_{13}u_k + C_{14}v_k + C_{15} 
\end{split}
\right]}{r_{ij} + u_j + v_j + \rho_{jk} + u_k + v_k},
\end{split}
\end{equation}
where
\begin{equation*}
\begin{aligned}
C_{1} &= \frac{\tilde{A}-B_j}{2},\ \ 
C_{2} = \frac{\tilde{A}_j}{2},\ \ 
C_{3} = \frac{\tilde{A}-B_k}{2},\ \ 
C_{4} = \frac{\tilde{A}_k}{2},\\
C_{5} &= \tilde{A}_{j},\ \ 
C_{6} = \tilde{A}-B_j,\ \
C_{7} = \tilde{A}_k-B_j,\ \ 
C_{8} = \tilde{A}_j,\\
C_{9} &= \tilde{A}_{jk},\ \ 
C_{10} = \tilde{A}_k,\ \  
C_{11} = \left[\tilde{R}(t)-B_j\rho_{jk}+\tilde{A}(\rho_{ij}+\rho_{jk})\right],\\ 
C_{12} &= \left[\tilde{R}_j(t)+\tilde{A}_j(\rho_{ij}+\rho_{jk})\right],\ \
C_{13} = \left[\tilde{R}(t)+\tilde{A}(\rho_{ij}+\rho_{jk})\right],\\
C_{14} &= \left[\tilde{R}_k(t)+\tilde{A}_k(\rho_{ij}+\rho_{jk})\right] \mbox{ and }\\ 
C_{15} &= \frac{(\rho_{ij}+\rho_{jk})}{2}\left[2\tilde{R}(t) + \tilde{A}(\rho_{ij}+\rho_{jk})\right],
\end{aligned}
\end{equation*}
with (neighborhood parameters)
\begin{equation} \label{Eq:NeighborhoodParameters2}
    \begin{aligned}
    \tilde{A}_{jk} = \sum_{m\in \tilde{\mathcal{N}}_i \backslash \{j,k\}} A_m,\ \ \ \ 
    \tilde{R}_{jk}(t) = \sum_{m\in \tilde{\mathcal{N}}_i \backslash \{j,k\}} R_m(t),\\
    \tilde{A}_j = \tilde{A}_{jk} + A_k,\ \ \tilde{A}_k = \tilde{A}_{jk} + A_j,\ \ 
    \tilde{A} = \tilde{A}_{jk} + A_j + A_k,\\
    \tilde{R}_j = \tilde{R}_{jk} + R_k,\ \ \tilde{R}_k = \tilde{R}_{jk} + R_k,\ \ 
    \tilde{R} = \tilde{R}_{jk} + R_j + R_k.\\
    \end{aligned}
\end{equation}
Note that all the coefficients stated above are non-negative except for $C_1,C_3,C_6,C_7$ and $C_{11}$.

\paragraph{\textbf{Solving the Extended \textbf{RHCP3} for }$(u_{j}^{*},v_{j}^{*}, u_{k}^{*},v_{k}^{*})$ and $\{j^{*},k^{*}\}$}

Based on \eqref{Eq:RHCP3Extended1} and \eqref{Eq:OP3EObjectiveSimplified}, the optimal decisions $(u_{j}^{*}, v_{j}^{*}, u_{k}^{*},v_{k}^{*})$ are given by
\begin{equation}
\label{Eq:OP3E_Formal}
u_{j}^{*}, v_{j}^{*}, u_{k}^{*}, v_{k}^{*})=\ \underset{(u_{j}, v_{j}, u_{k}, v_{k})}{\arg\min}\ {J_{H}(u_{j}, v_{j},u_{k}, v_{k}),}%
\end{equation}
where $(u_{j},v_{j})\in\mathbb{U}_{jl},\ (u_{k},v_{k})\in\mathbb{U}_{kn}$ for $l,n \in\{1,2\}$ as in \eqref{Eq:Constraints4}.

\paragraph*{\textbf{- Case 1}} 
$(u_j,v_j)\in\mathbb{U}_{j1}$, $(u_k,v_k)\in\mathbb{U}_{k1}$ in \eqref{Eq:Constraints4}: Then, $v_j = v_j^* = 0,\ v_k = v_k^* = 0$ and \eqref{Eq:OP3E_Formal} takes the form: 
\begin{equation}\label{Eq:OP3E_FormalPart1}
    \begin{aligned}
    (u_j^*,u_k^*) =\  &\underset{(u_j,u_k)}{\arg\min}\ {J_H(u_j,0,u_k,0)}\\
    &0 \leq u_j \leq u_j^B,\\
    &0 \leq u_k \leq u_k^B(u_j,0),\\
    &u_j+u_k \leq H-(\rho_{ij}+\rho_{jk}). 
    \end{aligned}
\end{equation}
To write the constraints in \eqref{Eq:OP3E_FormalPart1}, \eqref{Eq:Constraints4} and the relationships:
\begin{align*}
&u_j \leq \bar{u}_j(u_k,0) \iff u_j \leq u_j^B \And u_j \leq  H-(\rho_{ij}+\rho_{jk}+u_k),\\
&u_k \leq \bar{u}_k(u_j,0) \iff \\
&u_k \leq u_k^B(u_j,0) \And 
  u_k \leq H-(\rho_{ij}+u_j+\rho_{jk}),
\end{align*} 
have been used. Note that $u_k^B(u_j,0)$ is linear and increasing with $u_j$ (see \eqref{Eq:Lambdaj03}).

\paragraph*{\textbf{- Case 2}} 
$(u_j,v_j)\in\mathbb{U}_{j1}$, $(u_k,v_k)\in\mathbb{U}_{k2}$ in \eqref{Eq:Constraints4}: Then, $v_j = v_j^* = 0,\ u_k = u_k^* = u_k^B(u_j^*,0)$ and \eqref{Eq:OP3E_Formal} takes the form: 
\begin{equation}\label{Eq:OP3E_FormalPart2}
    \begin{aligned}
    (u_j^*,v_k^*) =\  &\underset{(u_j,v_i)}{\arg\min}\ {J_H(u_j,0,u_k^B(u_j,0),v_k)}\\
    &0 \leq u_j \leq u_j^B,\\
    &v_k \geq 0,\\
    &u_j+u_k^B(u_j,0)+v_k \leq H-(\rho_{ij}+\rho_{jk}). 
    \end{aligned}
\end{equation}
The constraints in \eqref{Eq:OP3E_FormalPart2} are from \eqref{Eq:Constraints4} and the relationships:
\begin{align*}
&u_j \leq \bar{u}_j(u_k^B(u_j,0),v_k) \iff \\
&u_j \leq u_j^B \And u_j \leq   H-(\rho_{ij}+\rho_{jk}+u_k^B(u_j,0)+v_k) \mbox{ and }\\
&v_k \leq \bar{v}_k(u_j,0) \iff v_k \leq H-(\rho_{ij}+u_j+\rho_{jk}+u_k^B(u_j,0)).
\end{align*}

\paragraph*{\textbf{- Case 3}}
$(u_j,v_j)\in\mathbb{U}_{j2}$, $(u_k,v_k)\in\mathbb{U}_{k1}$ in \eqref{Eq:Constraints4}: Then, $u_j = u_j^* = u_j^B,\ v_k = v_k^* = 0$ and \eqref{Eq:OP3E_Formal} takes the form: 
\begin{equation}\label{Eq:OP3E_FormalPart3}
    \begin{aligned}
    (v_j^*,u_k^*) =\  &\underset{(v_j,u_k)}{\arg\min}\ {J_H(u_j^B,v_j,u_k,0)}\\
    &v_j \geq 0,\\
    &0 \leq u_k \leq u_k^B(u_j^B,v_j),\\
    &v_j+u_k \leq H-(\rho_{ij}+u_j^B+\rho_{jk}). 
    \end{aligned}
\end{equation}
The constraints in \eqref{Eq:OP3E_FormalPart3}, are from \eqref{Eq:Constraints4} and the relationships:
\begin{align*}
&v_j \leq \bar{v}_j(u_k,0) \iff v_j \leq  H-(\rho_{ij}+u_j^B+\rho_{jk}+u_k) \mbox{ and }\\
&u_k \leq \bar{u}_k(u_j^B,v_j) \iff \\
&u_k \leq u_k^B(u_j^B,v_j) \And 
  u_k \leq H-(\rho_{ij}+u_j^B+v_j+\rho_{jk}).
\end{align*} 
Note that $u_k^B(u_j^B,v_j)$ is linear and increasing with $v_j$ \eqref{Eq:Lambdaj03}.

\paragraph*{\textbf{- Case 4}} 
$(u_j,v_j)\in\mathbb{U}_{j2}$, $(u_k,v_k)\in\mathbb{U}_{k2}$ in \eqref{Eq:Constraints4}: Then, $u_j = u_j^* = u_j^B,\ u_k = u_k^* = u_k^B(u_j^B,v_j^*),\ $ and \eqref{Eq:OP3E_Formal} takes the form: 
\begin{equation}\label{Eq:OP3E_FormalPart4}
    \begin{aligned}
    (v_j^*,v_k^*) =\  \underset{(v_j,v_k)}{\arg\min}\ &{J_H(u_j^B,v_j,u_k^B(u_j^B,v_j),v_k)}\\
    v_j &\geq 0, \\ 
    v_k &\geq 0, \\
    v_j + v_k + u_k^B(u_j^B,v_j) &\leq H-(\rho_{ij}+u_j^B+\rho_{jk}).
    \end{aligned}
\end{equation}
To write the last constraint in \eqref{Eq:OP3E_FormalPart4}, \eqref{Eq:Constraints4} and the relationships:
\begin{align*}
v_j \leq&\ \bar{v}_j(u_k^B(u_j^B,v_j),v_k) \iff \\ 
&v_j \leq H-(\rho_{ij}+u_j^B+\rho_{jk}+u_k^B(u_j^B,v_j)+v_k), \mbox{ and, }\\
v_k \leq&\ \bar{v}_k(u_j^B,v_j) \iff \\
&v_k \leq H-(\rho_{ij}+u_j^B+v_j+\rho_{jk}+u_k^B(u_j^B,v_j)),
\end{align*} 
have been used.

\paragraph*{\textbf{- Combined Result}}
Similar to all the sub-problems discussed under \textbf{RHCP2} and \textbf{RHCP1}, the sub-problems in \eqref{Eq:OP3E_FormalPart1}, \eqref{Eq:OP3E_FormalPart2}, \eqref{Eq:OP3E_FormalPart3} and \eqref{Eq:OP3E_FormalPart4} formulated above belong to the same class of constrained rational function optimization problems (see  \eqref{Eq:ConstrainedOptimizationBivariate}) discussed in Appendix A. Therefore, each of these four problems is solved exploiting the computationally cheap theoretical solution presented in Appendix A. Exact mapping details of each of these four problems to the generic RFOP form in \eqref{Eq:ConstrainedOptimizationBivariate} are similar to those of \textbf{RHCP1} discussed before. Therefore, we omit those details for brevity.

Upon obtaining solutions to \eqref{Eq:OP3E_FormalPart1}-\eqref{Eq:OP3E_FormalPart4}, the main optimization problem \eqref{Eq:OP3E_Formal} is solved by just comparing objective function values of the obtained individual solutions.  

\paragraph{\textbf{Solving for Optimal Next Destinations: $\{j^*,k^*\}$}}
The second step of the extended \textbf{RHCP3} (i.e.,  \eqref{Eq:RHCP3Extended2}) is to choose the optimum target sequence $\{j,k\}$ according to   
\begin{equation}\label{Eq:OP3E_FormalStep2}
\{j^*,k^*\} = \underset{j\in\mathcal{N}_i, k\in \mathcal{N}_j}{\arg\min}\ J_H(u_j^*,v_j^*,u_k^*,v_k^*).    \end{equation}
Note that this step requires the cost value of the optimal solution $U_{ijk}^{*} = [u_j^*,v_j^*,u_k^*,v_k^*]$ obtained for each $j\in\mathcal{N}_i$ and $k\in\mathcal{N}_j$ (in \eqref{Eq:OP3E_Formal}). 

In the case of extended \textbf{RHCP3}, as shown in Fig. \ref{Fig:ActionHorizon}, above $j^*$ in \eqref{Eq:OP3E_FormalStep2} defines the \textquotedblleft Action\textquotedblright\ that the agent has to take at (current) time $t$. In other words, the agent $a$  has to leave target $i$ at time $t$ and follow the path $(i,j^*)\in\mathcal{E}$ to visit target $j^*$.

\paragraph{\textbf{Discussion}}

Aside from the obvious increment in computational requirements, one clear disadvantage of the Ex-RHC method is that agents now require more information to make their next-visit target $j^{\ast}$ decisions thus compromising the distributed nature of the solution. Even though we expect that the payoff of such a compromise is a considerable performance improvement, this is far from evident in the numerical examples shown in Section \ref{Sec:SimulationResults}. One reason may be the substantial errors in estimating $R_{k}(t)$ trajectories (required to evaluate $J_{k}$ in \eqref{Eq:ObjDerivationOP3EStep1} (or \eqref{Eq:ModifiedObjectiveForRHCP3_2})) when there are multiple agents and when target $k$ is located far from target $i$. These errors indirectly and negatively affect the crucial $j^{\ast}$ decision. The Ex-RHC method generally performs better when (\romannum{1}) the number of both agents and targets are relatively low, and (\romannum{2}) transit times in the graph are relatively short.

\subsection{Using a Weighted Objective in Extended \textbf{RHCP3}}\label{SubSec:EDRHCAlphaBetaMethod}
Inspired by the modification proposed in Section \ref{SubSec:EDRHCAlphaMethod} for the \textbf{RHCP3} of the ED-RHC method, we now reform the above extended \textbf{RHCP3} of the Ex-ED-RHC method. First, the neighborhood objective function decomposition \eqref{Eq:ObjDerivationOP3EStep1} is modified so as to incorporate weight factors $\alpha, \beta \in [0,1]$ as 
\begin{equation}\label{Eq:ModifiedObjectiveForRHCP3_2}
    \tilde{J}_i = \alpha J_j + \beta J_k +  (1-\alpha-\beta)\sum_{m\in\tilde{\mathcal{N}}_i\backslash\{j,k\}} J_m.
\end{equation}

The motivation behind \eqref{Eq:ModifiedObjectiveForRHCP3_2} is to emphasize the contribution to the global cost by the neglected neighborhood targets $m\in\tilde{\mathcal{N}}_i\backslash \{j,k\}$ and by the target $k$ - due to the choice of $j\in\mathcal{N}_i$ as the immediate next target to visit. Hence $\alpha < \frac{1}{3}$ and $\beta < \frac{2}{3}-\alpha$ is preferred and we propose to use $\alpha = \frac{1}{\vert \bar{\mathcal{N}}_i \vert^2}$ and $\beta = \frac{1}{\vert \bar{\mathcal{N}}_i \vert}$ as nominal choices. We represent the Ex-ED-RHC method which uses this modification as ED-RHC$^{\alpha\beta}$. 

\begin{lemma}
\label{Lm:ED-RHCPalphabeta} 
In Ex-RHC$^{\alpha\beta}$, when $\alpha=\beta=0$ is used, the optimum next-visit target sequence $\{j^{*},k^{*}\}$ (i.e., the solution of \eqref{Eq:RHCP3Extended2}) for the extended \textbf{RHCP3} is, 
\begin{align}
\label{Eq:ED-RHCPalphabeta}\{j^{*},k^{*}\} = \underset{j \in\mathcal{N}%
_{i},\ k \in\mathcal{N}_{j}}{\arg\min} \Big[(2\tilde{R}(t)+\tilde{A}(\rho
_{ij}+\rho_{jk})) - (2R_{j}(t)\nonumber\\
+A_{j}(\rho_{ij}+\rho_{jk})) - (2R_{k}(t)+A_{k}(\rho_{ij}+\rho_{jk})) \Big].
\end{align}
\end{lemma}
\emph{Proof: }
This result can be obtained directly following the same steps as in the proof of Lemma \ref{Lm:ED-RHCPalpha}. Therefore it is omitted.
\hfill $\blacksquare$

The above result generalizes Lemma \ref{Lm:ED-RHCPalpha} and provides a direct, simple and intuitive way to obtain the next target decision $j^*$ for the extended \textbf{RHCP3} (skipping \eqref{Eq:RHCP3Extended1}). Note that the three terms in \eqref{Eq:ED-RHCPalphabeta} respectively approximates the contribution to the main objective \eqref{Eq:MainObjective} during transit time $\rho_{ij}+\rho_{jk}$ by the neighborhood, target $j$ and target $k$.

To provide an intuition on how the choices of $\alpha$ and $\beta$ affects the PMG problem performance, we provide Figs. \ref{Fig:AlphaBetaVariation1} and \ref{Fig:AlphaBetaVariation2} which shows that in some PMG scenarios $\alpha=\beta=0$ is directly the optimum choice and in some others, there may be a particular $\alpha<\frac{1}{3}$ and a $\beta<\frac{2}{3}-\alpha$ choice which provides the optimum performance.   

\begin{figure}[!t]
     \centering
     \begin{subfigure}[b]{0.48\columnwidth}
         \centering
         \includegraphics[width = 1.5in]{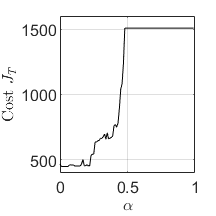}
         \caption{$J_T$ vs $\alpha$ with fixed $\beta$}
     \end{subfigure}
     \begin{subfigure}[b]{0.48\columnwidth}
         \centering
         \includegraphics[width = 1.5in]{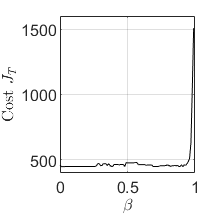}
         \caption{$J_T$ vs $\beta$ with fixed $\alpha$}
     \end{subfigure}
     \caption{Single-Agent Case (SASE4 in Fig. \ref{Fig:SASE4}).}
    \label{Fig:AlphaBetaVariation1}
\end{figure}

\begin{figure}[!t]
     \centering
     \begin{subfigure}[b]{0.48\columnwidth}
         \centering
         \includegraphics[width = 1.5in]{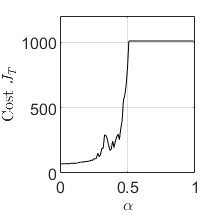}
         \caption{$J_T$ vs $\alpha$ with fixed $\beta$}
     \end{subfigure}
     \begin{subfigure}[b]{0.48\columnwidth}
         \centering
         \includegraphics[width = 1.5in]{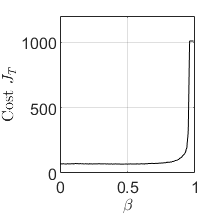}
         \caption{$J_T$ vs $\beta$ with fixed $\alpha$}
     \end{subfigure}
     \caption{Multi-Agent Case (MASE4 in Fig. \ref{Fig:MASE4}).}
    \label{Fig:AlphaBetaVariation2}
\end{figure}

\section{Simulation Results} \label{Sec:SimulationResults}

\subsection{Performance Comparison}

We begin by comparing the performance measured through $J_{T}$ in \eqref{Eq:MainObjective} over several different PMN problem configurations using: 
(\romannum{1}) The IPA-TCP method \cite{Zhou2019} 
(\romannum{2}) The RHC method proposed in Section \ref{Sec:SolvingRHCPs}, 
(\romannum{3}) The RHC$^{\alpha}$ method (Section \ref{SubSec:EDRHCAlphaMethod}), and
(\romannum{4}) The Ex-RHC$^{\alpha\beta}$ method (Section \ref{Sec:TwoTargetLookahead}). All of these control techniques have been implemented in a JavaScript based simulator, which is made available at \href{http://www.bu.edu/codes/simulations/shiran27/PersistentMonitoring/}{http://www.bu.edu/codes/simulations/shiran27/ PersistentMonitoring/}. Readers are invited to reproduce the reported results and also to try new problem configurations using the developed interactive simulator.


We specifically consider three configurations with a single agent as shown in Figs. \ref{Fig:SASE1}-\ref{Fig:SASE4} and with multiple agents as shown in Figs. \ref{Fig:MASE1}-\ref{Fig:MASE4}. Blue circles represent the targets, while black lines represent trajectory segments that agents can use to travel between targets. Red triangles and yellow vertical bars indicate the agent locations and the target uncertainty levels, respectively. Moreover, since both $s_{a}(t)$ and $R_{i}(t)$ are time-varying, the figures show only their state at time $t=T$.

In each problem configuration, the target parameters were chosen as $A_{i}=1,\ B_{i}=10$ and $R_{i}(0)=0.5,\ \forall i\in\mathcal{T}$, and their locations are specified and placed inside a $600\times600$ mission space. The overall time period is $T=500$. Each agent is assumed to follow first-order dynamics (similar to \cite{Zhou2019, Welikala2019P3}) with a maximum speed of $V_{ij}=50$ units per second on each trajectory segment $(i,j)\in\mathcal{E}$. The initial locations of the agents were chosen such that they are uniformly distributed among the targets at time $t=0$ (i.e., $s_{a}(0)=Y_{i}$ with $i=1+(a-1)\times\mathrm{round}(M/N)$). The (non-critical) upper-bound for each planning horizon $w$ was chosen as $H=T/2=250$ for the three RHC approaches and $\alpha=\frac{1}{|\mathcal{N}_{i}|^{2}},\ \beta=\frac{1}{|\mathcal{N}_{i}|}$ were used in the RHC$^{\alpha}$ and Ex-RHC$^{\alpha\beta}$ methods.

Each sub-figure caption in Figs. \ref{Fig:SASE1}-\ref{Fig:MASE4} provides the cost value $J_{T}$ in \eqref{Eq:MainObjective} observed under each controller (i.e., either IPA-TCP, RHC, RHC$^{\alpha}$ or Ex-RHC$^{\alpha\beta}$). These cost values are summarized in Tab. \ref{Tab:SummaryResults}. From the observed results, note that the proposed RHC method has performed considerably better (on average $50.4\%$ better) than the IPA-TCP method for multi-agent problem configurations. For single-agent problem configurations, both methods have performed equally except for SASE3. The proposed RHC$^{\alpha}$ approach further improves these performances compared to the IPA-TCP method by an average of $66.8\%$ for multi-agent situations and by $9.76\%$ for single-agent situations. On the other hand, the proposed Ex-RHC$^{\alpha\beta}$ method provides on average $63.3\%$ and $9.9\%$ improvements respectively for multi-agent and single-agent cases compared to the IPA-TCP method. In light of the fact that Ex-RHC$^\alpha\beta$ compromises the distributed nature of the original RHC, all evidence points to the conclusion that there is no benefit to this extension for most network topologies.

\begin{table}[]
\centering
\caption{Summary of obtained results across all the simulation examples.}
\label{Tab:SummaryResults}
\resizebox{\columnwidth}{!}{%
\begin{tabular}{|l|c|r|r|r|c|r|r|r|}
\hline
\multirow{2}{*}{$J_T$ in \eqref{Eq:MainObjective}} & \multicolumn{4}{c|}{\begin{tabular}[c]{@{}c@{}}Singe Agent \\ Simulation Examples\end{tabular}} & \multicolumn{4}{c|}{\begin{tabular}[c]{@{}c@{}}Multi-Agent \\ Simulation Examples\end{tabular}} \\ \cline{2-9} 
 & 1 & \multicolumn{1}{c|}{2} & \multicolumn{1}{c|}{3} & \multicolumn{1}{c|}{4} & 1 & \multicolumn{1}{c|}{2} & \multicolumn{1}{c|}{3} & \multicolumn{1}{c|}{4} \\ \hline
IPA-TCP & \multicolumn{1}{r|}{831.3} & 129.2 & 651.3 & 497.9 & \multicolumn{1}{r|}{270.2} & 91.7 & 274.0 & 201.3 \\ \hline
RHC & \multicolumn{1}{r|}{791.1} & 141.4 & 912.3 & 490.4 & \multicolumn{1}{r|}{105.5} & 63.7 & 114.1 & 97.2 \\ \hline
RHC$^\alpha$ & \multicolumn{1}{r|}{\textbf{790.1}} & 121.3 & \textbf{527.7} & 464.8 & \multicolumn{1}{r|}{96.6} & \textbf{40.4} & \textbf{63.7} & \textbf{60.1} \\ \hline
Ex-RHC$^{\alpha\beta}$ & \multicolumn{1}{r|}{790.1} & \textbf{120.9} & 529.7 & \textbf{449.5} & \multicolumn{1}{r|}{\textbf{95.7}} & 45.0 & 75.0 & 70.2 \\ \hline
\end{tabular}%
}
\end{table}

\begin{figure}[!h]
     \centering
     \begin{subfigure}[b]{0.24\columnwidth}
         \centering
         \includegraphics[width = \textwidth]{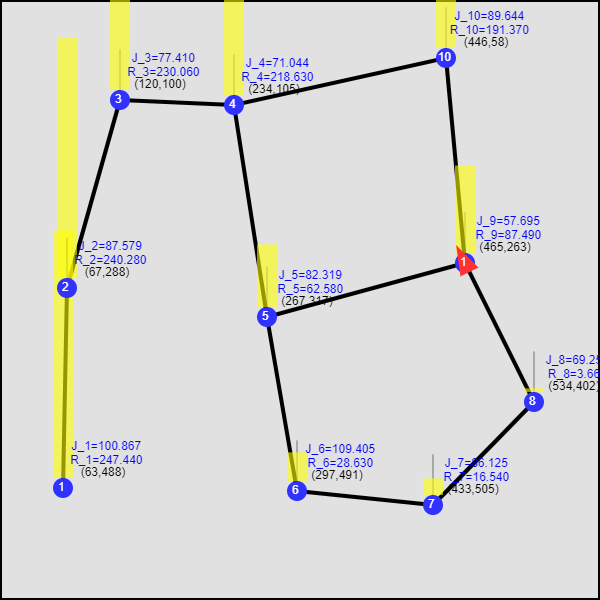}
         \caption{IPA-TCP:\\ \centering{$J_T= 831.3$}.}
     \end{subfigure}
     \begin{subfigure}[b]{0.24\columnwidth}
         \centering
         \includegraphics[width = \textwidth]{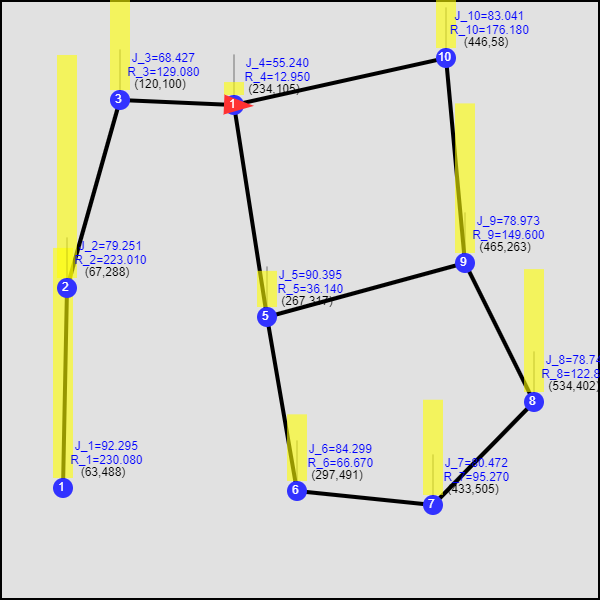}
         \caption{RHC:\\ \centering{$J_T= \textbf{791.1}$.}}
     \end{subfigure}
     \begin{subfigure}[b]{0.24\columnwidth}
         \centering
         \includegraphics[width = \textwidth]{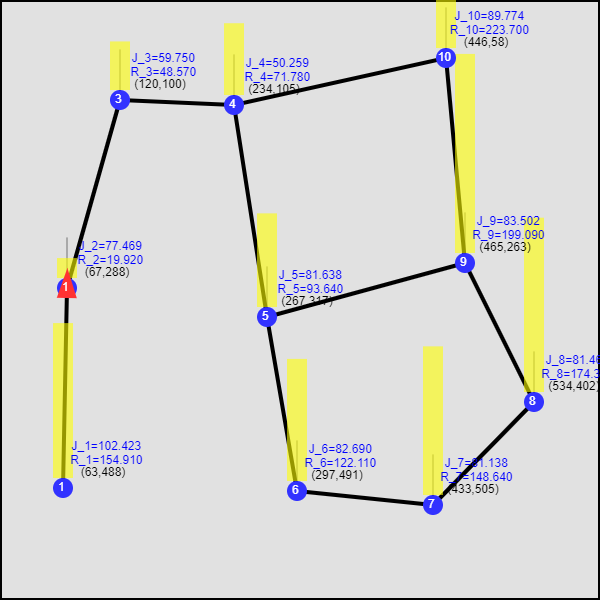}
         \caption{RHC$^\alpha$:\\ \centering{$J_T= 790.1$}.}
     \end{subfigure}
     \begin{subfigure}[b]{0.24\columnwidth}
         \centering
         \includegraphics[width = \textwidth]{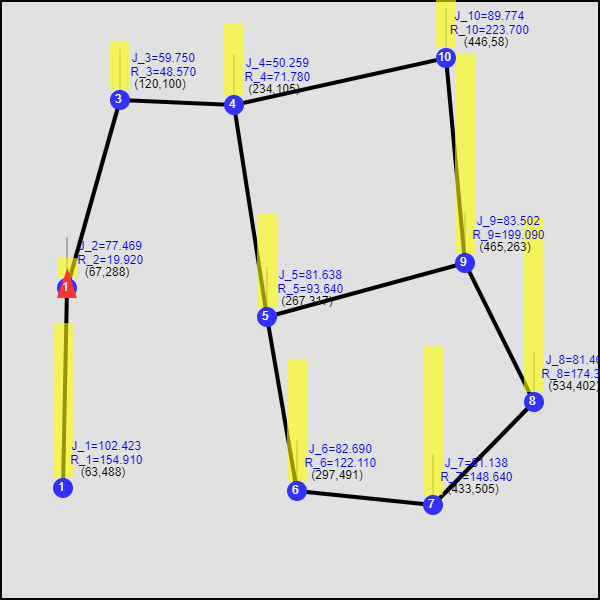}
         \caption{Ex-RHC$^{\alpha\beta}$:\\ \centering{$J_T= 790.1$}.}
     \end{subfigure}
    \caption{Single-agent simulation example 1 (SASE1).}
    \label{Fig:SASE1}
\end{figure}

\begin{figure}[!h]
     \centering
     \begin{subfigure}[b]{0.24\columnwidth}
         \centering
         \includegraphics[width = \textwidth]{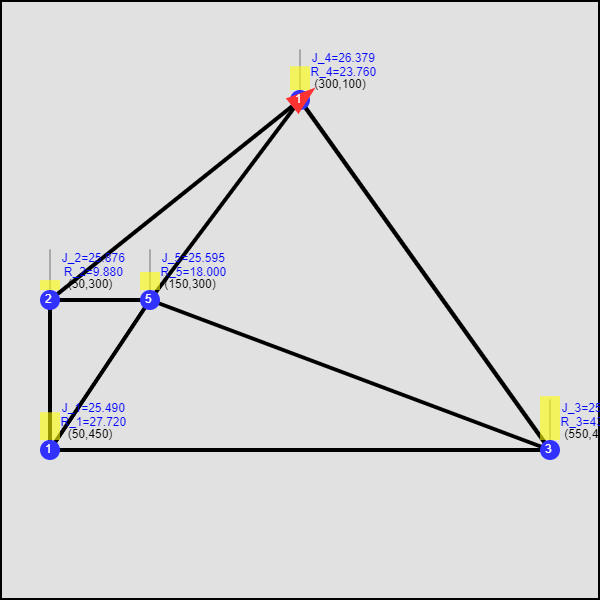}
         \caption{IPA-TCP:\\ \centering{$J_T= 129.2$.}}
     \end{subfigure}
     \begin{subfigure}[b]{0.24\columnwidth}
         \centering
         \includegraphics[width = \textwidth]{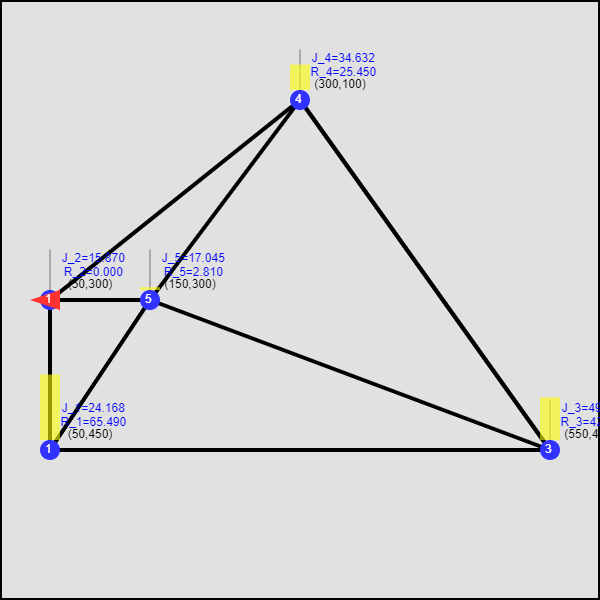}
         \caption{RHC:\\ \centering{$J_T= 141.4$.}}
     \end{subfigure}
     \begin{subfigure}[b]{0.24\columnwidth}
         \centering
         \includegraphics[width = \textwidth]{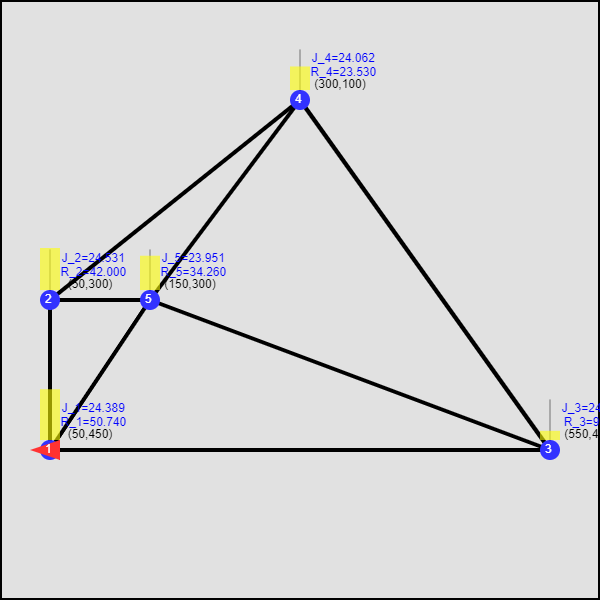}
         \caption{RHC$^\alpha$:\\ \centering{$J_T= 121.3$}.}
     \end{subfigure}
     \begin{subfigure}[b]{0.24\columnwidth}
         \centering
         \includegraphics[width = \textwidth]{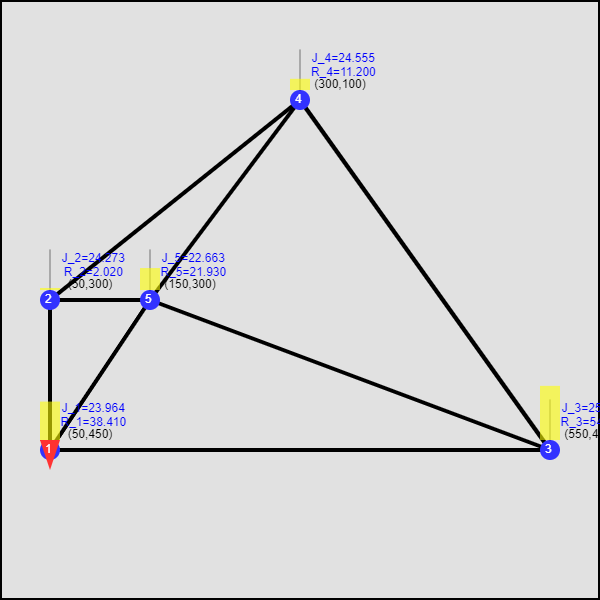}
         \caption{Ex-RHC$^{\alpha\beta}$:\\ \centering{$J_T= \textbf{120.9}$}.}
     \end{subfigure}
    \caption{Single-agent simulation example 2 (SASE2).}
    \label{Fig:SASE2}
\end{figure}

\begin{figure}[!h]
     \centering
     \begin{subfigure}[b]{0.24\columnwidth}
         \centering
         \includegraphics[width = \textwidth]{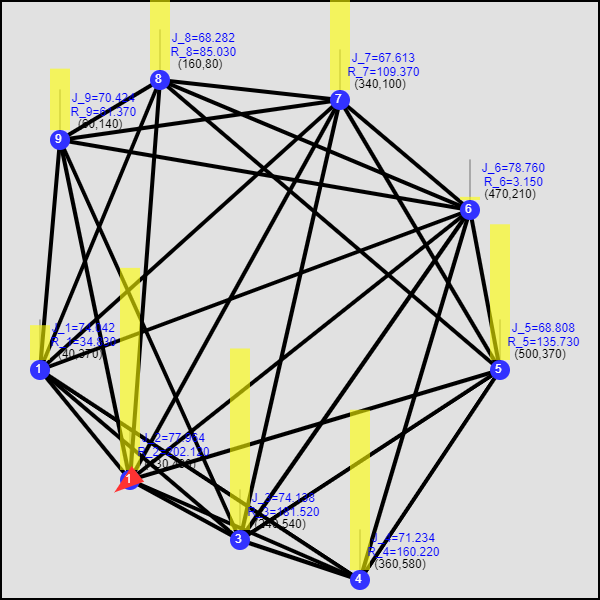}
         \caption{IPA-TCP:\\ \centering{$J_T= 651.3$}.}
     \end{subfigure}
     \begin{subfigure}[b]{0.24\columnwidth}
         \centering
         \includegraphics[width = \textwidth]{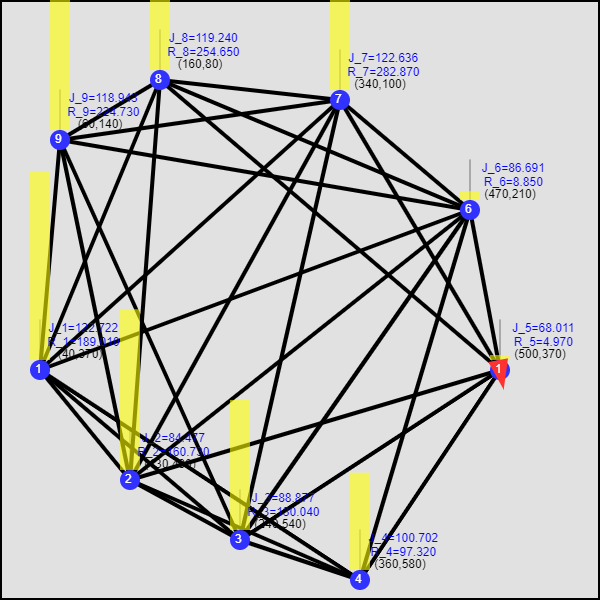}
         \caption{RHC:\\ \centering{$J_T= 912.2$.}}
     \end{subfigure}
     \begin{subfigure}[b]{0.24\columnwidth}
         \centering
         \includegraphics[width = \textwidth]{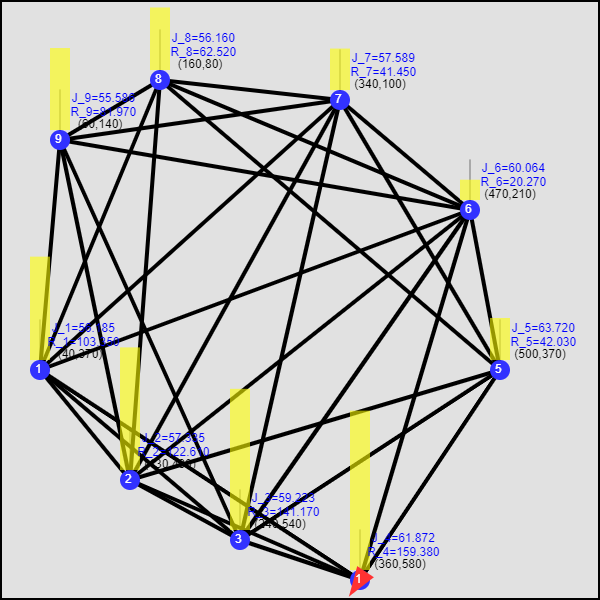}
         \caption{RHC$^\alpha$:\\ \centering{$J_T= \textbf{527.7}$}.}
     \end{subfigure}
     \begin{subfigure}[b]{0.24\columnwidth}
         \centering
         \includegraphics[width = \textwidth]{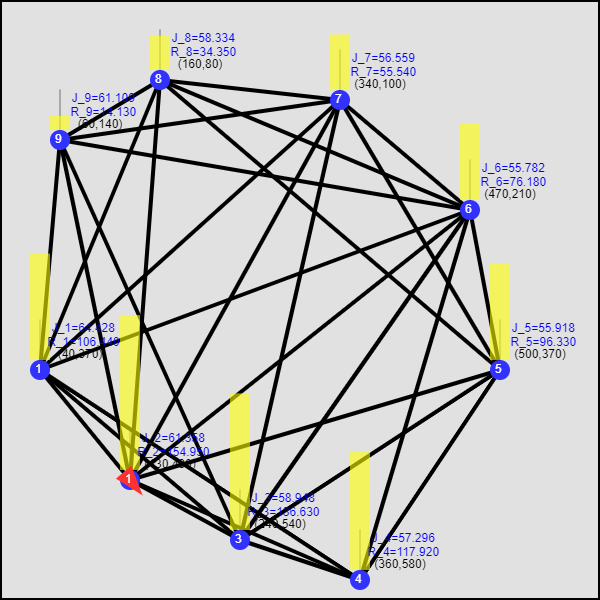}
         \caption{Ex-RHC$^{\alpha\beta}$:\\ \centering{$J_T= 529.7$}.}
     \end{subfigure}
    \caption{Single-agent simulation example 3 (SASE3).}
    \label{Fig:SASE3}
\end{figure}

\begin{figure}[!h]
     \centering
     \begin{subfigure}[b]{0.24\columnwidth}
         \centering
         \includegraphics[width = \textwidth]{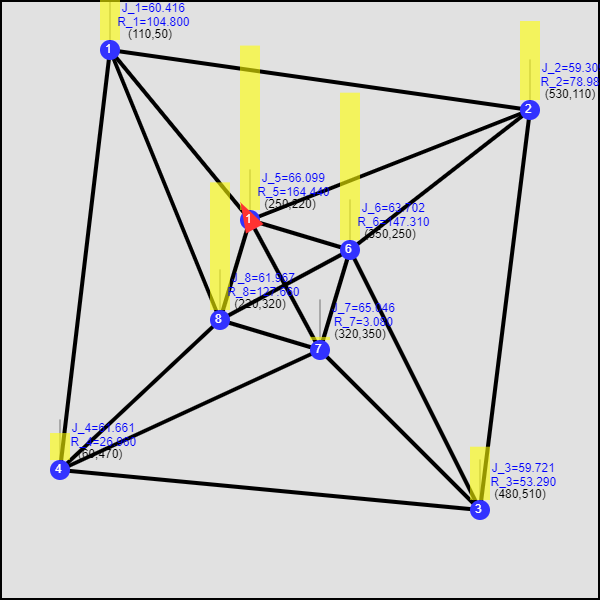}
         \caption{IPA-TCP: \\ \centering{$J_T= 497.9$.}}
     \end{subfigure}
     \begin{subfigure}[b]{0.24\columnwidth}
         \centering
         \includegraphics[width = \textwidth]{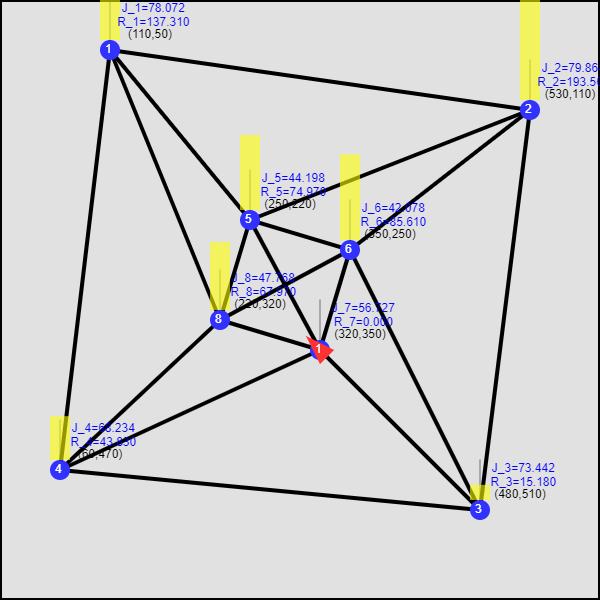}
         \caption{RHC: \\ \centering{$J_T=490.4$.}}
     \end{subfigure}
     \begin{subfigure}[b]{0.24\columnwidth}
         \centering
         \includegraphics[width = \textwidth]{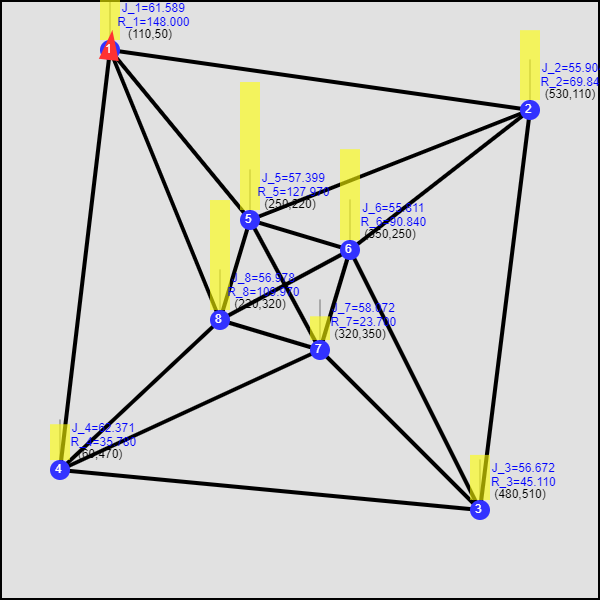}
         \caption{RHC$^\alpha$:\\ \centering{$J_T=464.8$}.}
     \end{subfigure}
     \begin{subfigure}[b]{0.24\columnwidth}
         \centering
         \includegraphics[width = \textwidth]{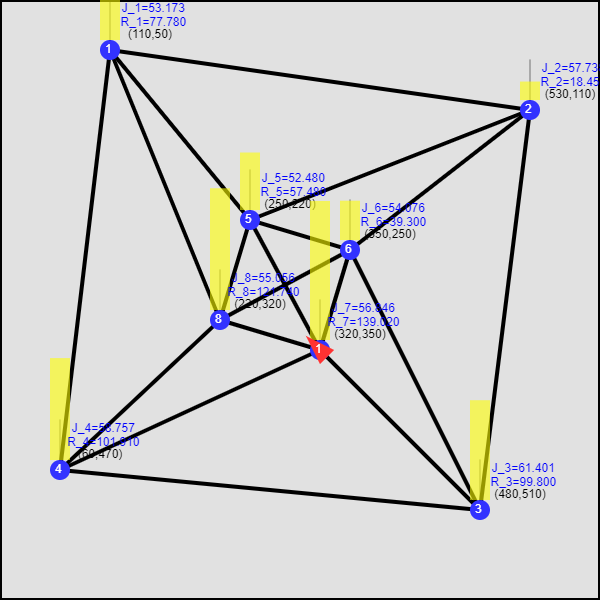}
         \caption{Ex-RHC$^{\alpha\beta}$:\\ \centering{$J_T=\textbf{449.5}$}.}
     \end{subfigure}
    \caption{Single-agent simulation example 4 (SASE4).}
    \label{Fig:SASE4}
\end{figure}

\begin{figure}[!h]
     \centering
     \begin{subfigure}[b]{0.24\columnwidth}
         \centering
         \includegraphics[width = \textwidth]{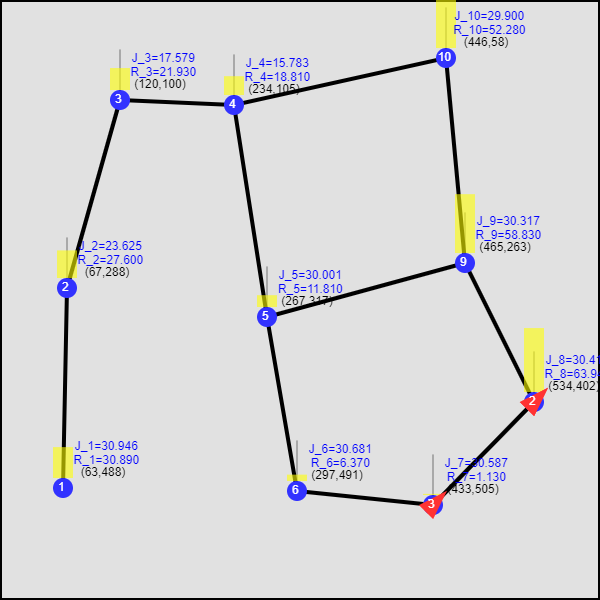}
         \caption{IPA-TCP: \\ \centering{$J_T= 270.2$.}}
     \end{subfigure}
     \begin{subfigure}[b]{0.24\columnwidth}
         \centering
         \includegraphics[width = \textwidth]{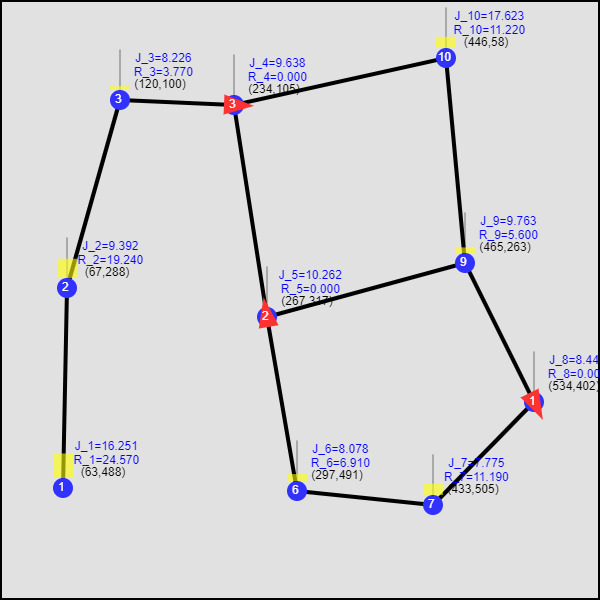}
         \caption{RHC: \\ \centering{$J_T= 105.5$.}}
     \end{subfigure}
     \begin{subfigure}[b]{0.24\columnwidth}
         \centering
         \includegraphics[width = \textwidth]{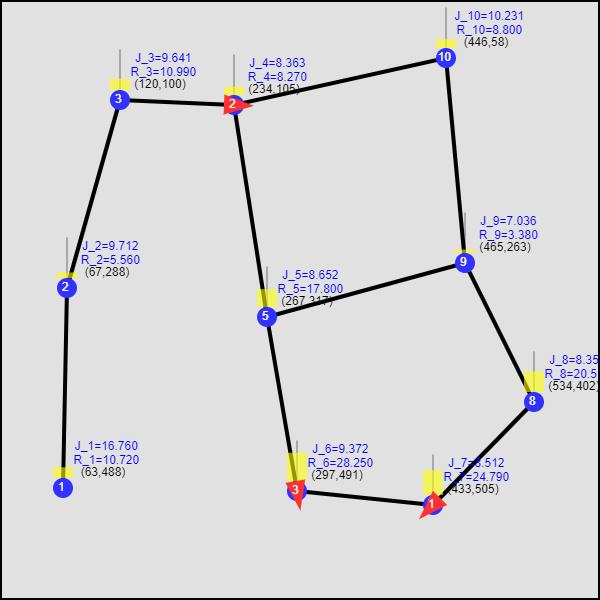}
         \caption{RHC$^\alpha$:\\ \centering{$J_T=96.6$}.}
     \end{subfigure}
     \begin{subfigure}[b]{0.24\columnwidth}
         \centering
         \includegraphics[width = \textwidth]{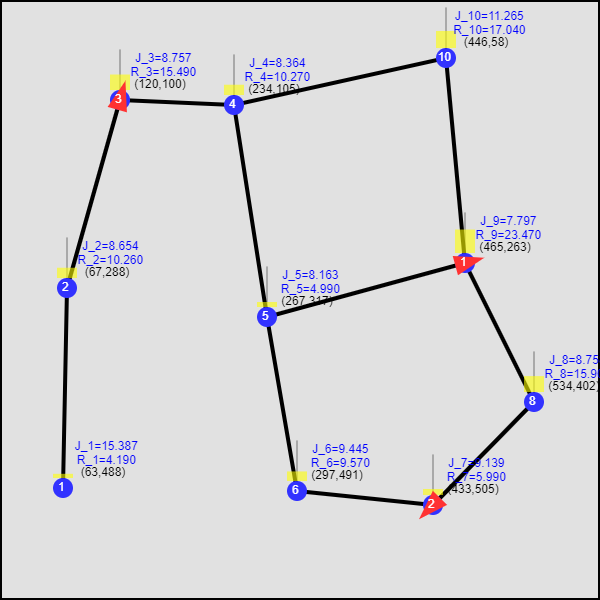}
         \caption{Ex-RHC$^{\alpha\beta}$:\\ \centering{$J_T=\textbf{95.7}$}.}
     \end{subfigure}
    \caption{Multi-agent simulation example 1 (MASE1).}
    \label{Fig:MASE1}
\end{figure}

\begin{figure}[!h]
     \centering
     \begin{subfigure}[b]{0.24\columnwidth}
         \centering
         \includegraphics[width = \textwidth]{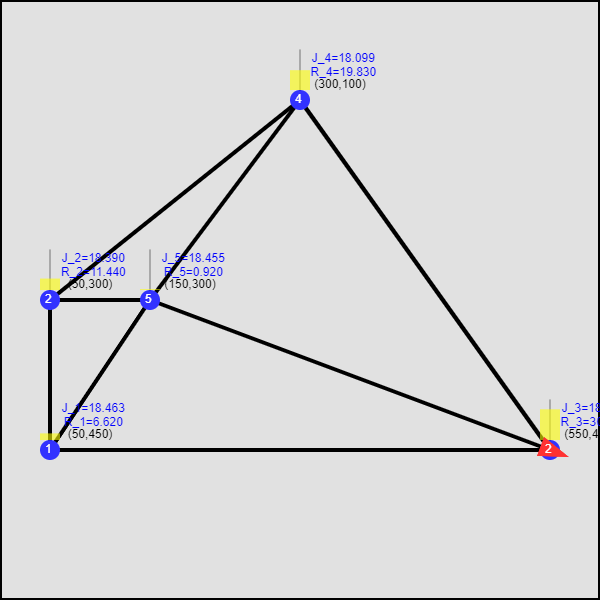}
         \caption{IPA-TCP: \\ \centering{$J_T= 91.7$.}}
     \end{subfigure}
     \begin{subfigure}[b]{0.24\columnwidth}
         \centering
         \includegraphics[width = \textwidth]{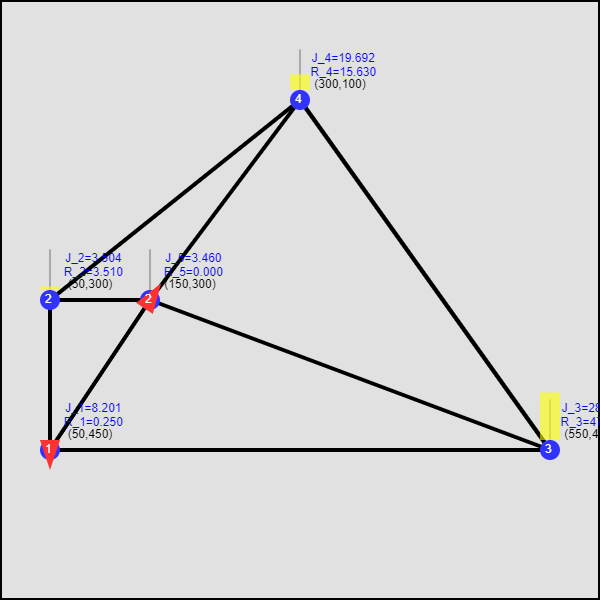}
         \caption{RHC: \\ \centering{$J_T= 63.7$.}}
     \end{subfigure}
     \begin{subfigure}[b]{0.24\columnwidth}
         \centering
         \includegraphics[width = \textwidth]{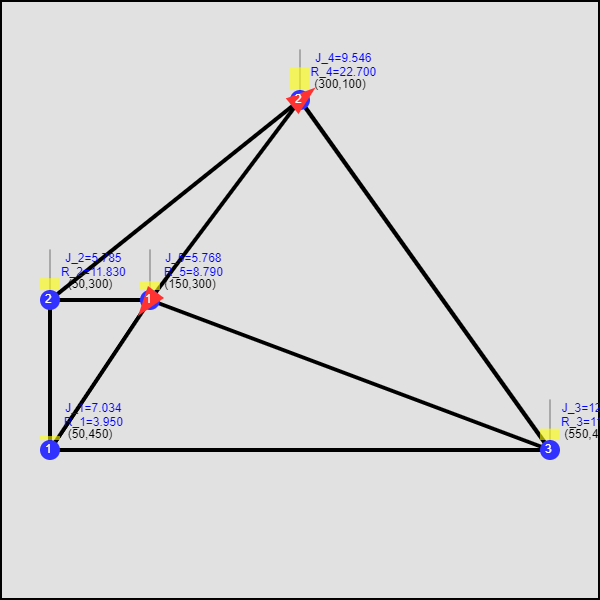}
         \caption{RHC$^\alpha$:\\ \centering{$J_T= \textbf{40.4}$}.}
     \end{subfigure}
     \begin{subfigure}[b]{0.24\columnwidth}
         \centering
         \includegraphics[width = \textwidth]{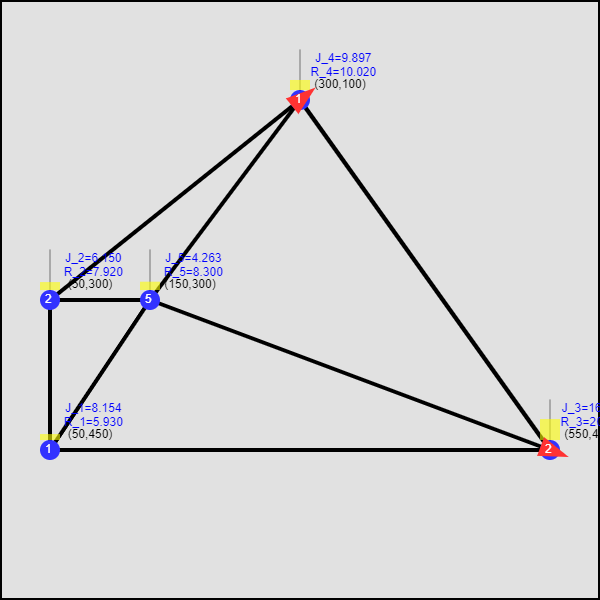}
         \caption{Ex-RHC$^{\alpha\beta}$:\\ \centering{$J_T=45.0$}.}
     \end{subfigure}
    \caption{Multi-agent simulation example 2 (MASE2).}
    \label{Fig:MASE2}
\end{figure}

\begin{figure}[!h]
     \centering
     \begin{subfigure}[b]{0.24\columnwidth}
         \centering
         \includegraphics[width = \textwidth]{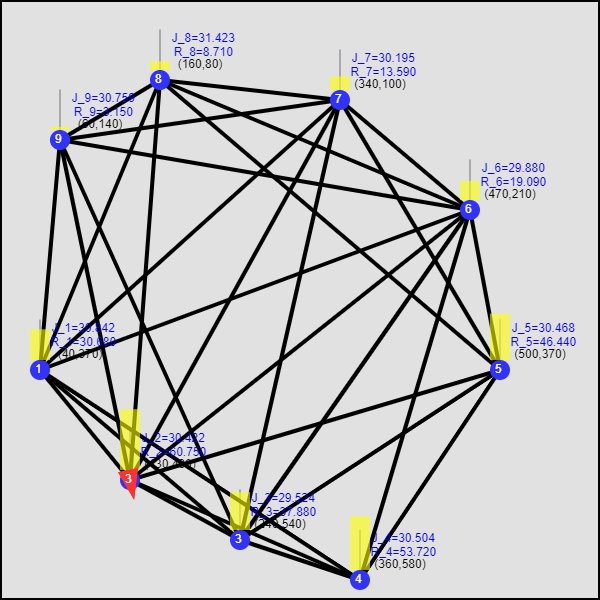}
         \caption{IPA-TCP: \\ \centering{$J_T= 274.0$.}}
     \end{subfigure}
     \begin{subfigure}[b]{0.24\columnwidth}
         \centering
         \includegraphics[width = \textwidth]{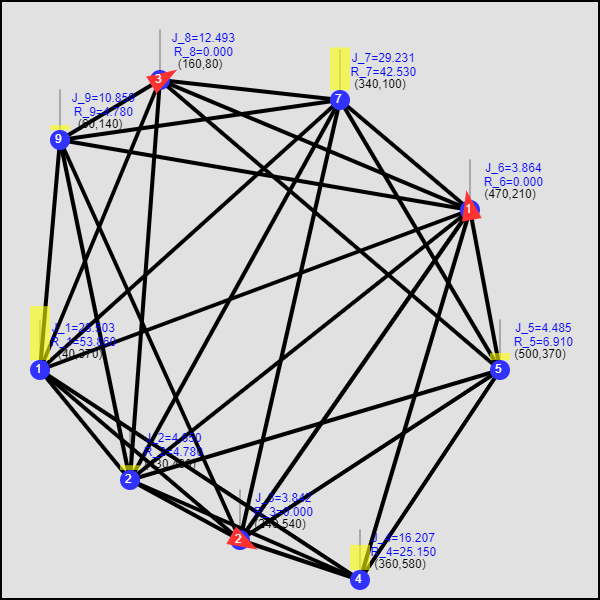}
         \caption{RHC: \\ \centering{$J_T= 114.1$.}}
     \end{subfigure}
     \begin{subfigure}[b]{0.24\columnwidth}
         \centering
         \includegraphics[width = \textwidth]{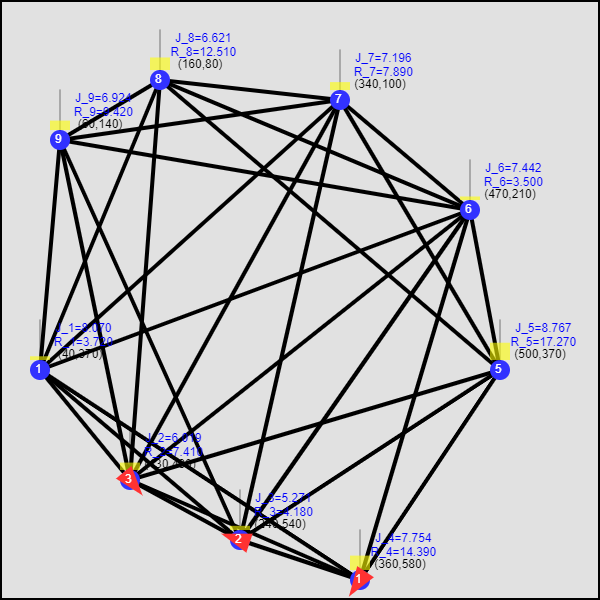}
         \caption{RHC$^\alpha$:\\ \centering{$J_T= \textbf{63.7}$}.}
     \end{subfigure}
     \begin{subfigure}[b]{0.24\columnwidth}
         \centering
         \includegraphics[width = \textwidth]{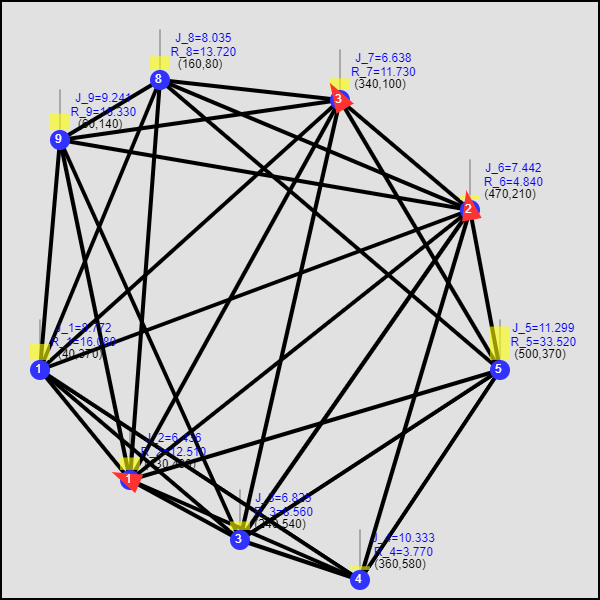}
         \caption{Ex-RHC$^{\alpha\beta}$:\\ \centering{$J_T=75.0$}.}
     \end{subfigure}
    \caption{Multi-agent simulation example 3 (MASE3).}
    \label{Fig:MASE3}
\end{figure}

\begin{figure}[!h]
     \centering
     \begin{subfigure}[b]{0.24\columnwidth}
         \centering
         \includegraphics[width = \textwidth]{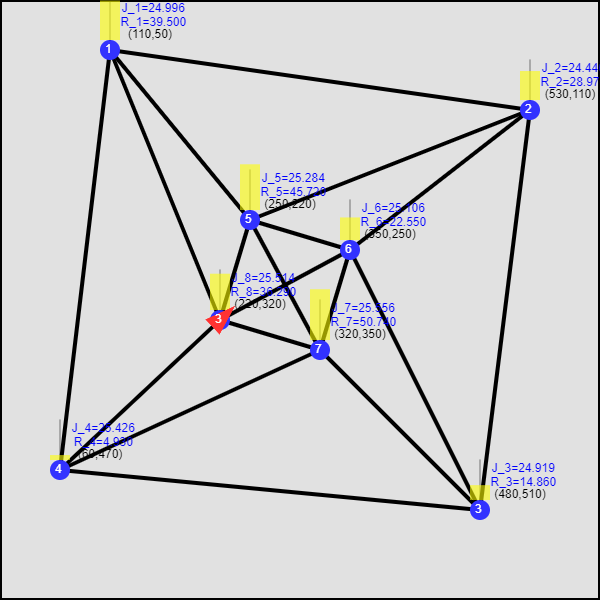}
         \caption{IPA-TCP: \\ \centering{$J_T= 201.3$.}}
     \end{subfigure}
     \begin{subfigure}[b]{0.24\columnwidth}
         \centering
         \includegraphics[width = \textwidth]{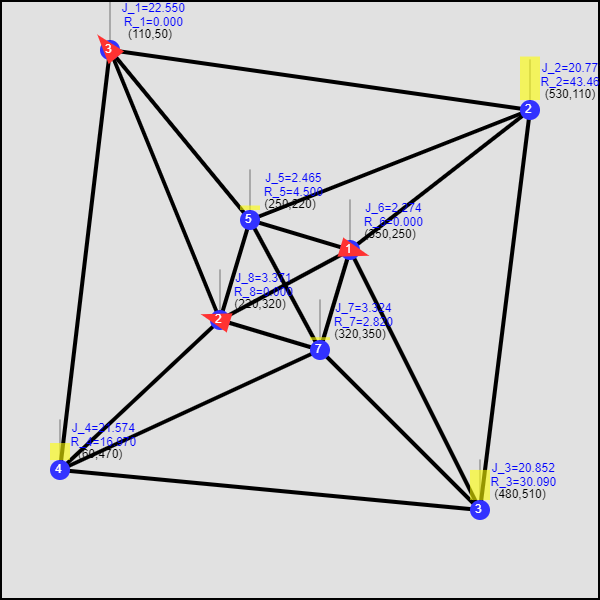}
         \caption{RHC: \\ \centering{$J_T= 97.2$.}}
     \end{subfigure}
     \begin{subfigure}[b]{0.24\columnwidth}
         \centering
         \includegraphics[width = \textwidth]{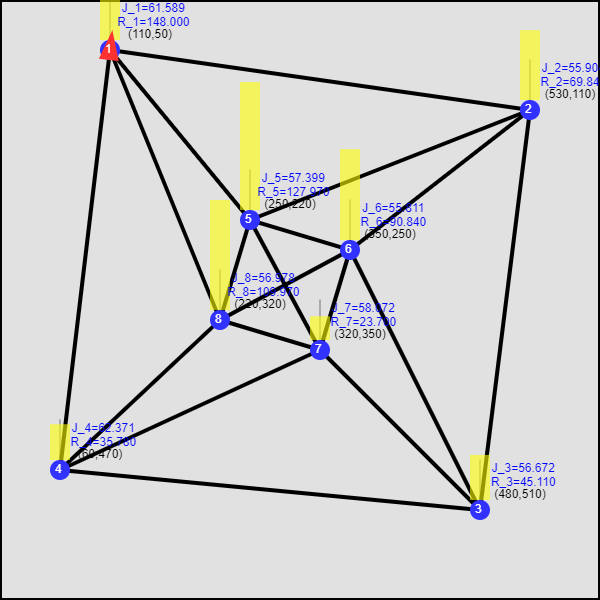}
         \caption{RHC$^\alpha$:\\ \centering{$J_T= \textbf{60.1}$}.}
     \end{subfigure}
     \begin{subfigure}[b]{0.24\columnwidth}
         \centering
         \includegraphics[width = \textwidth]{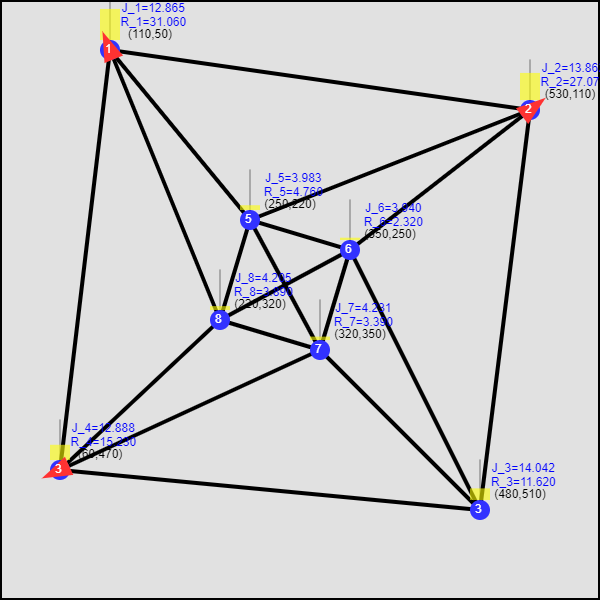}
         \caption{Ex-RHC$^{\alpha\beta}$:\\ \centering{$J_T=70.2$}.}
     \end{subfigure}
    \caption{Multi-agent simulation example 4 (MASE4).}
    \label{Fig:MASE4}
\end{figure}

\subsection{Impact of Using a Variable Planning Horizon}

For the problem configurations shown in Fig. \ref{Fig:SASE4} and \ref{Fig:MASE4}, the respective sub-figures in Fig. \ref{Fig:CostVsHorizonPlot} shows the performance of the RHC solution (i.e., $J_{T}$) under different upper-bound values on each planning horizon $w$ (i.e., $H$). In each case, the optimum upper-bound value (i.e., $H^{*}$) and its corresponding (minimum) cost value are indicated in the respective sub-figure caption. These results imply that having a large enough $H$ can directly give a performance level that is very closer to the optimum (within $1.1\%$). Hence, there is no evident importance in attempting to fine-tune the $H$ value. This highlights the impact of our main contribution of introducing a variable planning horizon (i.e., $w$ in \eqref{Eq:VariableHorizon}) so as to determine the optimal planning horizon at each instance where the RHCP is solved.

\begin{figure}[h]
\centering
\hfill\begin{subfigure}[b]{0.45\columnwidth}
\centering
\includegraphics[width = \textwidth]{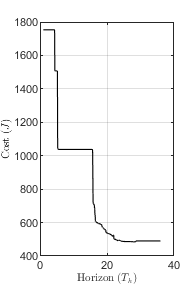}
\caption{SASE 4: Optimum at: \\
$H^* = 28.58$ with $J_T= 484.1$.\\
Yet, at: $H = 250$, $J_T= 490.4$.}
\end{subfigure}
\hfill\begin{subfigure}[b]{0.45\columnwidth}
\centering
\includegraphics[width = \textwidth]{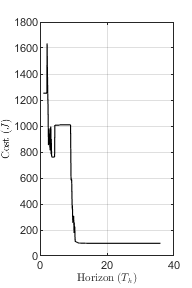}
\caption{MASE 4: Optimum at: \\
$H^* = 12.27$ with $J_T= 96.1$. \\
Yet, at: $H = 250$, $J_T= 97.2$.}
\end{subfigure}
\hfill\caption{Cost $J_{T}$ vs $H$ plot for SASE4 and MASE4.}%
\label{Fig:CostVsHorizonPlot}%
\end{figure}

\subsection{Robustness to the Randomness}

In this section, we introduce randomness to different aspects of the persistent monitoring system (e.g., system parameters, state dynamics or state information shared with neighbors) to determine the RHC version with superior robustness properties. In particular, we change the magnitude of the introduced randomness and explore how it affects the mean ($\mu$) and the variance ($\sigma^{2}$) of the performance $J_{T}$ (observed over 250 realizations) given by the RHC$^{\alpha}$ and Ex-RHC$^{\alpha\beta}$ methods for the MASE1 in Fig. \ref{Fig:MASE1}.

We consider five cases where we introduce randomness to: 
(\romannum{1}) Target uncertainty growth rate $A_i$,
(\romannum{2}) Maximum agent travel speed $V_{ij}$ on each trajectory segment $(i,j)\in\mathcal{E}$, 
(\romannum{3}) Target locations $Y_i$,
(\romannum{4}) Target uncertainty $R_{i}(t)$ trajectories of each target $i\in\mathcal{T}$, and
(\romannum{5}) Neighbor state information $\{R_{j}(t):j\in\mathcal{N}_{i}\}$ received at each target $i\in\mathcal{T}$ to solve the RHCP.
Details on how each parameter was perturbed and the observed respective results are as follows.


\paragraph{\textbf{Noise in $A_i$}}
The target uncertainty growth rate $A_i$ in \eqref{Eq:TargetDynamics} was replaced with $A_i\zeta_i(t)$ where $\zeta_i(t)\sim U[1-m,1+m]$. The noise magnitude $m$ was varied between $0.5$ to $5$ as shown in Fig. \ref{Fig:NoiseA_i} and observed that average performance deviation from its nominal value is very small (remains within $1.85\%$) for $m\in[0,5]$. The corresponding variance values observed are also bounded at reasonable levels for $m\in[0,5]$.

\begin{figure}[!h]
\centering
\begin{subfigure}[b]{0.48\columnwidth}
\centering
\includegraphics[width = 1.5in]{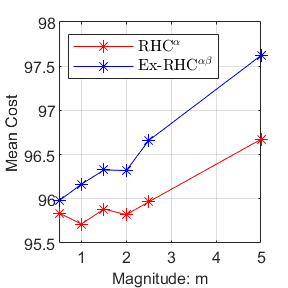}
\caption{Mean cost vs $m$}
\end{subfigure}
\begin{subfigure}[b]{0.48\columnwidth}
\centering
\includegraphics[width = 1.5in]{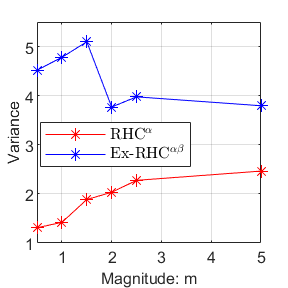}
\caption{Variance vs $m$}
\end{subfigure}
\caption{The effect of random noise in $A_i$.}
\label{Fig:NoiseA_i}
\end{figure}

\paragraph{\textbf{Noise in $V_{ij}$}}

Recall that we have used a first-order agent motion model with a fixed maximum velocity $V_{ij}$ for each trajectory segment $(i,j)\in\mathcal{E}$. In such a model, as proved in \cite{Zhou2019}, the optimal transit time is $\rho_{ij}=\frac{l_{ij}}{V_{ij}}$ where $l_{ij}$ is the length of the trajectory segment $(i,j)$. In this experiment, we have replaced $V_{ij}$ with $V_{ij}\zeta_{ij}$ where $\zeta_{ij}\sim U[1-m,1+m]$ and renewed $\zeta_{ij}$ whenever an agent has decided to travel on a trajectory segment $(i,j)$. The noise magnitude $m$ was selected such that $m\in\lbrack0,1]$. According to the results shown in Fig. \ref{Fig:NoiseV_ij}, in both RHC methods explored, the average performance deviation observed from its nominal value is less than $17.61\%$ for $m\leq0.6$. Moreover, for such $m$ values, the corresponding variance values observed are also reasonably small (i.e., such that $\sigma/\mu\leq0.038$).

\begin{figure}[!t]
     \centering
     \begin{subfigure}[b]{0.48\columnwidth}
         \centering
         \includegraphics[width = 1.5in]{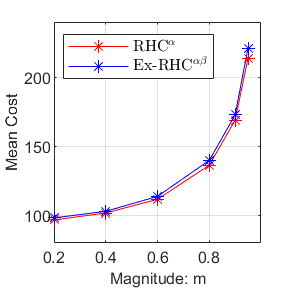}
         \caption{Mean cost vs $m$}
     \end{subfigure}
     \begin{subfigure}[b]{0.48\columnwidth}
         \centering
         \includegraphics[width = 1.5in]{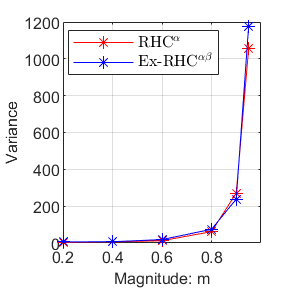}
         \caption{Variance vs $m$}
     \end{subfigure}
     \caption{The effect of random noise in $V_{ij}$.}
    \label{Fig:NoiseV_ij}
\end{figure}

\paragraph{\textbf{Noise in $Y_{i}$}}
Note that throughout this work we assumed the target locations $Y_i\in \R^n, \forall i\in\mathcal{T}$ to be a fixed set of locations. However, we now assume $Y_i$ changes dynamically following $\ddot{Y}_i(t) =  \zeta_i(t)$ where each component of $\zeta_i(t)$ is drawn from $U[-m,+m]$. We further constrain $\Vert Y_i(t)-Y_i(0)\Vert \leq R$ using standard projection techniques and we set the boundary radius $R=20$ and $Y_i(0) = Y_i$ (i.e., the nominal target location).

In this scenario, when an agent solves the RHCP, it uses the transit times based on current target locations (i.e., $\rho_{ij}=\frac{\Vert Y_i(t)-Y_j(t)\Vert}{V_{ij}}, \forall j\in \mathcal{N}_i$). Following the first order agent model, during the traveling period (say from target $i$ to target $j$) an agent $a\in\mathcal{A}$ uses $\dot{s}_a(t)=V_{ij}\cdot \frac{Y_j(t)-s_a(t)}{\Vert Y_j(t)-s_a(t)\Vert}$ until $\Vert Y_j(t)-s_a(t) \Vert$ is sufficiently small.

The noise magnitude $m$ is selected such that $m \in [0,100]$. From the results shown in Fig. \ref{Fig:NoiseV_ij}, the average performance deviation from its nominal value is again relatively small (less than $5.48\%$) for all $m\leq [0,100]$. Despite their gradual growth with $m$, the corresponding variance values remain relatively small.

\begin{figure}[t]
     \centering
     \begin{subfigure}[b]{0.48\columnwidth}
         \centering
         \includegraphics[width = 1.5in]{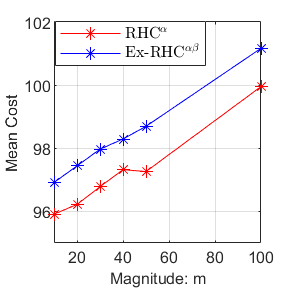}
         \caption{Mean cost vs $m$}
     \end{subfigure}
     \begin{subfigure}[b]{0.48\columnwidth}
         \centering
         \includegraphics[width = 1.5in]{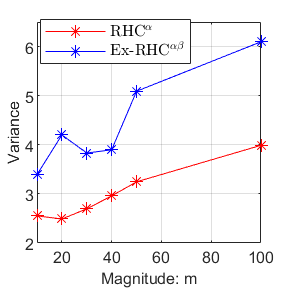}
         \caption{Variance vs $m$}
     \end{subfigure}
     \caption{The effect of random noise in $Y_i$.}
    \label{Fig:NoiseY_i}
\end{figure}

\paragraph{\textbf{Noise in $R_{i}(t)$}}

We next assumed that each target uncertainty state trajectory $R_{i}(t),i\in\mathcal{T}$ following the dynamics in \eqref{Eq:TargetDynamics} experiences occasional random discrete perturbations of the form $\zeta_{i}(t)1\{t=t_{i}^{e}\}$ where $\zeta_{i}(t)\sim U[-m,+m]$ and $t_{i}^{e}$ is a Poisson arrival process with a mean inter-event interval $\lambda$. The noise magnitude $m$ and the inter-event interval value $\lambda$ are such that $m\in\lbrack0,50]$ and $\lambda\in\{25,50,100\}$. The observed results are shown in Fig. \ref{Fig:NoiseR_i}, where we see that both the average performance and the corresponding variance show a slow growth with $m$ when $\lambda=100$. However, as expected this growth rate increases when smaller $\lambda$ values (i.e., higher rates) are used.

\begin{figure}[!t]
     \centering
     \begin{subfigure}[b]{0.48\columnwidth}
         \centering
         \includegraphics[width = 1.6in]{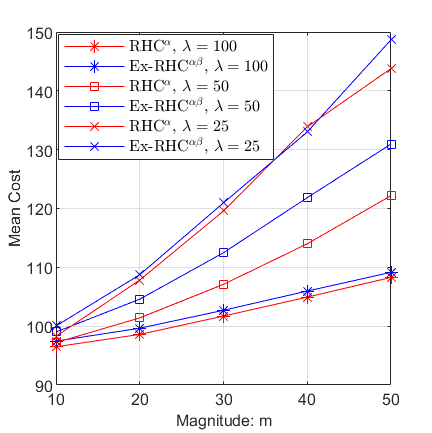}
         \caption{Mean cost vs $m$}
     \end{subfigure}
     \begin{subfigure}[b]{0.48\columnwidth}
         \centering
         \includegraphics[width = 1.6in]{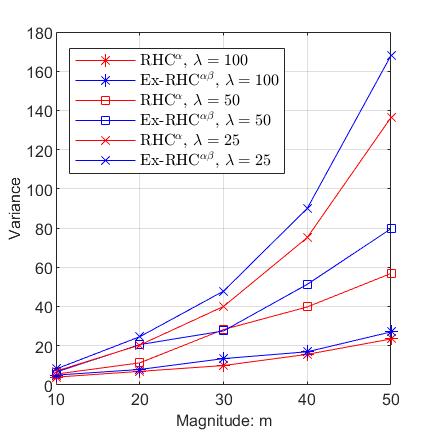}
         \caption{Variance vs $m$}
     \end{subfigure}
     \caption{The effect of random noise in $R_i$.}
    \label{Fig:NoiseR_i}
\end{figure}

It is important to point out that this example illustrates the benefit of the event-driven nature of the proposed controllers by triggering a random event at each random time $t_{i}^{e}$ affecting any agent residing at a target within the neighborhood of $i$: each such agent is forced to re-solve its RHCP. This allows the agents to react to the exogenously introduced random state perturbations. The impact of this advantage is illustrated in Fig. \ref{Fig:NoiseR_i2} (for the case where $\lambda=25$).

\begin{figure}[!t]
     \centering
     \begin{subfigure}[b]{0.48\columnwidth}
         \centering
         \includegraphics[width = 1.5in]{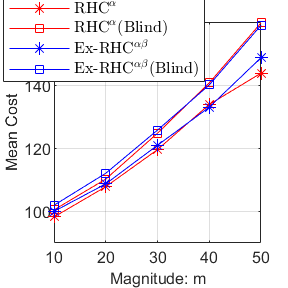}
         \caption{Mean cost vs $m$}
     \end{subfigure}
     \begin{subfigure}[b]{0.48\columnwidth}
         \centering
         \includegraphics[width = 1.5in]{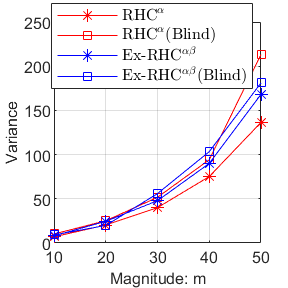}
         \caption{Variance vs $m$}
     \end{subfigure}
     \caption{The effect of random noise in $R_i$: A comparison when the agents are allowed to react and not allowed to react (i.e., ``Blind'') to the random events occurring in the neighborhood.}
    \label{Fig:NoiseR_i2}
\end{figure}

\paragraph{\textbf{Noise in the State Information of Neighbors: $\{R_{j}(t): j\in\mathcal{N}_{i}\}$}}

In this case, we assumed that each agent residing at some target $i$ when solving the RHCP at some time $t$ receives its neighboring target state information through a noisy communication channel which superimposes its actual value $R_{j}(t)$ (which follows \eqref{Eq:TargetDynamics}) with a random noise process $\zeta_{j}(t)\sim U[-m,+m]$ independently for each $j\in\mathcal{N}_{i}$. The noise magnitude $m$ was varied between $4$ to $28$ to generate Fig. \ref{Fig:NoiseR_j} where, as can be seen, the average performance deteriorates linearly with $m$ but at a lower $1.0\%$ rates per unit of $m$.

\begin{figure}[!t]
     \centering
     \begin{subfigure}[b]{0.48\columnwidth}
         \centering
         \includegraphics[width = 1.5in]{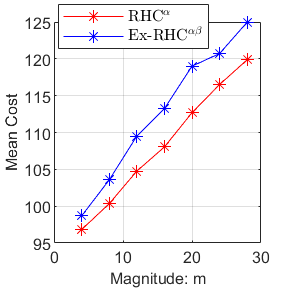}
         \caption{Mean cost vs $m$}
     \end{subfigure}
     \begin{subfigure}[b]{0.48\columnwidth}
         \centering
         \includegraphics[width = 1.5in]{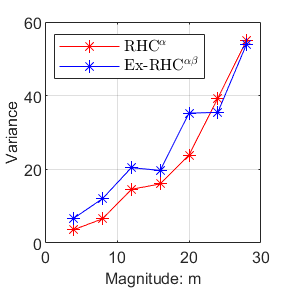}
         \caption{Variance vs $m$}
     \end{subfigure}
     \caption{The effect of random noise in the state information of neighbors $\{R_j(t): j\in \mathcal{N}_i\}$.}
    \label{Fig:NoiseR_j}
\end{figure}

In all, both RHC versions, RHC$^{\alpha}$ and Ex-RHC$^{\alpha\beta}$, while showing superior performance levels when compared to the IPA-TCP method, they also show resilience to a diverse form of random factors. Across all experiments, the RHC$^{\alpha}$ method has shown superior robustness to these factors compared to the Ex-RHC$^{\alpha\beta}$ method.

\subsection{Some Practical Considerations}

The proposed RHC method and related theoretical results are not limited by the nature of the graph topology. However, if the graph is not \emph{connected}, the RHC method may result in behaviors such as an agent getting trapped in a region of the graph, leading to lower performance levels, depending on the particular graph topology.

In our PMN problem setup, we assumed the transit time values $\rho_{ij},\ (i,j)\in\mathcal{E}$ and the trajectory segments $(i,j)\in\mathcal{E}$ to be predefined. However, if the problem setup allows selecting these quantities, based on the target state dynamics in \eqref{Eq:TargetDynamics} and the global objective in \eqref{Eq:MainObjective}, each $\rho_{ij}$ should represent the \textquotedblleft minimum time to reach target $j$ from target $i$\textquotedblright\ and $(i,j)$ should represent the corresponding agent trajectory. Clearly, such choices will depend on the agent dynamic model, possible obstacles in the actual mission space and the target locations $Y_{i}$ and $Y_{j}$.

Even though the covering events (defined in Section \ref{SubSec:RHCIntro}) can effectively enforce the no-simultaneous-target-sharing policy \eqref{Eq:NoTargetSharing}, still, multiple agents can end up in the same target at the same time due to: (i) initial conditions, (ii) symmetric problem configurations and (iii) communication delays. In such a rare situation, each agent residing on a common target $i\in\mathcal{T}$ can solve its RHCP limiting it to a partition of the current neighborhood $\bar{\mathcal{N}}_{i}$, so as to prevent simultaneous target sharing at its next-visit target.

\section{Conclusion}\label{Sec:Conclusion} 
Optimal multi-agent persistent monitoring problem defined on a set of targets interconnected according to a fixed network topology is considered in this paper. In contrast to existing computationally expensive and slow threshold-based parametric control solutions, a novel computationally efficient, parameter-free, and robust event-driven receding horizon control solution is proposed. Specifically, we simultaneously determine both an optimal planning horizon and a finite sequence of optimal trajectory decisions for each agent at different discrete event times on its trajectory. Numerical results show significant improvements compared to existing gradient-based parametric control solutions, as well as robustness to various forms of randomness in the system. Ongoing work is aimed to extend the proposed controller to PMN problems with variable transit times dictated by different higher order agent dynamic models.



\section*{APPENDIX}

\subsection{Constrained Optimization of Bivariate Rational Functions}
\label{App:A}

\paragraph{\textbf{Convexity of Rational Functions}}
Consider a \emph{rational function} $h:\mathbb{R}\rightarrow\mathbb{R}$ of the form $h(r)=\frac{f(r)}{g(r)}$ and assume $g(r)>0\ \forall r\in\mathcal{U}\subseteq\mathbb{R}$ where $\mathcal{U}$ is a closed interval. In the following the argument of $f(r),g(r)$ or $h(r)$ is omitted for notational convenience. Also, the notation \textquotedblleft\ $^{\prime}$ \textquotedblright\ is used to denote the derivative of a function with respect to $r$.

\begin{lemma}
\label{Lm:ConvexityOfRationalFunction} Whenever polynomials $g(r)$ and $f(r)$ satisfy
\begin{equation}
g[gf^{\prime\prime\prime}-fg^{\prime\prime\prime}]-3g^{\prime\prime
}[gf^{\prime}-fg^{\prime}]=0,\ \forall r\in\mathcal{U}%
,\label{Eq:ConvexityOfRationalFunction}%
\end{equation}
$h(r)$ is convex (or concave) on $\mathcal{U}$ if $\Delta_{h}(r_{0})>0$ (or $\Delta_{h}(r_{0})<0$) where $r_{0}\in\mathcal{U}$ and 
\begin{equation}
\Delta_{h}(r)\triangleq g[gf^{\prime\prime}-fg^{\prime\prime}]-2g^{\prime
}[gf^{\prime}-fg^{\prime}].\label{Eq:Deltah}
\end{equation}
\end{lemma}

\emph{Proof: } The first and second order derivatives of $h(r)$ are 
\[
h^{\prime}=\frac{gf^{\prime}-fg^{\prime}}{g^{2}}\mbox{ and }h^{\prime\prime
}=\frac{g[gf^{\prime\prime}-fg^{\prime\prime}]-2g^{\prime}[gf^{\prime
}-fg^{\prime}]}{g^{3}}.
\]
Note that $h^{\prime\prime}(r)=\frac{\Delta_{h}(r)}{g^{3}(r)}$ and $g^{3}(r)>0\ \forall r\in\mathcal{U}$. Therefore, convexity of $h(r)$ will only depend on the condition:
\[
h^{\prime\prime}(r)>0,\ \forall r\in\mathcal{U}\iff\Delta_{h}(r)>0,\forall
r\in\mathcal{U}.
\]
This condition is easily seen to be satisfied whenever
\[
\Delta_{h}(r_{0})>0\mbox{ for some }r_{0}\in\mathcal{U}\mbox{ and }\Delta
_{h}^{\prime}(r)=0\mbox{ for all }r\in\mathcal{U}.
\]
Finally, evaluating $\Delta_{h}^{\prime}(r)$ yields the expression in \eqref{Eq:ConvexityOfRationalFunction}
\[
\Delta_{h}^{\prime}(r)=g[gf^{\prime\prime\prime}-fg^{\prime\prime\prime
}]-3g^{\prime\prime}[gf^{\prime}-fg^{\prime}],
\]
which completes the proof. \hfill$\blacksquare$

\begin{remark}\label{Rm:ExampleCase}
According to Lemma \ref{Lm:ConvexityOfRationalFunction}, the condition in \eqref{Eq:ConvexityOfRationalFunction} along with $\Delta_{h}(r_{0})>0$ (or $\Delta_{h}(r_{0})<0$) for some $r_{0}\in\mathcal{U}$ is sufficient to determine the convexity (or concavity) of $h(r)$ on $\mathcal{U}$. As an example, \eqref{Eq:ConvexityOfRationalFunction} is satisfied whenever the rational function $h(r)$ has a denominator polynomial $g(r)$ of first degree and a numerator polynomial $f(r)$ of second degree. In such a case, the convexity/concavity of $h(r)$ over $\mathcal{U}$ can be identified by simply evaluating the sign of $\Delta_{h}(r)$ at a convenient $r=r_{0}\in\mathcal{U}$ point.
\end{remark}

\paragraph{\textbf{Constrained Minimization of $h(r)$}}
Assume $h(r)$ to be a rational function which satisfies the conditions discussed above: $g(r)>0,\ \Delta_{h}^{\prime}(r)=0\ \forall r\in \mathcal{U}\subseteq\mathbb{R}$. Further, assume the signs of $\Delta_{h}(r_{0})$ and $h^{\prime}(r_{0})$ are known at some point of interest $r=r_{0}\in\mathcal{U}$ (recall that the sign of $\Delta_{h}(r_{0})$ mimics the sign of $h^{\prime\prime}(r),r\in\mathcal{U}$). According to Lemma \ref{Lm:ConvexityOfRationalFunction}, the latter assumption fully determines the convexity (or concavity) of $h(r)$ on $\mathcal{U}$ and its gradient direction at $r=r_{0}$, respectively. Now, consider the following optimization problem:
\begin{equation}
\begin{aligned}
r^* = &\underset{r}{\mathrm{argmin}}\ h(r), \\
&r_0 \leq r \leq r_1
\end{aligned}
\label{Eq:ConvexOptimizationSingleVar}%
\end{equation}
where $[r_{0},r_{1}]\subseteq\mathcal{U}$. A critical $r$ value $r=r^{\#}$
(which is important to the analysis) is defined as
\begin{equation*}
r^{\#}\triangleq%
\begin{cases}
\{r:h^{\prime}(r)=0,r>r_{0}\} & \mbox{if }\Delta_{h}(r_{0})>0\And h^{\prime
}(r_{0})<0,\\
\{r:h(r)=h(r_{0}),r>r_{0}\} & \mbox{if }\Delta_{h}(r_{0})<0\And h^{\prime
}(r_{0})>0.
\end{cases}
\end{equation*}
Note that the two cases considered above are the only ones where a stationary point of $h(r)$ could occur for some $r>r_{0},\ r\in\mathcal{U}$ (see also Fig. \ref{Fig:LineGraphs}).

\begin{lemma}
\label{Lm:ConvexOptimizationSingleVar} 
The optimal solution to \eqref{Eq:ConvexOptimizationSingleVar} is as follows:
\newline If $\Delta
_{h}(r_{0})<0 \And h^{\prime}(r_{0})>0$
\[
r^{*} =
\begin{cases}
r_{1} \mbox{ if } r_{1} > r^{\#}\\
r_{0} \mbox{ otherwise,}
\end{cases}
\]
else if $\Delta_{h}(r_{0})>0 \And h^{\prime}(r_{0})<0$
\[
r^{*} =
\begin{cases}
r^{\#} \mbox{ if } r_{1} > r^{\#}\\
r_{1} \mbox{ otherwise,}
\end{cases}
\]
otherwise,
\[
r^{*} =
\begin{cases}
r_{0} \mbox{ if } \Delta_{h}(r_{0}) \geq0 \And h^{\prime}(r_{0}) \geq0\\
r_{1} \mbox{ otherwise.}
\end{cases}
\]
\end{lemma}

\emph{Proof: } The proof easily follows by inspection of all cases shown in Fig. \ref{Fig:LineGraphs}. \hfill$\blacksquare$

\begin{figure}[h]
\centering
\includegraphics[width=2.7in]{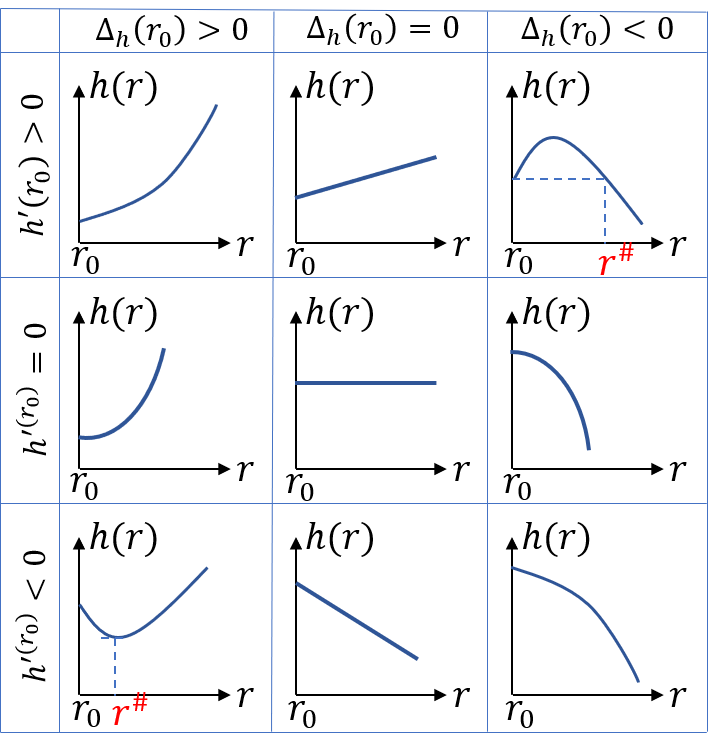}  \caption{Graphs of possible $\{h(r): r \geq r_{0},\ r \in\mathcal{U}\}$ profiles for different cases of $h^{\prime}(r_{0})$ and $\Delta_{h}(r_{0})$ (recall $sgn(\Delta_{h}(r_{0}))=sgn(h^{\prime\prime}(r))$ determines the convexity or concavity).}
\label{Fig:LineGraphs}%
\end{figure}

In essence, an optimization problem of the form \eqref{Eq:ConvexOptimizationSingleVar} can be solved based exclusively on the values of $h^{\prime}(r_{0}),\ \Delta_{h}(r_{0})$ and $r^{\#}$. Note that $r^{\#}$ is only required in two special cases and for the application example mentioned in Remark \ref{Rm:ExampleCase}, it can be obtained simply by solving for the roots of a quadratic expression (single variable).

\paragraph{\textbf{Bivariate Rational Functions}}
Next, consider the class of \emph{bivariate rational functions} that can be represented by a function $H:\mathbb{R}_{+}^{2}\rightarrow\mathbb{R}$ of the
form
\begin{equation}
H(x,y)=\frac{F(x,y)}{G(x,y)}=\frac{C_{1}x^{2}+C_{2}y^{2}+C_{3}xy+C_{4}%
x+C_{5}y+C_{6}}{C_{7}x+C_{8}y+C_{9}},\label{Eq:BivariateRational}%
\end{equation}
where the coefficients $C_{1},\ldots,C_{9}$) are known scalar constants with $C_{7}\geq0,\ C_{8}\geq0$ and $C_{9}>0$. Note that the range space of $H(x,y)$ is limited to the non-negative orthant of $\mathbb{R}^{2}$ (denoted by $\mathbb{R}_{+}^{2}$).

Developing conditions for the convexity of $H(x,y)$ is a complicated task. Even if such conditions were derived, interpreting them and exploiting them to solve a two-dimensional constrained optimization problem that involves minimizing $H(x,y)$ (analogous to \eqref{Eq:ConvexOptimizationSingleVar}) is challenging. To address this, the behavior of $H(x,y)$ is next studied along a generic line segment of the form $y=mx+b$ starting at some point $(x_{0},y_{0})\in\mathbb{R}_{+}^{2}$ as shown in Fig. \ref{Fig:LineSampling}. A parameter $r$ is used to represent a generic location $(x_{r},y_{r})$ on this line as $(x_{r},y_{r})=(x_{0}+r,\ y_{0}+mr)$ where $r$ is introduced exploiting the gradient $m$ of the line segment: 
\begin{equation}
\frac{y_{r}-y_{0}}{x_{r}-x_{0}}=m\implies\frac{y_{r}-y_{0}}{m}=\frac
{x_{r}-x_{0}}{1}=r.\label{Eq:Parametrization1}%
\end{equation}
A rational function $h(r)$ can now be defined as  \begin{equation}
h(r)\triangleq H(x_{0}+r,\ y_{0}+mr)=\frac{F(x_{0}+r,\ y_{0}+mr)}{G(x_{0}+r,\ y_{0}+mr)}=\frac{f(r)}{g(r)},
\label{Eq:HxyOnLine}
\end{equation}
to represent $H(x,y)$ along the line segment of interest.

The parameter $r$ is constrained such that $r\in\mathcal{U}\triangleq \lbrack-x_{0},\frac{-y_{0}}{m}]$ to limit the line segment to $\mathbb{R}_{+}^{2}$. This allows $h(r)$ to fall directly into the category of rational functions discussed in Lemma \ref{Lm:ConvexityOfRationalFunction} and in Remark \ref{Rm:ExampleCase}.

\begin{figure}[!t]
     \centering
     \begin{subfigure}[b]{0.48\columnwidth}
        \centering
        \includegraphics[width=1.6in]{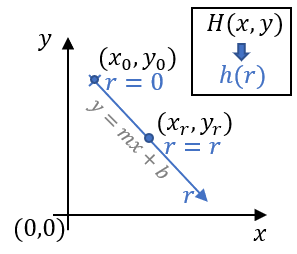}
        \caption{}
        \label{Fig:LineSampling}
     \end{subfigure}
     \begin{subfigure}[b]{0.48\columnwidth}
        \centering
        \includegraphics[width=1.6in]{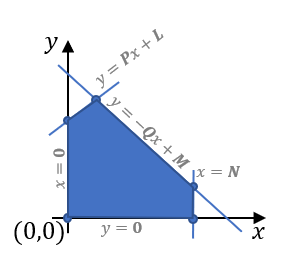}
        \caption{}
        \label{Fig:FeasibleSpace}%
     \end{subfigure}
     \caption{(a) $H(x,y)$ along the line $y=mx+b$,\ \ 
     (b) Feasible space for $H(x,y)$ in         \eqref{Eq:ConstrainedOptimizationBivariate}.}
    \label{Fig:Hxy}
    \vspace{-4mm}
\end{figure}

\begin{theorem}
\label{Th:ConvexityOfBivariateRational} 
The rational function $h(r),\ r \in \mathcal{U}$ defined in \eqref{Eq:HxyOnLine} is convex (or concave) if $\Delta_{h}(r_{0}) > 0$ (or $\Delta_{h}(r_{0}) < 0$), where $r_{0}\in\mathcal{U}$ and $\Delta_{h}(r)$ is defined in \eqref{Eq:Deltah}.
\end{theorem}

\emph{Proof: } According to \eqref{Eq:HxyOnLine} and $\mathcal{U}$ defined above, the denominator polynomial $g(r) = G(x_{0}+r,y_{0}+mr) > 0$ for all $r\in\mathcal{U}$ as $C_{7}\geq0,\ C_{8}\geq0$ and $C_{9}>0$ in \eqref{Eq:BivariateRational}.

Since $g(r)$ and $f(r)$ are polynomials of degree 1 and 2 respectively, they satisfy condition \eqref{Eq:ConvexityOfRationalFunction}. Thus, Lemma \ref{Lm:ConvexityOfRationalFunction} is applicable for $h(r)$ in \eqref{Eq:HxyOnLine} and its convexity will depend on the condition $\Delta_{h}(r_{0})>0$. \hfill$\blacksquare$

It is worth pointing out that $\Delta_{h}(r)$ is in fact independent of $r$ as $\Delta_{h}^{\prime}(r)=0,\forall r\in\mathcal{U}$ (see the last step of the proof of Lemma \ref{Lm:ConvexityOfRationalFunction} and \eqref{Eq:ConvexityOfRationalFunction}). However, it will depend on other parameters contained in \eqref{Eq:BivariateRational} including $x_{0},y_{0}$ and $m$. For example, when the line segment defined by $x_{0}=0,y_{0}=0,m=0$ (i.e., the $x$-axis) is used, $\Delta_{h}(r)=2C_{6}C_{7}^{2}-2C_{4}C_{7}C_{9}+2C_{1}C_{9}^{2},\forall r\in\mathbb{R}_{+}.$

In the introduced parameterization scheme, the parameter $r$ represents the distance along the $x$ axis from $x_{0}$ (projected from the line segment $y=mx+b$). However, if $H(x,y)$ needs to be studied along the $y$ axis (from $y_{0}$ projected from a line segment $x=ny+c$), then using
\begin{equation}
\frac{y_{r}-y_{0}}{x_{r}-x_{0}}=\frac{1}{n}\implies\frac{y_{r}-y_{0}}{1}=\frac{x_{r}-x_{0}}{n}=r,\label{Eq:Parametrization2}%
\end{equation}
is more appropriate as it gives $(x_{r},y_{r})=(x_{0}+nr,y_{0}+r)$.

Theorem \ref{Th:ConvexityOfBivariateRational} enables determining the optimal $H(x,y)$ value along a known line segment (on $\mathbb{R}_{+}^{2}$) using Lemma \ref{Lm:ConvexOptimizationSingleVar} for a problem of the form \eqref{Eq:ConvexOptimizationSingleVar}. This capability is exploited next.

\paragraph{\textbf{Constrained Minimization of $H(x,y)$}}

The main objective of this discussion is to obtain a closed form solution to a constrained optimization problem of the form
\begin{equation}
\label{Eq:ConstrainedOptimizationBivariate}\begin{aligned} (x^*, y^*) = &\underset{(x,y)}{\mathrm{argmin}}\ H(x,y)\\ & 0\leq x \leq \mathbf{N},\\ & 0 \leq y \leq \min\{\mathbf{P}x+\mathbf{L},\, -\mathbf{Q}x+\mathbf{M}\}, \end{aligned}
\end{equation}
where $H(x,y)$ is a known bivariate rational function of the form \eqref{Eq:BivariateRational} and $\mathbf{P},\mathbf{Q},\mathbf{L},\mathbf{M}$ are known positive (scalar) constants. These constraints define a convex 2-Polytope as shown in Fig. \ref{Fig:FeasibleSpace}. The steps to solve the above problem are discussed next.

\paragraph*{\textbf{- Step 1}}
The unconstrained version of \eqref{Eq:ConstrainedOptimizationBivariate} is considered first. This is solved using the KKT necessary conditions \cite{Bertsekas2016nonlinear}, which reveal two equations of generic conics \cite{Rosenberg}. Therefore, the stationary points of $H(x,y)$ lie at the (four) intersection points of those two conics. The problem of determining the intersection of two conics boils down to solving a quartic equation, which has a well-known closed-form solution \cite{Auckly2007}. These (four) solutions are computed and stored in a \emph{solution pool} if they satisfy the problem constraints.

\paragraph*{\textbf{- Step 2}}
Next, the constrained version of \eqref{Eq:ConstrainedOptimizationBivariate} is considered. In such a case, it is possible for $(x^{\ast},y^{\ast})$ to lie on a constraint boundary. To capture such situations, $H(x,y)$ is optimized along each of the boundary line segments of the feasible space (there are five of them as shown in Fig. \ref{Fig:FeasibleSpace}).

On a selected boundary line segment, the first step is to parameterize $H(x,y)$ to obtain a single variable rational function $h(r)$ (following either \eqref{Eq:Parametrization1} or \eqref{Eq:Parametrization2}). Then, the next step is to solve the resulting convex (or concave) optimization problem (of the form \eqref{Eq:ConvexOptimizationSingleVar}) using Lemma \ref{Lm:ConvexOptimizationSingleVar}. Note that this is enabled by Theorem \ref{Th:ConvexityOfBivariateRational}. Finally, the obtained optimal solution is added to the solution pool from \textbf{Step 1}. 

\paragraph*{\textbf{- Step 3}}

The final step is to pick the best solution out of the solution pool (which only contains at most nine candidates solutions). Therefore, this is achieved by directly evaluating $H(x,y)$ and comparing all candidate solutions to each other.

This approach is computationally cheap, accurate and provides the global optimal solution compared to gradient-based methods which are susceptible to local optima. This concludes the discussion on how to solve a generic problem of the form \eqref{Eq:ConstrainedOptimizationBivariate}.

\subsection{Omitting the Denominator of the RHCP Objective Function} 
In this appendix, we consider the case where the RHCP objective function $J_H$ takes the form $$J_H(X_i(t),U_{ij};H) = \bar{J}_i(t,t+w).$$ 
Note that in above equation, the denominator term $w$ found in the original definition of $J_H$ in  \eqref{Eq:RHCNewChoices} is omitted. The main focus here is given to the \textbf{RHCP3} where now the objective function $J_H(u_j,v_j)$ takes the form
$$J_H(u_j,v_j) = C_1u_j^2 + C_2v_j^2 + C_3u_jv_j +C_4 u_j + C_5v_j + C_6.$$ 
Each coefficient $C_i$ for all $i$ above is same as in \eqref{Eq:OP3ObjectiveSimplified}. In this appendix, it is shown that the above objective function leads to a spurious policy for $j^*$ in \eqref{Eq:RHCGenSolStep2}.

The first step of \textbf{RHCP3} (i.e., \eqref{Eq:RHCGenSolStep1}) can be stated as:
\begin{equation}\label{Eq:OP3_FormalOld}
    \begin{aligned}
    (u_j^*,v_j^*) =&\ \underset{(u_j,v_j)}{\arg\min}\ {J_H(u_j,v_j)}\\
    &(u_j,v_j) \mbox{ s.t. \eqref{Eq:Constraints1}}.  
    \end{aligned}
\end{equation}
where $(u_j,v_j)\in\mathbb{U}_1$ or $(u_j,v_j)\in \mathbb{U}_2$ as in \eqref{Eq:Constraints1}.

\paragraph{\textbf{Solving \eqref{Eq:OP3_FormalOld} for Optimal Control }$(u_j^*,v_j^*)$ }

\paragraph*{\textbf{- Case 1}} $(u_j,v_j) \in \mathbb{U}_1 = \{0\leq u_{j}\leq\bar{u}_{j},\ v_{j}=0\}$ in \eqref{Eq:Constraints1}: Then, $v_j^* = 0$ and \eqref{Eq:OP3_FormalOld} takes the form: 
\begin{equation}\label{Eq:OP3_FormalOldPart1}
    \begin{aligned}
    u_j^* =\  &\underset{u_j}{\arg\min}\ {J_H(u_j,0)}\\
    &0 \leq u_j \leq \bar{u}_j.  
    \end{aligned}
\end{equation}

\begin{lemma}
\label{Lm:OP3_FormalOldPart1}
The optimal solution for \eqref{Eq:OP3_FormalOldPart1} is
\begin{equation}
    u_j^* = 0.
\end{equation}
\end{lemma}

\begin{proof}
Substituting $v_j=0$ in \eqref{Eq:OP3ObjectiveSimplified} gives $J_H(u_j,0)$ as  
\begin{equation}\label{Eq:Lm_OP3_FormalPart1_PfS0}
    J_H(u_j,0) = C_1u_j^2 + C_4u_j + C_6.
\end{equation}
Recall $C_4\geq0$, $C_6 \geq 0$ and $C_1 = \frac{1}{2}(\bar{A}-B_j)$.

First, consider the case where $C_1 = 0$. Then, $J_H(u_j,0)$ is linear in $u_j$. Also it will has a non-negative gradient as $C_4\geq0$. Therefore, when $C_1 = 0$, clearly $u_j^*=0$.

However, note that when $C_1\neq 0$, the unconstrained optimum of $J_H(u_j,0)$ is at $u_j = u_j^{\#} = \frac{-C_4}{2C_1}$ (using calculus). Also note that due to the quadratic nature of $J_H(u_j,0)$, it should be symmetric around $u_j=u_j^{\#}$.

As the second case, consider case where $C_1>0$. Then, $J_H(u_j,0)$ is convex and $u_j^{\#} \leq 0$. Therefore, when $C_1>0$, $u_j^*=0$. 

Finally, consider the case where $C_1<0$. Then , $J_H(u_j,0)$ is concave and $u_j^{\#} \geq 0$. In this case, using the aforementioned symmetry, the constrained optimum $u_j^*$ can be written as
\begin{equation}\label{Eq:Lm_OP3_FormalPart1_PfS1}
    u_j^* = 
    \begin{cases}
    \bar{u}_j &\mbox{ if }  2u_j^{\#} < \bar{u}_j \\
    0 &\mbox{ otherwise,}
    \end{cases}
\end{equation}
Now, it is required to prove that the condition $2u_j^{\#} < \bar{u}_j$ never occurs (whenever $C_1<0$). Using \eqref{Eq:OP3ObjectiveSimplified}, $u_j^{\#}$ can be written as,
\begin{equation}\label{Eq:Lm_OP3_FormalPart1_PfS2}
    u_j^{\#} = \frac{-C_4}{2C_1} = \frac{\bar{R}(t)+\bar{A}\rho_{ij}}{B_j-\bar{A}}.
\end{equation}
Note that $C_1<0 \iff B_j \geq \bar{A}$. Also from \eqref{Eq:OP3ObjectiveSimplified}, $R_i(t) \leq \bar{R}(t)$ and $A_i \leq \bar{A}$. Therefore, the denominator and the numerator of $u_j^{\#}$ above can be bounded as    
\begin{equation*}
    \left[\bar{R}(t)+\bar{A}\rho_{ij}\right] \geq \left[R_j(t)+ A_j\rho_{ij}\right] \mbox{ and } (B_j-\bar{A}) \leq (B_j-A_j). 
\end{equation*}
The above result gives (also using $\bar{u}_j,u_j^B$ definitions in \eqref{Eq:Constraints1}),
\begin{align}\label{Eq:Lm_OP3_FormalPart1_PfS3}
    u_j^{\#} =  \frac{\bar{R}(t)+\bar{A}\rho_{ij}}{B_j-\bar{A}} 
    \geq \frac{R_j(t)+A_j\rho_{ij}}{B_j-A_j} = u_j^B \geq \bar{u}_j.  
\end{align}
Therefore, $u_j^{\#} \geq \bar{u}_j$ and hence the condition  $2u_j^{\#} < \bar{u}_j$ in \eqref{Eq:Lm_OP3_FormalPart1_PfS1} does not hold. Thus, even when $C_1<0$, $u_j^*=0$. This completes the proof.
\end{proof}

\paragraph*{\textbf{- Case 2}} $(u_j,v_j)\in \mathbb{U}_2 = \{u_{j}=u_{j}^{B},\ 0\leq v_{j}\leq\bar{v}_{j}\}$. Then, $u_j = u_j^* = u_j^B$ and $0 \leq v_j \leq \bar{v}_j$. Therefore, \eqref{Eq:OP3_FormalOld} takes the form: 
\begin{equation}\label{Eq:OP3_FormalOldPart2}
    \begin{aligned}
    v_j^* =\  &\underset{v_j}{\arg\min}\ {J_H(u_j^B,v_j)}\\
    &0 \leq v_j \leq \bar{v}_j.
    \end{aligned}
\end{equation}

\begin{lemma}\label{Lm:OP3_FormalOldPart2}
The optimal solution for \eqref{Eq:OP3_FormalOldPart2} is
\begin{equation}
    v_j^* = 0.
\end{equation}
\end{lemma}
\begin{proof}
Substituting $u_j=u_j^B$ in \eqref{Eq:OP3ObjectiveSimplified} gives $J_H(u_j^B,v_j)$ as  
\begin{equation}\label{Eq:Lm_OP3_FormalPart1_PfS1}
    J_H(u_j^B,v_j) 
= C_2v_j^2 + \left[C_3u_j^B+C_5\right]v_j + \left[C_1(u_j^B)^2 + C_4u_j^B + C_6\right].
\end{equation} 
Recall $C_2,C_3,C_5,u_j^B \geq 0 $ and $C_2 = \frac{\bar{A}_j}{2}=\frac{1}{2}\sum_{m\in \bar{\mathcal{N}}_i\backslash \{j\}}A_m$.

If $C_2=0$, the objective $J_H(u_j^B,v_j)$ is linear in $v_j$. Also its gradient is non-negative. Therefore, when $C_2 = 0$, clearly $v_j^*=0$.

Consider the case where $C_2>0$. Then $J_H(u_j^B,v_j)$ has its unconstrained optimum is at $v_j = v_j^{\#}$ where
$$v_j^{\#} = \frac{-\left[C_3u_j^B+C_5\right]}{2C_2},$$ (using calculus). Also note that due to the quadratic nature of $J_H(u_j^B,v_j)$, it should also be symmetric around $v_j=v_j^{\#}$. 

Since $C_2>0$, $J_H(u_j^B,v_j)$ is convex. Also, $v_j^{\#} \leq 0$ as $C_3,C_5,u_j^B \geq 0$. This implies that the constrained optimum is at $v_j^*=0$ even when $C_2>0$. This completes the proof.
\end{proof}

\paragraph*{\textbf{- Combined Result}}
\begin{theorem}\label{Eq:OP3CombinedResultOld}
The optimal solution of \eqref{Eq:OP3_FormalOld} is $u_j^*=0, \ v_j^*=0$, and the optimal cost is $J_H(u_j^*,v_j^*)=C_6$.
\end{theorem}
\begin{proof}
First, assume that the optimal solution of \eqref{Eq:OP3_FormalOld} $(u_j^*,v_j^*)$ belongs to $\mathbb{U}_1$ in \eqref{Eq:Constraints1}. Then, Lemma \ref{Lm:OP3_FormalOldPart1} gives that $u_j^*=0,\ v_j^*=0$. The corresponding objective function value (using \eqref{Eq:Lm_OP3_FormalPart1_PfS0}) is 
$$ \left[J_H(u_j^*,v_j^*)\right]_{Class1} = C_6.$$

However, if the optimal solution of \eqref{Eq:OP3_FormalOld} is assumed to be in $\mathbb{U}_2$ of \eqref{Eq:Constraints1}, Lemma \ref{Lm:OP3_FormalOldPart2} gives that  $u_j^*=u_j^B,\ v_j^*=0$. The corresponding objective function value (using \eqref{Eq:Lm_OP3_FormalPart1_PfS1}) is 
$$ \left[J_H(u_j^*,v_j^*)\right]_{Class2} = \left[C_1(u_j^B)^2 + C_4u_j^B + C_6\right].$$

If the latter case (i.e., $(u_j,v_j) \in \mathbb{U}_2$) provides a better performing solution \begin{align}\nonumber
    C_1(u_j^B)^2 + C_4u_j^B + C_6 \leq C_6.
\end{align}
Using $u_j^B\geq 0$, above condition can be simplified into: $C_1u_j^B + C_4 \leq 0 $, which is only possible when $C_1<0$ as $C_4 \geq 0$. Therefore, this condition can be simplified as: $C_1<0$ and, 
$$\frac{-C_4}{C_1} \leq u_j^B.$$ 
Using \eqref{Eq:Lm_OP3_FormalPart1_PfS2}, the above condition can be written as $2u_j^{\#} \leq \bar{u}_j$. however, \eqref{Eq:Lm_OP3_FormalPart1_PfS3} shows that whenever $C_1\leq 0$, $u_j^{\#} \geq \bar{u}_j$. Thus, clearly the condition $2u_j^{\#} \leq \bar{u}_j$ does not hold (i.e., a contradiction). 

Therefore, the optimal solution of \eqref{Eq:OP3_FormalOld} should belong to $\mathbb{U}_1$ and hence $u_j^* = 0,\ v_j^*=0$ and $J_H(u_j^*,v_j^*)=C_6$. 
\end{proof}

As a result of the above theorem, when the agent $a$ is ready to leave target $i$ at time $t$, it can compute the optimal trajectory costs $J_H(u_j^*,v_j^*)$ for all $j\in \mathcal{N}_i$ by simply using the expression for $C_6$ where 
\begin{equation}\label{Eq:Op3OptimumCostF}
    J_H(u_j^*,v_j^*) = C_6 = \frac{\rho_{ij}}{2}\left[2\bar{R}(t) + \bar{A}\rho_{ij}\right].
\end{equation}

\paragraph{\textbf{Solving for Optimal Next Destination $j^*$}}
The second step of the \textbf{RHCP3} (i.e., \eqref{Eq:RHCGenSolStep2}) is to choose the optimum neighbor $j$ according to 
\begin{equation}\label{Eq:OP3_FormalOldStep2}
j^* = \underset{j\in\mathcal{N}_i}{\arg\min}\ J_H(u_j^*,v_j^*).    
\end{equation}
As shown in \eqref{Fig:OneVisitTimeline}, above $j^*$ defines the \textquotedblleft Action\textquotedblright\ that the agent has to take at $t=t$.

\begin{theorem}\label{Th:OP3jStarOld}
The optimal solution to \eqref{Eq:OP3_FormalOldStep2} is the neighbor $j=j^*\in \mathcal{N}_i$ whom can be reached in a shortest time, i.e.,  
$$j^* = \arg\min_{j\in \mathcal{N}_i}\ \rho_{ij}.$$
\end{theorem}

\begin{proof}
The objective function of the discrete optimization problem  \eqref{Eq:OP3_FormalOldStep2} is \eqref{Eq:Op3OptimumCostF}. Therefore, 
$$j^* = \underset{j\in\mathcal{N}_i}{\arg\min}\ \frac{\rho_{ij}}{2}\left[2\bar{R}(t) + \bar{A}\rho_{ij}\right].$$
Note that $\bar{R}(t)$ and $\bar{A}$ terms are independent of $j$ (see \eqref{Eq:OP3ObjectiveSimplified}). Therefore, the above objective function (i.e., $C_6$) can be seen as a quadratic function of $\rho_{ij}$. Also, it is convex and its poles are located at $\rho_{ij}=0$ and $\rho_{ij}=-\frac{2\bar{R}(t)}{\bar{A}}\leq 0$. Thus, $C_6$ monotonically increases with $\rho_{ij}$. As a result, $j^*$ is the neighbor $j$ with the smallest $\rho_{ij}$ value.
\end{proof}

Above theorem implies that it is optimal to choose the next destination target only based on the (shortest) transit time. This is clearly unfavorable as an agent could converge to oscillate between two targets in the target topology while ignoring others. Hence the importance of the denominator $w$ term included in the RHCP objective function definition \eqref{Eq:RHCNewChoices} is evident.


\addtolength{\textheight}{-12cm}








\bibliographystyle{IEEEtran}
\bibliography{References}

\begin{thebibliography}{10}
\providecommand{\url}[1]{#1}
\csname url@samestyle\endcsname
\providecommand{\newblock}{\relax}
\providecommand{\bibinfo}[2]{#2}
\providecommand{\BIBentrySTDinterwordspacing}{\spaceskip=0pt\relax}
\providecommand{\BIBentryALTinterwordstretchfactor}{4}
\providecommand{\BIBentryALTinterwordspacing}{\spaceskip=\fontdimen2\font plus
\BIBentryALTinterwordstretchfactor\fontdimen3\font minus
  \fontdimen4\font\relax}
\providecommand{\BIBforeignlanguage}[2]{{%
\expandafter\ifx\csname l@#1\endcsname\relax
\typeout{** WARNING: IEEEtran.bst: No hyphenation pattern has been}%
\typeout{** loaded for the language `#1'. Using the pattern for}%
\typeout{** the default language instead.}%
\else
\language=\csname l@#1\endcsname
\fi
#2}}
\providecommand{\BIBdecl}{\relax}
\BIBdecl

\bibitem{Lin2013}
X.~Lin and C.~G. Cassandras, ``{An Optimal Control Approach to The Multi-Agent
  Persistent Monitoring Problem in Two-Dimensional Spaces},'' \emph{IEEE Trans.
  on Automatic Control}, vol.~60, no.~6, pp. 1659--1664, 2015.

\bibitem{Huynh2010}
V.~A. Huynh, J.~J. Enright, and E.~Frazzoli, ``{Persistent Patrol with
  Limited-Range On-Board Sensors},'' in \emph{Proc. of 49th IEEE Conf. on
  Decision and Control}, 2010, pp. 7661--7668.

\bibitem{Hari2019}
S.~K. Hari, S.~Rathinam, S.~Darbha, K.~Kalyanam, S.~G. Manyam, and D.~Casbeer,
  ``{The Generalized Persistent Monitoring Problem},'' in \emph{Proc. of
  American Control Conf.}, vol. 2019-July, 2019, pp. 2783--2788.

\bibitem{Zhou2019}
N.~Zhou, C.~G. Cassandras, X.~Yu, and S.~B. Andersson, ``{Optimal
  Threshold-Based Distributed Control Policies for Persistent Monitoring on
  Graphs},'' in \emph{Proc. of American Control Conf.}, 2019, pp. 2030--2035.

\bibitem{Yu2016}
J.~Yu, M.~Schwager, and D.~Rus, ``{Correlated Orienteering Problem and its
  Application to Persistent Monitoring Tasks},'' \emph{IEEE Trans. on
  Robotics}, vol.~32, no.~5, pp. 1106--1118, 2016.

\bibitem{Welikala2019P3}
S.~Welikala and C.~G. Cassandras, ``{Asymptotic Analysis for Greedy
  Initialization of Threshold-Based Distributed Optimization of Persistent
  Monitoring on Graphs},'' in \emph{Proc. of 21st IFAC World Congress}, 2020.

\bibitem{Rezazadeh2019}
N.~Rezazadeh and S.~S. Kia, ``{A Sub-Modular Receding Horizon Approach to
  Persistent Monitoring for A Group of Mobile Agents Over an Urban Area},'' in
  \emph{IFAC-PapersOnLine}, vol.~52, no.~20, 2019, pp. 217--222.

\bibitem{Zhou2018}
N.~Zhou, X.~Yu, S.~B. Andersson, and C.~G. Cassandras, ``{Optimal Event-Driven
  Multi-Agent Persistent Monitoring of a Finite Set of Data Sources},''
  \emph{IEEE Trans. on Automatic Control}, vol.~63, no.~12, pp. 4204--4217,
  2018.

\bibitem{Lan2013}
X.~Lan and M.~Schwager, ``{Planning Periodic Persistent Monitoring Trajectories
  for Sensing Robots in Gaussian Random Fields},'' in \emph{Proc. of IEEE Intl.
  Conf. on Robotics and Automation}, 2013, pp. 2415--2420.

\bibitem{Khazaeni2018}
Y.~Khazaeni and C.~G. Cassandras, ``{Event-Driven Cooperative Receding Horizon
  Control for Multi-Agent Systems in Uncertain Environments},'' \emph{IEEE
  Trans. on Control of Network Systems}, vol.~5, no.~1, pp. 409--422, 2018.

\bibitem{Song2014}
C.~Song, L.~Liu, G.~Feng, and S.~Xu, ``{Optimal Control for Multi-Agent
  Persistent Monitoring},'' \emph{Automatica}, vol.~50, no.~6, pp. 1663--1668,
  2014.

\bibitem{Leahy2016}
K.~Leahy, D.~Zhou, C.~I. Vasile, K.~Oikonomopoulos, M.~Schwager, and C.~Belta,
  ``{Persistent Surveillance for Unmanned Aerial Vehicles Subject to Charging
  and Temporal Logic Constraints},'' \emph{Autonomous Robots}, vol.~40, no.~8,
  pp. 1363--1378, 2016.

\bibitem{Trevathan2018}
J.~Trevathan and R.~Johnstone, ``{Smart Environmental Monitoring and Assessment
  Technologies (SEMAT)—A New Paradigm for Low-Cost, Remote Aquatic
  Environmental Monitoring},'' \emph{Sensors (Switzerland)}, vol.~18, no.~7,
  2018.

\bibitem{Yu2015}
J.~Yu, S.~Karaman, and D.~Rus, ``{Persistent Monitoring of Events With
  Stochastic Arrivals at Multiple Stations},'' \emph{IEEE Trans. on Robotics},
  vol.~31, no.~3, pp. 521--535, 2015.

\bibitem{khazaeni2018b}
Y.~Khazaeni and C.~G. Cassandras, ``{Event-Driven Trajectory Optimization for
  Data Harvesting in Multi-Agent Systems},'' \emph{IEEE Trans. on Control of
  Network Systems}, vol.~5, no.~3, pp. 1335--1348, 2018.

\bibitem{Smith2011}
S.~L. Smith, M.~Schwager, and D.~Rus, ``{Persistent Monitoring of Changing
  Environments Using a Robot with Limited Range Sensing},'' in \emph{Proc. of
  IEEE Intl. Conf. on Robotics and Automation}, 2011, pp. 5448--5455.

\bibitem{Maini2018}
P.~Maini, K.~Yu, P.~B. Sujit, and P.~Tokekar, ``{Persistent Monitoring with
  Refueling on a Terrain Using a Team of Aerial and Ground Robots},'' in
  \emph{Proc. of IEEE Intl. Conf. on Intelligent Robots and Systems}, 2018, pp.
  8493--8498.

\bibitem{Mathew2015}
N.~Mathew, S.~L. Smith, and S.~L. Waslander, ``{Multirobot Rendezvous Planning
  for Recharging in Persistent Tasks},'' \emph{IEEE Trans. on Robotics},
  vol.~31, no.~1, pp. 128--142, 2015.

\bibitem{Pinto2020}
S.~C. Pinto, S.~B. Andersson, J.~M. Hendrickx, and {Christos G. Cassandras},
  ``{Multi-Agent Infinite Horizon Persistent Monitoring of Targets with
  Uncertain States in Multi-Dimensional Environments},'' in \emph{Proc. of 21st
  IFAC World Congress (to appear)}, 2020.

\bibitem{Welikala2020Ax6}
\BIBentryALTinterwordspacing
S.~Welikala and C.~G. Cassandras, ``{Event-Driven Receding Horizon Control for
  Distributed Estimation in Network Systems},'' 2020. [Online]. Available:
  \url{https://arxiv.org/abs/2009.11958}
\BIBentrySTDinterwordspacing

\bibitem{Yu2020}
X.~Yu, S.~B. Andersson, N.~Zhou, and C.~G. Cassandras, ``{Scheduling Multiple
  Agents in a Persistent Monitoring Task Using Reachability Analysis},''
  \emph{IEEE Trans. on Automatic Control}, vol.~65, no.~4, pp. 1499--1513,
  2020.

\bibitem{Sun2020}
C.~Sun, S.~Welikala, and C.~G. Cassandras, ``{Optimal Composition of
  Heterogeneous Multi-Agent Teams for Coverage Problems with Performance Bound
  Guarantees},'' \emph{Automatica}, vol. 117, p. 108961, 2020.

\bibitem{Cassandras2010b}
C.~G. Cassandras, Y.~Wardi, C.~G. Panayiotou, and C.~Yao, ``{Perturbation
  Analysis and Optimization of Stochastic Hybrid Systems},'' \emph{European
  Journal of Control}, vol.~16, no.~6, pp. 642--661, 2010.

\bibitem{Welikala2020P5}
S.~Welikala and C.~G. Cassandras, ``{Event-Driven Receding Horizon Control For
  Distributed Persistent Monitoring on Graphs},'' in \emph{Proc. of 59th IEEE
  Conf. on Decision and Control (to appear)}, 2020.

\bibitem{Li2006}
W.~Li and C.~G. Cassandras, ``{A Cooperative Receding Horizon Controller for
  Multi-Vehicle Uncertain Environments},'' \emph{IEEE Trans. on Automatic
  Control}, vol.~51, no.~2, pp. 242--257, 2006.

\bibitem{Wang2017}
Y.-W. Wang, Y.-W. Wei, X.-K. Liu, N.~Zhou, and C.~G. Cassandras, ``{Optimal
  Persistent Monitoring Using Second-Order Agents with Physical Constraints},''
  \emph{IEEE Trans. on Automatic Control}, vol.~64, no.~8, pp. 3239--3252,
  2017.

\bibitem{Chen2019}
R.~Chen and C.~G. Cassandras, ``{Optimization of Ride Sharing Systems Using
  Event-driven Receding Horizon Control},'' in \emph{Proc. of 2020 Intl.
  Workshop on Discrete Event Systems (to appear)}, 2020.

\bibitem{Bertsekas2016nonlinear}
D.~P. Bertsekas, \emph{{Nonlinear Programming}}.\hskip 1em plus 0.5em minus
  0.4em\relax Athena Scientific, 2016.

\bibitem{Rosenberg}
\BIBentryALTinterwordspacing
S.~Rosenberg, ``{Conics}.'' [Online]. Available:
  \url{http://math.bu.edu/people/sr/GandS/handouts/Chapter6.pdf}
\BIBentrySTDinterwordspacing

\bibitem{Auckly2007}
D.~Auckly, ``{Solving the Quartic with a Pencil},'' \emph{American Mathematical
  Monthly}, vol. 114, no.~1, pp. 29--39, 2007.

\end{thebibliography}

\end{document}